\documentclass[11pt]{article}
\usepackage{algorithm} 
\usepackage{algorithmic}  
\usepackage[algo2e]{algorithm2e} 
\RestyleAlgo{ruled}
\usepackage{amsthm,amsmath,amssymb}
\usepackage{natbib}
\usepackage{multirow}
\usepackage[pdftex]{graphicx}
\usepackage{subfigure}
\usepackage{makecell}
\usepackage{booktabs}
\usepackage{array}
\usepackage{fullpage}
\usepackage{url} 
\usepackage{bm}
\usepackage{smile}
\usepackage{mathtools}
\usepackage{wrapfig}
\usepackage{lipsum}
\usepackage{mathrsfs}
\usepackage{dsfont}
\usepackage{titling}
\usepackage{epstopdf}
\usepackage{bbm}
\usepackage{relsize}
\usepackage{smile}
\usepackage{natbib}
\usepackage{multirow}
\usepackage{color}
\usepackage{mathtools}
\usepackage{caption}

\usepackage{natbib}
\usepackage{multirow}
\usepackage{amsfonts}

\usepackage[usenames,dvipsnames,svgnames,table]{xcolor}
\usepackage[colorlinks=true,
            linkcolor=blue,
            urlcolor=blue,
            citecolor=blue]{hyperref}

\usepackage{comment}
\usepackage{tikz}

\usepackage{mathtools}

\usepackage{enumitem}

	\makeatletter
	\newcommand{\mypm}{\mathbin{\mathpalette\@mypm\relax}}
	\newcommand{\@mypm}[2]{\ooalign{%
			\raisebox{.1\height}{$#1+$}\cr
			\smash{\raisebox{-.6\height}{$#1-$}}\cr}}
	\makeatother

	\def\beq{\begin{equation}\begin{aligned}[b]}
	
	\def\eeq{\end{aligned}\end{equation}}

\begin{document}

\title{Two-stage Hypothesis Tests for Variable Interactions with FDR Control}

\author{Jingyi Duan\thanks{Department of Statistics and Data Science, Cornell University, Ithaca, NY 14850, USA; e-mail: \texttt{jd2222@cornell.edu}.}~~~~Yang Ning\thanks{Department of Statistics and Data Science, Cornell University, Ithaca, NY 14850, USA; e-mail: \texttt{yn265@cornell.edu}.}~~~~Xi Chen\thanks{Department of Technology, Operations, and Statistics, New York University, New York, NY 10012, USA; e-mail: \texttt{xc13@stern.nyu.edu}.}~~~~Yong Chen\thanks{Department of Biostatistics, Epidemiology and Informatics, University of Pennsylvania, Philadelphia, PA 19104, USA; e-mail: \texttt{ychen123@upenn.edu}.}}
\date{\today}

\maketitle

\vspace{-0.3in}

\begin{abstract}
In many scenarios such as genome‐wide association studies where dependences between variables commonly exist, it is often of interest to infer the interaction effects in the model. However, testing pairwise interactions among millions of variables in complex and high-dimensional data suffers from low statistical power and huge computational cost.  To address these challenges, we propose a two-stage testing procedure with false discovery rate (FDR) control, which is known as a less conservative multiple‐testing correction. Theoretically, the difficulty in the FDR control dues to the data dependence among test statistics in two stages, and the fact that the number of hypothesis tests conducted in the second stage depends on the screening result in the first stage. By using the Cramér type moderate deviation technique, we show that our procedure controls FDR at the desired level asymptotically in the generalized linear model (GLM), where the model is allowed to be misspecified. In addition, the asymptotic power of the FDR control procedure is rigorously established.  We demonstrate via comprehensive simulation studies that our two-stage procedure is computationally more efficient than the classical BH procedure, with a comparable or improved statistical power. Finally, we apply the proposed method to a bladder cancer data from dbGaP where the scientific goal is to identify
genetic susceptibility loci for bladder cancer.

\end{abstract}

\section{Introduction}
 
 
In modern data science, additive models with only main effects are often insufficient to characterize the association between covariates and responses, where the effect of one variable may be dependent upon the value of the others. A natural approach to deal with this challenge is to include interaction effects, which are often treated as parameters of interest in the analysis. 
 For example, in human genetics, gene-environment interactions and gene–gene interactions attract increasing attention, as the single nucleotide polymorphisms (SNPs) discovered so far can only explain a small portion of complex disease heritability \citep{manolio2009finding}. Moreover, researchers also found that the distribution of disease among populations is often caused by the  interactions between many susceptibility genes and environmental exposures \citep{sing2004dynamic}.


Modeling and estimating linear or even nonlinear interaction effects form an important topic in statistics \citep{shujie2015estimation, ma2015semiparametric, li2014fast, liu2016partial, fan2019nonlinear,zhou2019evaluating}. With high-dimensional data, one strand of research is to fit a high-dimensional regression model with main effects and all possible pairwise interactions. For example, \cite{bien2013lasso} proposed to estimate the unknown parameters using lasso under a set of additional convex constraints, corresponding to the hierarchy principle for interactions; see also \cite{yan2017hierarchical} and the references therein. Similarly, \cite{zhao2016analysis} developed a group lasso approach  to jointly estimate the main effects and all possible pairwise interactions. An alternative approach based on the regularized principal Hessian matrix is proposed by \cite{tang2020high}, which directly estimates the interaction parameters. To reduce the computational cost in the lasso based approach,  performing feature selection via an initial screening step has been proposed and developed in a sequence of works \citep{hao2014interaction,fan2015innovated,li2021high,tian2021rase}. While this class of methods enjoys many desired theoretical results (e.g., estimation and variable selection consistency), the lasso/group lasso based methods may become non-practical with very high-dimensional features (e.g., in genetics), and the methods with variable screening usually do not  provide any inferential results, such as p-values, for the interaction parameters.   



 
To tackle with these difficulties, several two-stage multiple testing procedures have been proposed in genetics literature. In the first stage, a screening step based on a variety of test statistics is applied, which is similar to the screening step in the variable screening literature. The variables that pass the first stage are further examined for the interaction effects. For example, to test for the gene-gene interactions, \cite{kooperberg2008increasing} proposed to test the marginal effect of single genetic variant in the screening stage. \cite{murcray2009gene} and \cite{gauderman2010efficient} further generalized the method to test gene-environment interactions under case-control and case-parent trio studies. In the statistics literature,  \cite{dai2012two} is the first one that rigorously investigated the statistical properties of such two-stage testing procedures under a variety of settings, including generalized linear models (GLMs), Cox models, and case-control study. In particular, they proposed a novel two-stage method to control the family-wise error rate (FWER). To the best of our knowledge, all the existing two-stage testing procedures are tailored to the control of FWER. It is well known that the control of FWER in multiple testing problems tends to be very conservative and may suffer from low statistical power. The false discovery rate (FDR) control has been commonly used in practice to enhance the power of the testing procedures \citep{benjamini1995controlling}. However, how to control FDR in the two-stage procedures remains an open problem. 

The literature on the FDR control is vast. To name a few examples, the classical BH procedure provides a valid FDR control when the $p$-values are independent \citep{benjamini1995controlling}, and is generalized to handle the positive regression dependency among $p$-values \citep{benjamini2001control}. When the dependence among test statistics is sparse and weak, \cite{liu2013gaussian} proposed to apply the Cramér type moderate deviation technique to establish the FDR control in Gaussian graphical models; see also \cite{xia2019matrix,ye2021paired}. However, the existing methods and technique in the aforementioned work cannot be directly applied to our two-stage testing problem due to the following two reasons. First, the same dataset is used to construct hypothesis tests in both two stages, which implies the dependence among $p$-values. Second, the number of hypotheses conducted in the second stage depends on the output from the first stage and therefore is data dependent. Theoretically, quantifying the effect of dependence structure of $p$-values on the FDR procedure, and handling the extra randomness from the tests in the screening stage are the main challenges.

In our article, we propose a FDR control procedure for the two-stage testing problem in the GLM context. To make the framework more flexible, we allow the GLM to be misspecified. The two-stage testing framework is similar to \cite{dai2012two}, where the first stage is used to screen out the variables with weak marginal effects and the hypothesis tests for interactions are conducted for the remaining variables in the second stage. The main novelty of this work is how to control FDR at the desired nominal level and justify the validity of the proposed FDR procedure. 
In particular, the dependence among the test statistics in both two stages as well as the randomness from the tests in stage 1 are carefully considered in our method. In addition, since our FDR control procedure relies on the asymptotic normality of the test statistics, we also need to conduct a more refined analysis to quantify the convergence rate of the Gaussian approximation. Finally, using the Cramér type moderate deviation technique, we show that our procedure controls the false discovery proportion (FDP) and therefore FDR asymptotically. The asymptotic power of the FDR control procedure is also rigorously established. One interesting result is that, to attain the optimal power, the proposed method may require a more relaxed signal strength condition than the BH procedure, which is applied to test all possible pairwise interactions. In addition to the theoretical guarantees, our numerical results show that the proposed method can control FDR at the desired level, and is computationally more efficient than the BH procedure, without suffering from much loss of power. 





The rest of this paper is organized as follows. The two-stage testing and FDR control procedures are proposed in Section \ref{methodology}. The theoretical guarantees are provided in Section \ref{theory}. The simulation and real data applications are considered in Sections \ref{simulation} and \ref{data}, respectively. The technical details and proofs are deferred to the Appendix.


\section{Methodology}\label{methodology}

\subsection{Problem Setup}
Assume that we observe $n$ i.i.d. samples $\left(X_{1}, Y_{1}\right), \ldots,\left(X_n, Y_{n}\right)$, where $X_i=(X_{i1},...,X_{ip}) \in \mathbb{R}^{p}$ is a $p$ dimensional covariate vector and $Y_i \in \mathbb{R}$ is the response variable. In this paper, we allow $p$ to be much larger than the sample size $n$. In the high-dimensional setting, directly modeling the conditional distribution of $Y_i$ given $X_i$ can be difficult, in the presence of nonlinear effect and possibly pairwise (and even multi-way) interactions of $X_i$. Even if such a model for the conditional distribution can be successfully developed, the model would typically include an extremely large number of unknown parameters, which is often difficult to estimate in practice. 

When the goal is to infer the interaction between two variables $X_{ij}$ and $X_{ik}$, instead of fitting a high-dimensional model for the conditional distribution of $Y_i$ given $X_i$, many applied researchers often simply regress $Y_i$ on $X_{ij}$, $X_{ik}$ and their interaction using some working parametric models. Such an approach has been widely used in genome-wide association study (GWAS) to investigate the gene-gene interactions and gene-environment interactions. 
In this paper, we focus on the generalized linear model (GLM) with interactions. Formally, under the GLM, the density function of $Y_i$ given the two variables $X_{ij}$ and $X_{ik}$ is
\begin{equation}\label{glm0}
    f_{\beta^{jk}}(y|x_j, x_k) = \exp \left\{\frac{y \theta_{jk} -b\left(\theta_{jk}\right)}{\phi}+c(y,\phi)\right\},
\end{equation}
where $b(\cdot)$ and $c(\cdot)$ are known functions. Under the canonical link, we have 
$$
\theta_{jk}=\beta_{0jk} + \beta_{1jk}x_{j} + \beta_{2jk}x_{k} + \beta_{3jk}x_{j}x_{k}.
$$
Denote the parameter $\beta^{jk}=(\beta_{0jk},\beta_{1jk},\beta_{2jk},\beta_{3jk})^T\in\RR^4$. We further assume the dispersion parameter $\phi$ is known. The framework can be easily extended to deal with unknown dispersion parameters. For example, we consider linear regression with unknown noise variance in our simulation studies. For notational simplicity, we set $\phi=1$ throughout the paper. While a GLM is assumed in (\ref{glm0}), we do not assume it is correctly specified. Let us denote the true conditional density of $Y_i$ given $X_i$ by $p(y|x)$, which may not follow the GLM in (\ref{glm0}). The least false value of $\beta^{jk}$ is defined as the one that minimizes the KL-divergence between $p(y|x)$ and $f_{\beta^{jk}}(y|x_j, x_k)$,
\begin{equation}\label{least_false}
\beta^{jk}_0 = \argmin\mathrm{E}\Big[\log \frac{p(y|x)}{f_{\beta^{jk}}(y|x_j, x_k)}\Big],
\end{equation}
where the expectation is evaluated at the true conditional distribution $p(y|x)$. 

Under the misspecified GLM (\ref{glm0}), we are interested in testing the existence of the interaction effect, i.e.,  
\begin{equation}\label{h0}
H_{0jk}: \beta_{3jk}=0, ~~\textrm{versus}~~H_{1jk}: \beta_{3jk}\neq 0,
\end{equation}
where $1\leq j < k \leq p$. To account for the multiplicity of the hypothesis tests, we aim to control FDR in this process.

\subsection{Two-stage testing procedure with FDR control}

To test the multiple hypotheses (\ref{h0}) with FDR control, one standard approach is to apply the Benjamini–Hochberg (BH) procedure. However, such a procedure is computationally expensive in application with very large $p$ (such as in GWAS), as one has to conduct hypothesis testing $p(p-1)/2$ times. To reduce the computational burden, \cite{dai2012two} proposed to test the hypotheses (\ref{h0}) using a two-stage procedure and controled the family-wise error rate (FWER). In this section, we focus on the FDR control, which is known to be less conservative than the FWER control. 

In stage 1, we first test the main effect of each variable by regressing $Y_i$ on $X_{ij}$ for $1\leq j\leq p$. That is we impose the following working GLM for $Y_i$ given $X_{ij}$ 
\begin{equation}\label{glm1}
    f_{\beta^{j}}(y|x_j) = \exp \left\{\frac{y (\beta_{0j}+\beta_{1j}x_j) -b\left(\beta_{0j}+\beta_{1j}x_j\right)}{\phi}+c(y,\phi)\right\},
\end{equation}
where $\beta^{j}=(\beta_{0j},\beta_{1j})^T$. Similarly, the above model can be misspecified, and the least false value of $\beta^{j}$ can be defined in the same way as (\ref{least_false}). In this stage, we aim to test the hypothesis 
\begin{equation}\label{h1}
H_{0j}: \beta_{1j}=0, ~~\textrm{versus}~~H_{1j}: \beta_{1j}\neq 0.
\end{equation}
Let $\hat\beta^{j}$ be the maximum likelihood estimator (MLE) under the model (\ref{glm1}). Since the model can be misspecified, the following sandwich estimate of the asymptotic variance of $\hat\beta^{j}$ is used,
\begin{align}\label{eq_var_betaj}
    \begin{split}
 \hat{\mathrm{\textbf{cov}}}(\hat{\beta}^{j}) 
 \coloneqq 
 \left[ \frac{1}{n}\sum\limits_{i=1}^nb^{\prime\prime}\left(\emph{\textbf{X}}_{ij}^{\mathrm{s1}}\hat{\beta}^{j}\right)(\emph{\textbf{X}}_{ij}^{\mathrm{s1}})^T\emph{\textbf{X}}_{ij}^{\mathrm{s1}}\right]^{-1}
\hat{\mathrm{\textbf{cov}}}(\psi_{\hat{\beta}^{j}})
  \left[ \frac{1}{n}\sum\limits_{i=1}^nb^{\prime\prime}\left(\emph{\textbf{X}}_{ij}^{\mathrm{s1}}\hat{\beta}^{j}\right)(\emph{\textbf{X}}_{ij}^{\mathrm{s1}})^T\emph{\textbf{X}}_{ij}^{\mathrm{s1}}\right]^{-1}.
    \end{split}
\end{align}
where $\bX_{ij}^{\mathrm{s1}}=(1, X_{ij})$ is the covariate vector at stage 1, and $b''(\cdot)$ is the second derivative of $b(\cdot)$. Here, $\hat{\mathrm{\textbf{cov}}}(\psi_{\beta})$ is the sample covariance matrix of the score function $\psi_{\beta}$, 
$$
\hat{\mathrm{\textbf{cov}}}(\psi_{\beta})=\frac{1}{n}\sum\limits_{i=1}^n\left\{Y_{i}-b^{\prime}\left(\bX_{ij}^{\mathrm{s1}} \beta^j\right)\right\}^2(\emph{\textbf{X}}_{ij}^{\mathrm{s1}})^T\bX_{ij}^{\mathrm{s1}}. 
$$
The Wald test statistic for $H_{0j}$ is given by 
 \begin{equation}\label{eq_tj}
 \hat{T}_j \coloneqq \frac{\sqrt{n}\hat{\beta}^{j}_{(2)}}{\sqrt{\hat{\mathrm{\textbf{cov}}}(\hat{\beta}^{j})_{(2,2)}}},
 \end{equation}
where $\hat{\beta}^{j}_{(2)}$ is the second component of $\hat\beta^{j}$ and $\hat{\mathrm{\textbf{cov}}}(\hat{\beta}^{j})_{(2,2)}$ is the $(2,2)$th component of the estimated covariance matrix $\hat{\mathrm{\textbf{cov}}}(\hat{\beta}^{j})$. We reject the null hypothesis $H_{0j}$ in (\ref{h1}) if and only if $|\hat T_j|\geq \alpha$ for some tuning parameter $\alpha>0$. Intuitively, the aim of this stage is to reduce the number of variables in the followup analysis by screening out those with relatively weak main effects. 

In stage 2, we construct the test statistics for the hypothesis of interest $H_{0jk}$ in (\ref{h0}) if and only if both $H_{0j}$ and $H_{0k}$ are rejected in stage 1. In other words, if either $X_{ij}$ or $X_{ik}$ has weak main effect, we will not test their interaction and directly accept the null hypothesis $H_{0jk}$. Throughout the paper, when $H_{0j}$ is rejected, we can equivalently say that the variable $X_{ij}$ passes the test in stage 1, which will be further considered in stage 2. The rationale of this step is inspired from the so-called hierarchy principle, that is if the model contains the interaction of $X_{ij}$ and $X_{ik}$, then both main effects should be included. In addition, if we increase the value of $\alpha$ in stage 1, there are fewer variables whose interactions need to be tested in stage 2, so that the number of tests conducted in stage 2 is significantly reduced and therefore the computation is more efficient. We emphasize that while the two-stage testing procedure follows the hierarchy principle, the theoretical justification of the FDR control shown in the next section, however, does not assume any hierarchical structure between the main effect and the interactions. 

To construct the test, let $\hat\beta^{jk}$ be the maximum likelihood estimator (MLE) under the model (\ref{glm0}). The Wald test statistic for $H_{0jk}$ is 
  \begin{equation}\label{eq_tjk}
 \hat{T}_{jk} \coloneqq \frac{\sqrt{n}\hat{\beta}^{jk}_{(4)}}{\sqrt{\hat{\mathrm{\textbf{cov}}}(\hat{\beta}^{jk})_{(4,4)}}}, 
  \end{equation}
where the estimate of the asymptotic variance $\hat{\mathrm{\textbf{cov}}}(\hat{\beta}^{jk})$ is defined in the same way as (\ref{eq_var_betaj}) with $\hat{\beta}^{j}$ replaced by $\hat{\beta}^{jk}$ and $\bX_{ij}^{\mathrm{s1}}$ replaced by the covariate vector in stage 2, $\bX_{ijk}^{\mathrm{s2}}=(1, X_{ij}, X_{ik}, X_{ij}X_{ik})$. We reject the null hypothesis $H_{0jk}$ if and only if $H_{0j}$ and $H_{0k}$ are rejected in stage 1 and $|\hat T_{jk}|\geq t$ for some $t$ to be chosen. Apparently, the error of the test depends on the choice of $t$. To control the FDR at a given level $\eta>0$, we propose the following procedure in Algorithm \ref{alg_fdr}.

 \begin{algorithm}
 	\caption{Two-stage FDR control algorithm}
 	\begin{algorithmic}
 		\REQUIRE the desired FDR level $\eta>0$, and tuning parameter $\alpha\geq 0$
\begin{enumerate}
    \item Calculate the test statistics $\hat T_j$ in (\ref{eq_tj}) for any $1\leq j\leq p$.
    \item  Calculate the test statistics $\hat T_{jk}$ in (\ref{eq_tjk}), when $|\hat T_j|\geq \alpha$ and $|\hat T_k|\geq \alpha$, where  $\alpha$ is the tuning parameter. 
    \item Calculate the cutoff point for the test statistic $\hat{T}_{jk}$ as
 \begin{equation} \label{def1}
\hat{t}=\inf \left\{
0 \leq \emph{t} \leq \sqrt{ 2\log p}: \frac{G(t)\sum_{1 \leq j < k \leq p}\mathbbm{1}\left\{|\hat{T}_{j}| \geq \alpha,|\hat{T}_{k}| \geq \alpha\right\}}{\max \left(\sum_{\{1\leq j<k\leq p: |\hat{T}_{j}| \geq \alpha,|\hat{T}_{k}| \geq \alpha\}}\mathbbm{1}\left\{|\hat{T}_{jk}| \geq t\right\},1 \right)}\leq \eta
\right\},
\end{equation} 
where $G(t)=2-2 \Phi(t)$, $\Phi(t)=P(N(0,1)\leq t)$ is the c.d.f. of a standard Gaussian distribution and $\eta$ is the desired FDR level. If $\hat{t}$ in (\ref{def1}) does not exist, then let $\hat{t}=\sqrt{2\log p}$. 
\item For $1\leq j<k\leq p$ with $|\hat{T}_{j}| \geq \alpha,|\hat{T}_{k}| \geq \alpha$, we reject $\emph{H}_{0jk}$ if $|\widehat{T}_{jk}|> \hat{t}$. 
\end{enumerate}
 	\end{algorithmic}\label{alg_fdr}
 \end{algorithm}
 
In the definition of $\hat t$, to avoid the case where no hypothesis $H_{0jk}$ is rejected, we follow the convention and take the maximum of the number of rejected hypotheses and 1.  We note that, as a special case, if we set $\alpha=0$, all the hypotheses in stage 1 are rejected, i.e., all variables pass the stage 1. In this case, $\hat t$ in (\ref{def1}) reduces to 
$$
\hat{t}=\inf \left\{
0 \leq \emph{t} \leq \sqrt{ 2\log p}: \frac{G(t)(p^2-p)/2}{\max \left(\sum_{1 \leq j < k \leq p}\mathbbm{1}\{|\hat{T}_{jk}| \geq t\},1 \right)}\leq \eta
\right\},
$$
which is the cutoff from the classical BH procedure, applied to the multiple testing problem (\ref{h0}) for all $1\leq j<k\leq p$ with p-values obtained from the limiting distributions of the test statistics $\hat{T}_{jk}$. Thus, a key difference between the proposed  Algorithm \ref{alg_fdr} and the classical BH procedure is that the number of hypothesis tests conducted in our method is data dependent. This makes the analysis of our Algorithm \ref{alg_fdr} more complicated than the classical BH procedure.

\section{Theory}\label{theory}

\subsection{Assumptions}
To approximate the test statistics $\hat T_j$ and $\hat T_{jk}$, we introduce the following notations. Recall that $\beta^j=(\beta_{0j},\beta_{1j})^T$ and $\bX_{ij}^{\mathrm{s1}}=(1, X_{ij})$. In the first stage, the score function is
$$
\Psi_{n}(\beta^j,\bX_{j}^{\mathrm{s1}}, Y)=\frac{1}{n} \sum_{i=1}^{n}\left\{Y_{i}-b^{\prime}\left(\bX_{ij}^{\mathrm{s1}}\beta^j\right)\right\} \bX_{ij}^{\mathrm{s1}}.
$$
We can approximate the test statistic $\hat T_j$ by 
\begin{equation}\label{eq_Uj}
  U_j=n^{-1/2}\sum\limits_{i=1}^nU_{ij},   
\end{equation}
 where
 \[
 U_{ij} = \frac{\left(-\left[\mathrm{E}\left(b^{\prime\prime}(\emph{\textbf{X}}_{ij}^{\mathrm{s1}}\beta_0^j)(\emph{\textbf{X}}_{ij}^{\mathrm{s1}})^{\otimes 2}\right)
 \right]^{-1}\left\{Y_{i}-b^{\prime}\left(\emph{\textbf{X}}_{ij}^{\mathrm{s1}}\beta_0^j\right)\right\} \cdot (\emph{\textbf{X}}_{ij}^{\mathrm{s1}})^T+\beta_0^j\right)_{(2)}}{\sqrt{\mathrm{\textbf{cov}}( u({\beta_0^{j}},\emph{\textbf{X}}_{j}^{\mathrm{s1}},Y))_{(2,2)}}} ,
 \]
and $\mathrm{\textbf{cov}}(u({\beta_0^{j}},\emph{\textbf{X}}_{j}^{\mathrm{s1}},Y))$ is the covariance matrix of $u({\beta_0^{j}},\emph{\textbf{X}}_{j}^{\mathrm{s1}},Y)$ given by
\[
 u(\beta^j,\bX_{j}^{\mathrm{s1}}, Y) = -\left[\mathrm{E}\left(b^{\prime\prime}(\bX_{ij}^{\mathrm{s1}}\beta^j)(\bX_{ij}^{\mathrm{s1}})^{\otimes 2}\right)\right]^{-1}n^{1/2}\Psi_{n}(\beta^j,\bX_{j}^{\mathrm{s1}}, Y).
 \]
 
Recall that $\beta^{jk}=(\beta_{0jk},\beta_{1jk},\beta_{2jk},\beta_{3jk})^T$ and denote the covariate vector in stage 2 by $\bX_{ijk}^{\mathrm{s2}}=(1, X_{ij}, X_{ik}, X_{ij}X_{ik})$. Similarly, in stage 2, we introduce the notation
\[
 U_{jk}=n^{-1/2}\sum\limits_{i=1}^nU_{ijk}
 ,
 \]
 where
 \[
  U_{ijk} = \frac{\left(-\left[\mathrm{E}\left(b^{\prime\prime}(\emph{\textbf{X}}_{ijk}^{\mathrm{s2}}\beta_0^{jk})(\emph{\textbf{X}}_{ijk}^{\mathrm{s2}})^{\otimes 2}\right)\right]^{-1}\left\{Y_{i}-b^{\prime}\left(\emph{\textbf{X}}_{ijk}^{\mathrm{s2}} \beta_0^{jk}\right)\right\} \cdot (\emph{\textbf{X}}_{ijk}^{\mathrm{s2}})^T+\beta_0^{jk}\right)_{(4)}}{\sqrt{\mathrm{\textbf{cov}}( u({\beta_0^{jk}},\emph{\textbf{X}}_{jk}^{\mathrm{s2}},Y))_{(4,4)}}},
 \]
and $u({\beta_0^{jk}},\emph{\textbf{X}}_{jk}^{\mathrm{s2}},Y)$ is defined in a similar way. 

Let $H_0$ denote the collection of null hypotheses, i.e., $H_0=\{1\leq j<k\leq p: (\beta^{jk}_0)_{(4)}=0\}$. For any $(j,k),(m,l) \in \emph{H}_0$, denote 
\begin{equation}\label{eq_tildeU}
     \tilde{U}_{ijkml} = (U_{ijk}, S_{ij}, S_{ik}, U_{iml}, S_{im}, S_{il})\in\RR^6,
\end{equation}
where $S_{ij} = U_{ij} - \mathrm{E}(U_{ij})$ is the centered version of $U_{ij}$. Finally, denote the residuals by
\[
\epsilon_{ij} \coloneqq Y_i - b^{\prime}(\emph{\textbf{X}}_{ij}^{\mathrm{s1}}\beta_0^j),~~\textrm{and}~~
\epsilon_{ijk} \coloneqq Y_i - b^{\prime}(\emph{\textbf{X}}_{ijk}^{\mathrm{s2}}\beta_0^{jk}).
\]

Throughout the paper, we use the notation $a\vee b=\max(a,b)$, and $a_{n,p}=\Omega(b_{n,p})$ if there exists a constant $C>0$ such that $\lim \inf_{n,p\rightarrow\infty} a_{n,p}/b_{n,p} \geq C$. In the paper, we consider the asymptotic regime $p, n\rightarrow\infty$.

\begin{assumption}\label{allasp}
We make the following assumptions. 
\begin{enumerate}[label=\textbf{A}\textbf{\arabic*}]

\item\label{aspA1} There exists a constant $K$ such that
\[
\max\limits_{1\leq j \leq p}\max\limits_{1 \leq i \leq n}|X_{ij}|\leq K.
\]
\item\label{aspA2} 
Suppose for some constant $r>0$, $p \leq n^r$. For some constant $\sigma_1>0$, we have
\[
\max\limits_{1\leq j \leq p}\mathrm{E}|\epsilon_{ij}|^{4\vee (2r+2+\epsilon)}\leq \sigma_1^2,
\]
where $\epsilon>0$ is an arbitrarily small constant. Assume the same condition holds for $\epsilon_{ijk}$ for any $1 \leq j <k \leq p$. 


\item\label{aspA3} 
Suppose there exist some positive constants $K_0, K_1$, $C_b$, $C_{\tilde{b}}$, such that
\begin{equation}\label{eq_A3}
 \max\limits_{1\leq j \leq p}\max\limits_{1\leq i \leq n}|\emph{\textbf{X}}_{ij}^{\mathrm{s1}}\beta_0^j| \leq K_0, \ 
  \max\limits_{1\leq j <k \leq p}\max\limits_{1\leq i \leq n}|\emph{\textbf{X}}_{ijk}^{\mathrm{s2}}\beta_0^{jk}| \leq K_0,
\end{equation}
and
\[
\min\limits_{1\leq j \leq p }\mathrm{\textbf{cov}}( u({\beta_0^{j}},\emph{\textbf{X}}_{j}^{\mathrm{s1}},Y))_{(2,2)} > K_1,~~
\min\limits_{1\leq j <k \leq p}
\mathrm{\textbf{cov}}( u({\beta_0^{jk}},\emph{\textbf{X}}_{jk}^{\mathrm{s2}},Y))_{(4,4)}
>K_1.
\]
For all $|z|\leq K^2 + K_0$, we assume the second and third derivatives of $b(\cdot)$, denoted by $b^{\prime\prime}$ and $b^{\prime\prime\prime}$, exist and satisfy
\begin{equation}\label{eq_A3_2}
 1/C_b\leq b^{\prime\prime}(z)\leq C_b,\  | b^{\prime\prime\prime}(z)|\leq C_{\tilde{b}}.
\end{equation}

\item\label{aspA4} For any $1 \leq j < k \leq p$, we assume for  some constant $\tau > 0$, 
\[
  \lambda_{\mathrm{min}}\Big(\frac{1}{n}\sum_{i=1}^n(\emph{\textbf{X}}_{ij}^{s1})^{\otimes 2}\Big) \geq \tau, ~~~ 
  \lambda_{\mathrm{min}}\Big(\frac{1}{n}\sum_{i=1}^n(\emph{\textbf{X}}_{ijk}^{s2})^{\otimes 2}\Big) \geq \tau.
\]

\item \label{aspA5} There exist some constants $0 \leq \delta <1$,$\gamma > 0$, $\kappa >0, C>0$ such that,
\[
\max\limits_{1\leq j <k \leq p}\max\limits_{1\leq m <l \leq p}|\mathrm{\textbf{cov}}(U_{ijk}, U_{iml})| \leq \delta,
\]
 \begin{equation}
 \label{asp2}
    \begin{split}
         \mathrm{Card}\Big\{\{(j,k),(m,l)\}:(j,k),(m,l)\in \emph{H}_0,\left\|\mathrm{\textbf{cov}}(\tilde{U}_{ijkml})- \mathrm{\textbf{I}}\right\|_{\infty} > C(\log p)^{-2-\gamma}\Big\}= O(p^{4-\kappa}),
    \end{split}
 \end{equation}
 where $\tilde{U}_{ijkml}$ is defined in (\ref{eq_tildeU}) and $ \mathrm{Card}(A)$ is the cardinality of a set $A$. 
 \item \label{aspA6} Given the constants $C,\gamma,\delta,\kappa$ defined in \ref{aspA5}, denote 
 $$\tilde{H}_{01}=\big\{(j,k)\in \emph{H}_0:  |\mathrm{\textbf{cov}}(U_{ij}, U_{ik})| \leq C(\log p)^{-2-\gamma}\big\}.
 $$ 
 Assume that $\kappa> 2\delta/(1+\delta)$. 
The tuning parameter $\alpha$ in stage 1 belongs to $[0,\sqrt{2\log p}]$ and satisfies  
\begin{equation}
\label{eq_G4}
        \sum\limits_{(j,k)\in \tilde{H}_{01}}G_j(\alpha)G_k(\alpha) = \Omega(p^{\xi}),
\end{equation}
for some constant $\xi$ with 
\[
\max \left\{\frac{3}{2}+\frac{\delta}{1+\delta}, 2-\frac{\kappa}{2}+\frac{\delta}{1+\delta}\right\} <\xi,
\]
where $G_j(\alpha) \coloneqq \mathrm{P}\left(|\mathcal{N}(0,1)+\sqrt{n}\mathrm{E}(U_{ij})|\geq \alpha\right)$ is the tail probability of a normal distribution with mean $\sqrt{n}\mathrm{E}(U_{ij})$ and variance 1.
\end{enumerate}
\end{assumption}

Assumption~\ref{aspA1} is often used to simplify the analysis of the likelihood function and more generally, the quasi-likelihood \citep{van2012quasi}. Assumption~\ref{aspA2} requires some finite moments of the residuals. These two assumptions together enable us to apply the Nemirovski moment inequality to derive sharp bounds for the MLE in the misspecified GLM \citep{buhlmann2011statistics}. The boundedness of $X_{ij}$ in \ref{aspA1} can be relaxed to sub-Gaussianity, if one is willing to impose the sub-Gaussian condition on the residuals in \ref{aspA2}. As shown in \ref{aspA2}, when $r>1$, we allow $p$ to be much larger than $n$ and the price to pay is that we need stronger moment assumptions on the residuals. 

In Assumption \ref{aspA3}, $\max_{1\leq j \leq p}\max_{1\leq i \leq n}|\emph{\textbf{X}}_{ij}^{\mathrm{s1}}\beta_0^j|\leq C'K$ holds, if each entries of $\beta^j_0$ are bounded by a constant $C'$ in absolute value together with Assumption~\ref{aspA1}. In addition, we require the asymptotic variances of the MLEs are non-degenerate. Finally, (\ref{eq_A3_2}) on the derivatives of $b(\cdot)$ is satisfied for commonly used GLMs. Assumption \ref{aspA4} implies that the Hessian matrix in GLMs is well-conditioned. 


As for Assumption~\ref{aspA5}, we first require that the test statistics $U_{ijk}$ and $U_{iml}$ are not perfectly correlated, which is reasonable if the design matrix is not collinear. A very crucial condition in our analysis is  (\ref{asp2}). It states that for all quadruplets $\{j,k,m,l\}$, where $1 \leq j <k \leq p, 1 \leq m < l \leq p$ in $H_0$, there are not too many combinations in which (1)  $U_{ijk}$ and $U_{iml}$ have strong correlation or (2) $S_{ij}$ and $U_{iml}$ have strong correlation or (3) $S_{ij}$ and $S_{im}$ have strong correlation. Note that \cite{dai2012two} proved that $\mathrm{\textbf{cov}}(S_{ij}, U_{ijk})=0$, and thus it suffices to consider the above three cases.  Since there are  at most $[p(p-1)/2]^2=O(p^4)$ quadruplets of $\{j,k,m,l\}$ in total, this assumption simply requires that the number of quadruplets with strong correlation is of a smaller order. We expect that in many applications such as GWAS, the features are often weakly correlated such that the condition (\ref{asp2}) may hold. 
To conclude, Assumption \ref{aspA5} imposes sparsity constraints on the correlation of the test statistics, and is the key technical condition to apply the Cramér type moderate deviation technique in the analysis.
 
Assumption \ref{aspA6} characterizes the interplay among the choice of $\alpha$, the size of $\tilde{H}_{01}$ and the signal strength of the test in stage 1 that is $\mathrm{E}(U_{ij})$. To see this, when $U_{ij}$ and $U_{ik}$ have weak correlation, we can show that $\sum_{(j,k)\in \tilde{H}_{01}}G_j(\alpha)G_k(\alpha)$ is the approximation of $\mathrm{E}(\sum_{(j,k)\in \tilde{H}_{01}}\mathbbm{1} \{|\widehat{T}_{j}| \geq \alpha,|\widehat{T}_{k}| \geq \alpha\})$, i.e., the expected number of pairs within the set $\tilde{H}_{01}$ that pass the test in stage 1, whose value depends on the choice of $\alpha$, $\tilde{H}_{01}$ and also the limiting distribution of $(\hat T_j,\hat T_k)$. This assumption essentially requires that $\alpha$ cannot be too large. Otherwise, the number of hypothesis tests considered in stage 2 could be very small, such that the FDR algorithm can only identify very few (or even no) signals. In practice, we expect that a large number of variables may pass the test in stage 1, which makes this assumption reasonable. In theory, under mild conditions on the size of $\tilde{H}_{01}$ and the signal strength $\mathrm{E}(U_{ij})$, we can show that (\ref{eq_G4}) holds for $\alpha = \sqrt{\alpha_1\log p}$ with some small constant $\alpha_1>0$. The detailed derivation is deferred to Appendix \ref{app_assumption}. 


\subsection{Theoretical guarantees on FDR control}
Recall that to control FDR, we propose to use $\hat t$ defined in (\ref{def1}) in Algorithm \ref{alg_fdr} as the cutoff point for the test statistics in stage 2. Our main theorem in this section shows that the proposed procedure in Algorithm \ref{alg_fdr} can control the false discovery proportion (FDP) and also FDR asymptotically. For our two-stage algorithm, we formally define FDP and FDR as
\begin{equation}\label{eq_FDP}
\mathrm{FDP}=\frac{\sum_{\{(j,k)\in \emph{H}_0: |\widehat{T}_{j}| \geq \alpha,|\widehat{T}_{k}| \geq \alpha\}}\mathbbm{1}\left\{|\widehat{T}_{jk}| \geq \hat{t}\right\}}{\max \left(\sum_{\{1 \leq j < k \leq p: |\widehat{T}_{j}| \geq \alpha,|\widehat{T}_{k}| \geq \alpha\}}\mathbbm{1}\left\{|\widehat{T}_{jk}| \geq \hat{t}\right\},1 \right)},~~~\mathrm{FDR}=\mathrm{E}(\mathrm{FDP}),
\end{equation}
where the numerator of the FDP corresponds to the total number of rejected null hypotheses and the denominator is the total number of rejected hypotheses. 

\begin{theorem}[FDP Control]\label{them1}
Under Assumption~\ref{allasp}, we have
 \begin{equation}\label{eq_them1_0}
 \frac{\mathrm{FDP}}{\eta N/M}\rightarrow 1,
 \end{equation}
in probability  as  $(n,p) \rightarrow \infty$, where $\eta$ is the desired FDR level,
\begin{equation}\label{eq_them1_1}
N=\sum_{(j,k)\in\emph{H}_0} \mathbbm{1} \left\{|\widehat{T}_{j}| \geq \alpha,|\widehat{T}_{k}| \geq \alpha \right\}, ~~\textrm{and}~~M =\sum_{1 \leq j < k \leq p}\mathbbm{1}\left\{|\hat{T}_{j}| \geq \alpha,|\hat{T}_{k}| \geq \alpha\right\}.
\end{equation}
\end{theorem}

From (\ref{eq_them1_0}), we can see that the FDP of the two-stage method can be approximated by a random variable $\eta N/M$, where $N$ denotes the number of pairs in $H_0$ that pass the tests in stage 1 and $M$ is the total number of tests conducted in stage 2. For the standard BH procedure (i.e., taking $\alpha=0$), $N$ and $M$ reduce to $|H_0|$ and $p(p-1)/2$ respectively, which are deterministic, and our Theorem \ref{them1} is consistent with the existing theoretical results on FDR control, see \cite{liu2013gaussian}. 

In addition, since $N\leq M$, (\ref{eq_them1_0}) implies $\mathrm{FDP}\leq \eta$ with probability tending to 1. As $\mathrm{FDP}\leq 1$, we can further show that the FDR is controlled at the desired level, i.e., $\mathrm{FDR}\leq \eta$. In particular, we note the FDR control is valid for a wide range of the tuning parameter $\alpha$ as long as it satisfies (\ref{eq_G4}). We also note that, while the two-stage method is inspired by the hierarchy principle for interactions, we do not need this assumption in  Theorem \ref{them1}.

The main technical tool used in the proof of Theorem \ref{them1} is the Cramér type moderate deviation bound. Unlike the technique originally introduced by \cite{liu2013gaussian}, to deal with the extra randomness induced by the tests in stage 1, we establish the Cramér type moderate deviation bound for the random vector $n^{-1}\sum_{i=1}^n \tilde U_{ijkml}$, where $\tilde U_{ijkml}$ defined in (\ref{eq_tildeU}) consists of the pairs of the test statistics in both stage 1 and stage 2. Such result can be of independent interest. Another technical challenge is to characterize the difference between the test statistic $\hat T_j$ and its linear representation $U_j$ in (\ref{eq_Uj}). Even though the MLE $\hat\beta^j$ is asymptotically linear, once we account for the uncertainty in the estimation of the asymptotic variance of $\hat\beta^j$, we can only have the following result, $\max_{1\leq j\leq p}|\hat T_j-U_j|=O_p(\frac{\log p}{\sqrt{n}}+\max_j|(\beta_0^j)_{(2)}|\sqrt{\log p})$, where the term $O_p(|(\beta_0^j)_{(2)}|\sqrt{\log p})$ comes from the estimation error of the asymptotic variance multiplied by the expectation of $U_j$. Thus, the difference between the test statistic $\hat T_j$ and its linear representation $U_j$ may not converge to 0, when $\max_j|(\beta_0^j)_{(2)}|\gtrsim 1/\sqrt{\log p}$. Since we do not assume any type of hierarchy principle in Theorem \ref{them1}, we may expect $(\beta_0^j)_{(2)}\neq 0$ under the null hypothesis $H_{0jk}$. A more intuitive explanation of this issue is that $\hat T_j$ is no longer a pivotal statistic asymptotically when $(\beta_0^j)_{(2)}\neq 0$. To address this technical issue, our proof is based on a more refined analysis of the truncated relative error $\frac{|\hat{T}_j-U_j|}{|U_j|\vee c}$, where $c>0$ is a small constant. In particular, we show in Lemma \ref{UTlemma} in Appendix \ref{app_main} that $\max_{1\leq j \leq p}\frac{|\hat{T}_j-U_j|}{|U_j|\vee c}= 
 O_p(\frac{\log p}{\sqrt{n}})$, which is a key intermediate step in the proof of Theorem \ref{them1}.

\subsection{Power Analysis}
In this subsection, we investigate the power of our FDR control procedure. Formally, we define the power of our method as
\begin{equation}
\label{power}
    \mathrm{power} = \mathrm{E}\left(\frac{\sum_{\{(j,k)\in \emph{H}_1: |\widehat{T}_{j}| \geq \alpha,|\widehat{T}_{k}| \geq \alpha\}}\mathbbm{1}\left\{|\widehat{T}_{jk}| \geq \hat{t}\right\}}{|\emph{H}_1|}\right),
\end{equation}
where $H_1=\{1\leq j<k\leq p: (\beta^{jk}_0)_{(4)}\neq 0\}$ is the collection of alternative hypotheses,  and $\hat{t}$ is defined in (\ref{def1}). We expect that the power of our method depends on two factors: (1) the number of hypotheses that are rejected in stage 1 and (2) the signal strength, i.e., the value of $|(\beta_0^{jk})_{(4)}|$, in $H_1$. To study the power of our method, we introduce the following notations and assumptions. 

Recall that in (\ref{eq_them1_1}), $M$ and $N$ denote the number of pairs in $H$ and $H_0$ that pass the tests in stage 1. Denote
\begin{equation}\label{eq_a1}
a_1=M-N=\sum_{(j,k)\in \emph{H}_1}\mathbbm{1}\{|\widehat{T}_j| \geq \alpha,|\widehat{T}_k|\geq \alpha\}.
\end{equation}
For any $(j,k) \in \emph{H}_1$, the signal strength satisfies 
\begin{equation}\label{eq_signal}
|(\beta_0^{jk})_{(4)}| \geq \delta \sqrt{\frac{\log p}{n}},~~\textrm{for some}~~ \delta > 0.
\end{equation}

\begin{assumption}\label{asp_power}
The error term $\epsilon_{ijk} = Y_i-b^{\prime}(\emph{\textbf{X}}^{s2}_{ijk}\beta_0^{jk})$ is sub-exponential with some constant $\lambda>0$, i.e.
\begin{equation}\label{eq_subexponential}
\mathrm{E}\left(e^{s \epsilon_{ijk}}\right) \leq e^{s^{2} \lambda^{2} / 2}, \quad \forall|s| \leq \frac{1}{\lambda}
\end{equation}
for any $1\leq j <k \leq p$. Denote $\Sigma_{jk}^* =[\mathrm{E}_{\beta_0^{jk}}(b^{\prime\prime}(\emph{\textbf{X}}_{ijk}^{\mathrm{s2}}\beta_0^{jk})(\emph{\textbf{X}}_{ijk}^{\mathrm{s2}})^T\emph{\textbf{X}}_{ijk}^{\mathrm{s2}})]^{-1}$.  We also assume $|(\Sigma_{jk}^*(\emph{\textbf{X}}_{ijk}^{\mathrm{s2}})^T)_{(4)}|\leq \tilde K$ for some constant $\tilde K$.
\end{assumption}

While the sub-exponential condition on $\epsilon_{ijk}$ is stronger than the moment condition in Assumption \ref{aspA2}, it is commonly used to characterize the tail behavior of the estimator in GLM. The boundedness of $|(\Sigma_{jk}^*(\emph{\textbf{X}}_{ijk}^{\mathrm{s2}})^T)_{(4)}|$ is indeed implied by Assumptions \ref{aspA1} and \ref{aspA4}. For notational simplicity, we use a new constant $\tilde K$ to denote this bound. The following theorem shows the power of the our method.

\begin{theorem}[Asymptotic Power]\label{thm_power}
Assume the Assumptions \ref{aspA1}--\ref{aspA4} and \ref{asp_power} hold and the signal strength satisfies  (\ref{eq_signal}). Denote $c^*=G^{-1}({\eta a_1}/{M})/\sqrt{\log p}$, where $a_1>0$ is defined in (\ref{eq_a1}), $G(t)=2-2 \Phi(t)$ and $\eta<0.5$ is the desired FDR level. When the signal strength satisfies
\begin{equation}\label{eq_thm_power_1}
\delta\geq c^*\{\mathrm{\textbf{cov}}( u({\beta_0^{jk}},\emph{\textbf{X}}_{jk}^{\mathrm{s2}},Y))_{(4,4)}\}^{1/2}+2\lambda \tilde K+\zeta
\end{equation}
for an arbitrarily small constant $\zeta>0$, we have 
\begin{equation}\label{eq_thm_power_2}
\Big|\mathrm{power}-  \mathrm{E}\Big(\frac{a_1}{|\emph{H}_1|}\Big)\Big|\rightarrow 0,
\end{equation}
as $(n,p) \rightarrow \infty$. 
\end{theorem}

This theorem implies that the power of our two-stage method converges to $\mathrm{E}(a_1/|\emph{H}_1|)$.  Since the power cannot exceed $\mathrm{E}(a_1/|\emph{H}_1|)$ by the definition (\ref{power}), this theorem gives a sharp characterization of the signal strength conditions under which the two-stage method reaches the optimal power. Before we detail the condition (\ref{eq_thm_power_1}), we first note that the optimal power $\mathrm{E}(a_1/|\emph{H}_1|)$ is generally less than 1 and decreases with the tuning parameter $\alpha$. As a result, the two-stage method may suffer from loss of power, which can be viewed as the price to pay for using the two-stage method. However, when the interaction effect satisfies the following hierarchical structure, i.e., 
\begin{equation}\label{eq_structure}
\beta_{3jk}\neq 0~~\textrm{implies $\beta_{1j}\neq 0$ and $\beta_{1k}\neq 0$}, 
\end{equation}
one would expect that, for any $(j,k)\in H_1$, the test statistics $|\widehat{T}_j|$ and $|\widehat{T}_k|$ in stage 1 are large enough to exceed $\alpha$. Thus, $a_1$ can be close to the total number of alternatives $|H_1|$, so that there is no loss of power for the two-stage method (i.e., power$\rightarrow 1$).

We now discuss the signal strength condition (\ref{eq_thm_power_1}) under which the two-stage method attains the optimal power. Interestingly, it can be shown that compared to the standard BH procedure (i.e., taking $\alpha=0$), the two-stage method may require a more relaxed signal strength condition. To simplify the discussion, we assume $a_1\asymp p^{\theta_1}$ and $N\asymp p^{\theta_0}$ for some constants $0\leq \theta_0,\theta_1\leq 2$. First, consider the case  $\theta_1\geq \theta_0$. By the definition of $c^*$, we can show that $c^*\asymp (\log p)^{-1/2}$. Thus, as $(n,p) \rightarrow \infty$,  (\ref{eq_thm_power_1}) is equivalent to
\begin{equation}\label{eq_delta_case1}
\delta\geq 2\lambda \tilde K+\zeta^*,
\end{equation}
where $\zeta^*=\zeta+o(1)$. In the second case $\theta_1<\theta_0$, by the Gaussian tail bound, we can derive $c^*=\{2(\theta_0-\theta_1)\}^{1/2}+o(1)$, and therefore (\ref{eq_thm_power_1}) reduces to
\begin{equation}\label{eq_delta_case2}
\delta\geq \{2(\theta_0-\theta_1)\mathrm{\textbf{cov}}( u({\beta_0^{jk}},\emph{\textbf{X}}_{jk}^{\mathrm{s2}},Y))_{(4,4)}\}^{1/2}+2\lambda \tilde K+\zeta^*.
\end{equation}
The above results (\ref{eq_delta_case1}) and (\ref{eq_delta_case2}) together imply that the signal strength condition is weaker as $\theta_0-\theta_1$ decreases. Recall that the standard BH procedure corresponds to $N=|H_0|$ and $a_1=|H_1|$. In contrast, the two-stage method with some proper $\alpha>0$ may significantly reduce $N$ or equivalently $\theta_0$. In addition, if the hierarchical structure (\ref{eq_structure}) holds, we expect that $a_1 \approx |H_1|$ as discussed above.  Therefore, using the two-stage method may yield a smaller value of $\theta_0-\theta_1$ and a weaker signal strength condition. In line with (\ref{eq_delta_case1}) and (\ref{eq_delta_case2}), our simulation studies also confirm that the two-stage procedure may indeed lead to the improved power under some simulation settings. 
\begin{remark}\label{rem_1}
As shown in the discussion of Assumption \ref{aspA6}, a theoretically valid choice of the tuning parameter $\alpha$ is $\alpha = \sqrt{\alpha_1\log p}$ with some small constant $\alpha_1>0$. Since $\hat T_j$ is a normalized test statistic, the choice of $\alpha$ is not affected by the scale of the data or the noise variance (e.g., in linear regression). Thus, we expect that a universal choice of $\alpha$ or equivalently $\alpha_1$ may work well in practice. Depending on the applications, a good choice of $\alpha$ is to balance the computational cost and the power of the FDR control procedure (as the value of $a_1$ in (\ref{eq_a1}) depends on $\alpha$). In our simulation, we find that choosing $\alpha_1$ in $[0.1,0.5]$ often yields satisfactory power  and also significantly reduces the computational cost. In the real data analysis, since $p$ is very large, we use a slightly larger $\alpha_1=0.8$ to further reduce the computational cost. 
\end{remark}


\section{Simulation}\label{simulation}

\subsection{Simulation settings}
We conduct simulations to evaluate the performance of our two-stage method. We consider logistic model and linear model in correctly specified case and misspecified case respectively. For each model, we generate a $p$-dimensional multivariate normal random vector $X = (X_1,...,X_p) \sim \mathcal{N}(0, \Sigma)$, where $\Sigma$ is set to be the identity matrix 
or $\Sigma_{jk} = 0.5^{|j-k|}$ with the latter introducing some correlation among variables. The FDR control level is set to be $\eta = 0.05$. We consider $p = 100,500$ with sample size $n=1000$ for logistic model, and $p=100,n=50$ or $p=500,n=100$ for linear model. 
For each setting, the simulation is repeated 100 times.

\begin{itemize}
    \item \emph{Correctly-specified models.}
    The data are generated from the GLM in (\ref{glm0}) with 
\begin{equation}\label{data0}
\theta_{jk}=\beta_{0jk} + \beta_{1jk}x_{j} + \beta_{2jk}x_{k} + \beta_{3jk}x_{j}x_{k}.
\end{equation}
Specifically, for $1\leq j < k \leq p$, we consider the linear model
\begin{equation}\label{data1}
    Y = \theta_{jk} + \epsilon,
\end{equation}
where $\epsilon \sim \mathcal{N}(0,1)$ and logistic model
\begin{equation}\label{data2}
    \mathrm{logit}\left(\mathrm{P}\left(Y=1 \mid X_j=x_{j}, X_k=x_{k}\right)\right)=\theta_{jk}. 
\end{equation}
The parameters $\beta_{1jk}$ and $\beta_{2jk}$ are randomly chosen from the set $\{0, b\}$, where we vary $b$ from $0.2$ to $1$. We adopt the hierarchical structure (\ref{eq_structure}). If either $\beta_{1jk}$ or $\beta_{2jk}$ is $0$, $\beta_{3jk}$ is set to be $0$, otherwise it is randomly chosen from $\{0, b\}$ with probability 0.25 for 0 and 0.75 for $b$. We set $\beta_{0jk}=-2$ in logistic model and $\beta_{0jk}=-1$ in linear model. The covariance matrix of $X$ is the identity matrix. 
The tuning parameter $\alpha$ in stage 1 is chosen as
\begin{equation}\label{eq_alpha_1}
\alpha = \sqrt{\alpha_1 \log p},
\end{equation}
for some small constant $\alpha_1$ (e.g., $0.1$ or $0.5$). As we discussed before, if $\alpha_1 = 0$, all variables will pass stage 1 and the two-stage FDR control method reduces to the classical BH procedure. 

    \item \emph{Misspecified models.}
 We introduce additional terms to $\theta_{jk}$ in (\ref{data0}) so that the GLM (\ref{glm0}) is misspecified. Specifically, for $0 \leq j <k \leq p$, we set 
 \[
\theta_{jk}=\beta_{0jk} + \beta_{1jk}x_{j} + \beta_{2jk}x_{k} + \beta_{3jk}x_{j}x_{k}+ \beta_{4jk}x_{l} + \beta_{5jk}x_{u}x_{v},
 \]
 where $l,u,v$ are randomly chosen from $\{1\leq l,u,v \leq p: l,m,n \neq j, k\}$, therefore variables different from $x_j$ and $x_k$ are included in the model leading to misspecification. Given $\theta_{jk}$, the data generating model is still (\ref{data1}) for linear model and (\ref{data2}) for logistic model. The parameters $\beta_{1jk}$, $\beta_{2jk}$ and $\beta_{3jk}$ are generated in the same way as the correctly-specified case. The additional parameters $\beta_{4jk}$ and $\beta_{5jk}$ are again randomly chosen from the set $\{0, b\}$. Unlike the correctly-specified case, the covariance matrix of $X$ is  $\Sigma_{jk} = 0.5^{|j-k|}$.
\end{itemize} 
 
In the simulation studies, we compare the performance of the classical BH procedure with our two-stage method using two different values of $\alpha_1$ under the above data generating models. To evaluate the finite sample performance of the methods, we compute the empirical FDR (i.e., FDP) and empirical power defined in (\ref{eq_FDP}) and (\ref{power}) respectively, averaged over 100 simulations. 

In addition, we compare the computation efficiency of the methods. Note that when $\alpha_1=0$, all $(p-1)p/2$ pairs of variables need to be tested for interaction effect. If the two-stage method is used, $p$ tests are conducted in stage 1 and another  $(p_1-1)p_1/2$ tests are conducted in stage 2, where $p_1$ is the number of variables pass stage 1 for a chosen $\alpha_1$. Since the computation time of the algorithm roughly scales with the number of tests conducted, we define the  computation efficiency $\omega$ as the ratio of the number of tests conducted relative to the BH procedure, 
 \[
\omega = \frac{2p + p_1(p_1-1)}{p(p-1)}.
 \]

\begin{figure}
\centering
\subfigcapskip -5pt
\subfigure[p = 100]{
\label{Fig.sub.1}
\includegraphics[width=0.48\textwidth]{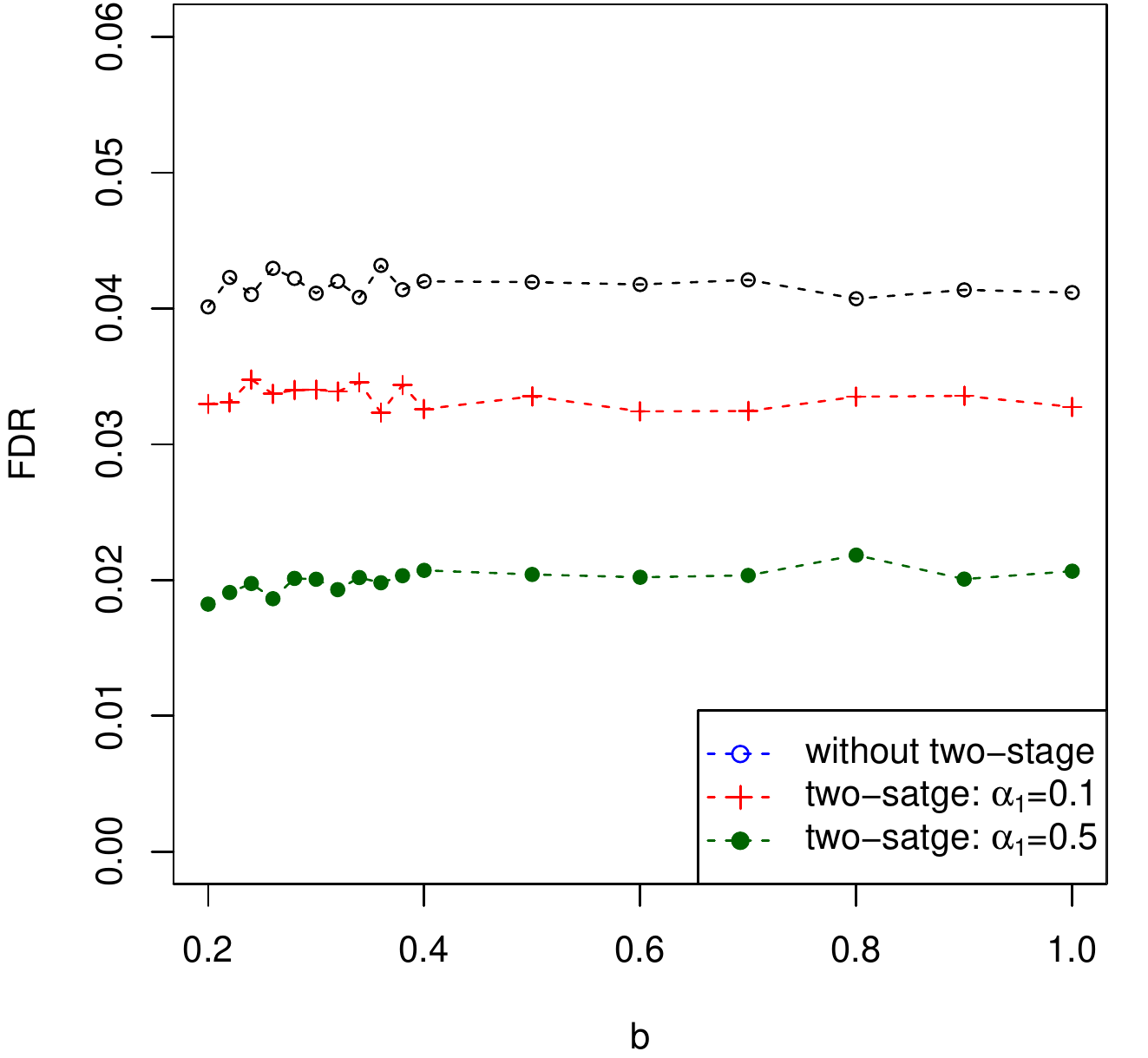}}
\subfigure[p = 100]{
\label{Fig.sub.2}
\includegraphics[width=0.48\textwidth]{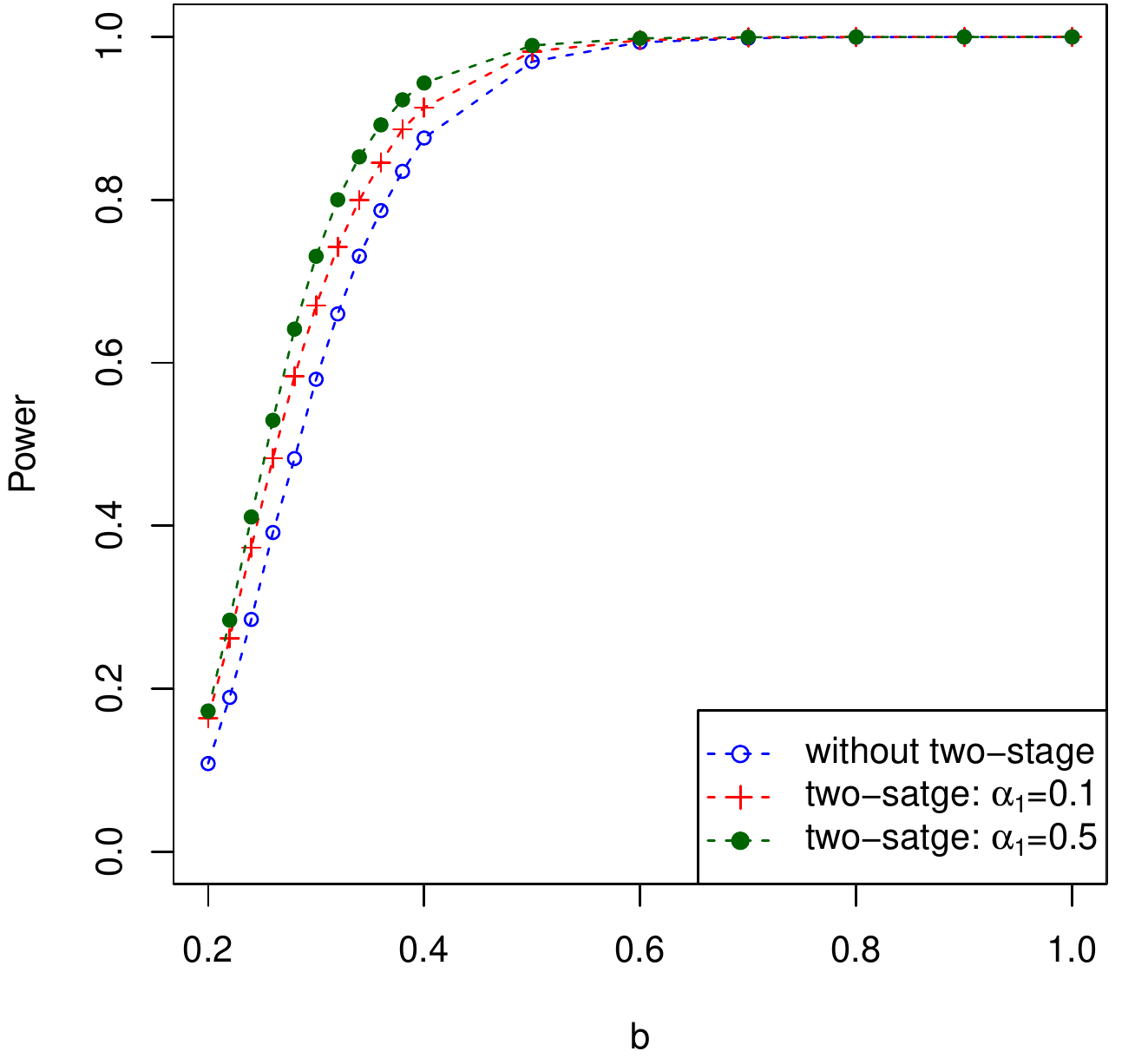}} \vskip -5pt 
\subfigure[p = 500]{
\label{Fig.sub.3}
\includegraphics[width=0.48\textwidth]{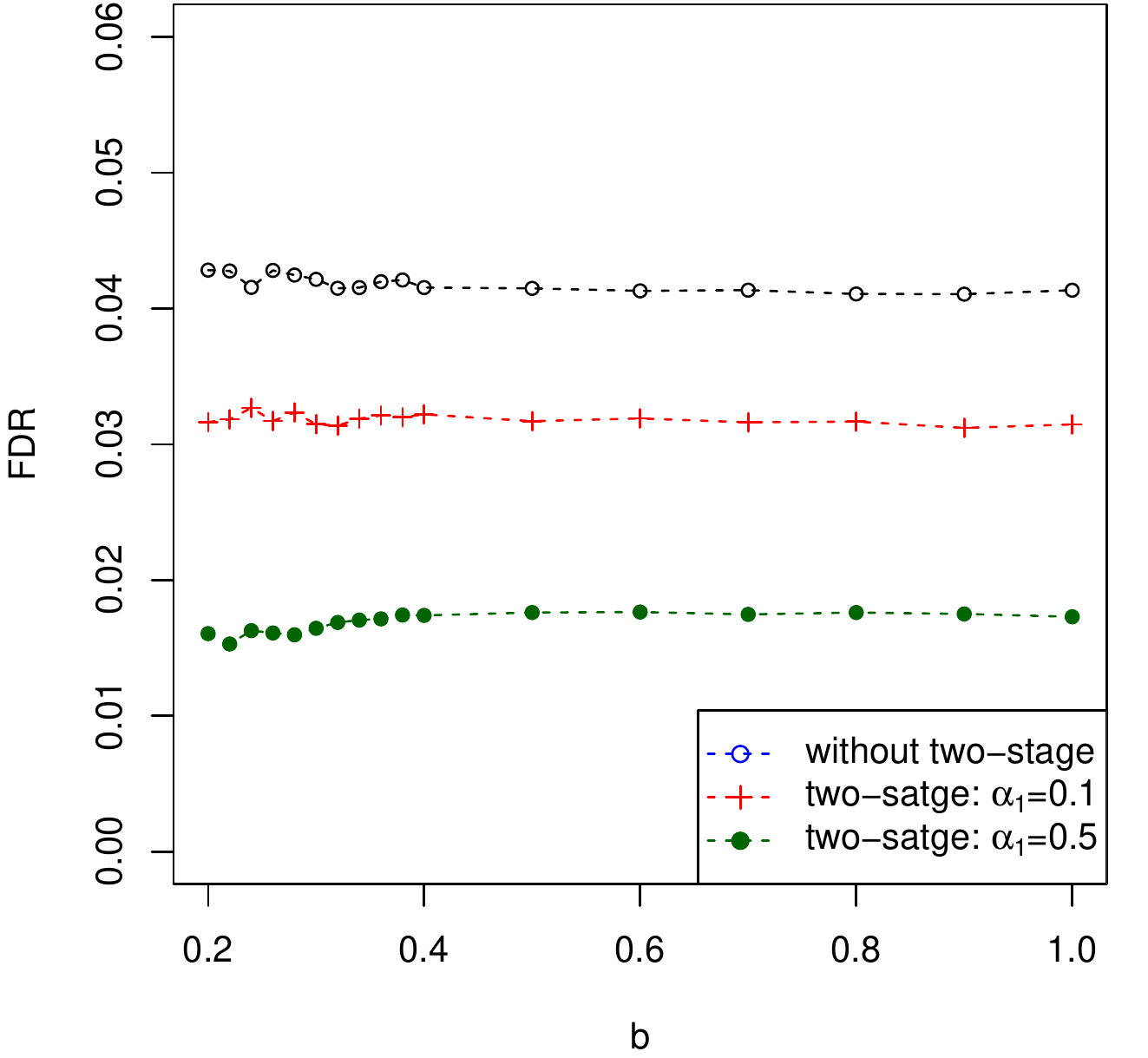}}
\subfigure[p = 500]{
\label{Fig.sub.2}
\includegraphics[width=0.48\textwidth]{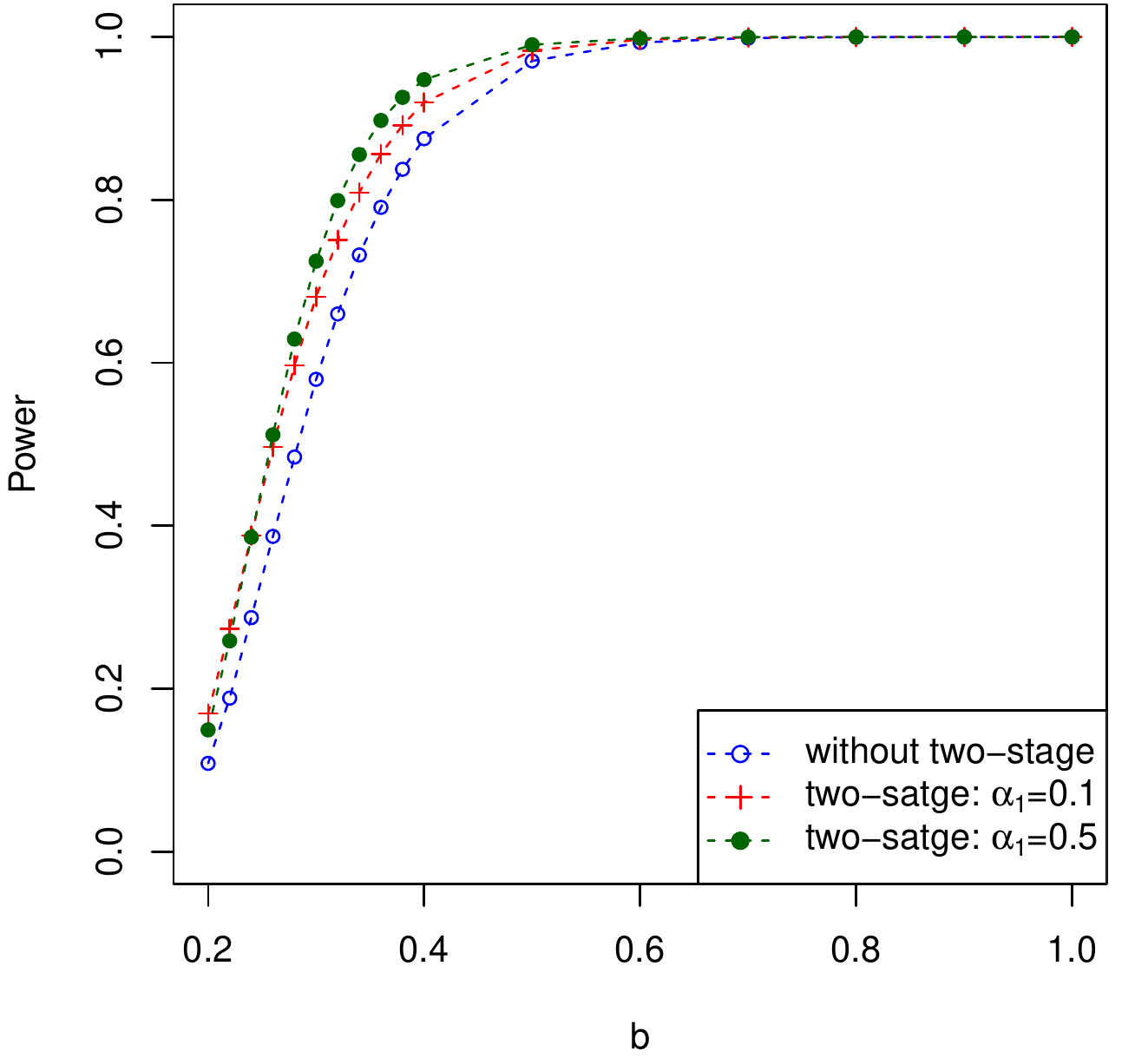}}
\caption{FDR and power curves for correctly-specified logistic models.}
\label{glm_correct}
\end{figure}

\begin{figure}
\centering
\subfigcapskip -5pt
\subfigure[p = 100]{
\label{Fig.sub.1}
\includegraphics[width=0.48\textwidth]{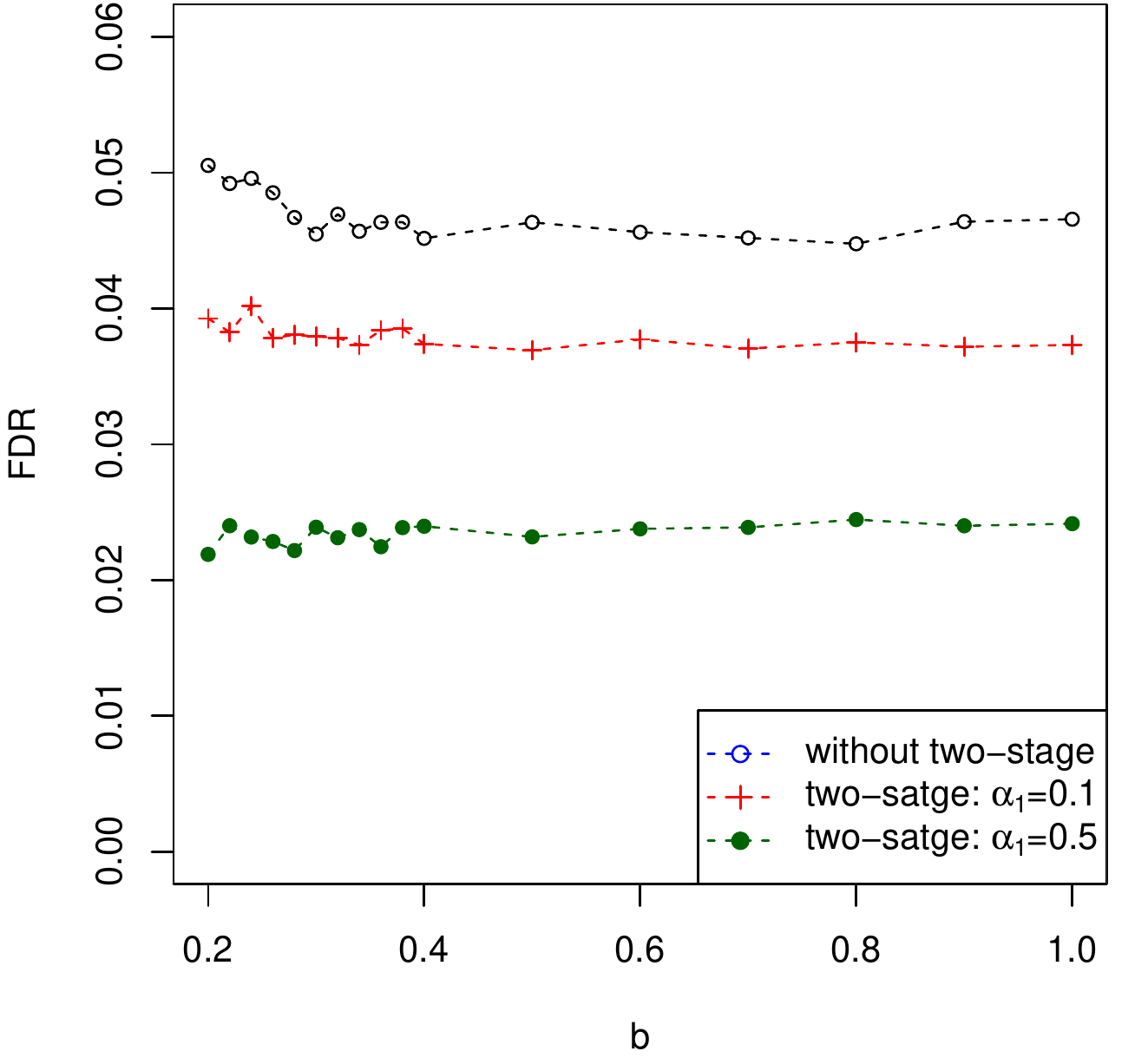}}
\subfigure[p = 100]{
\label{Fig.sub.2}
\includegraphics[width=0.48\textwidth]{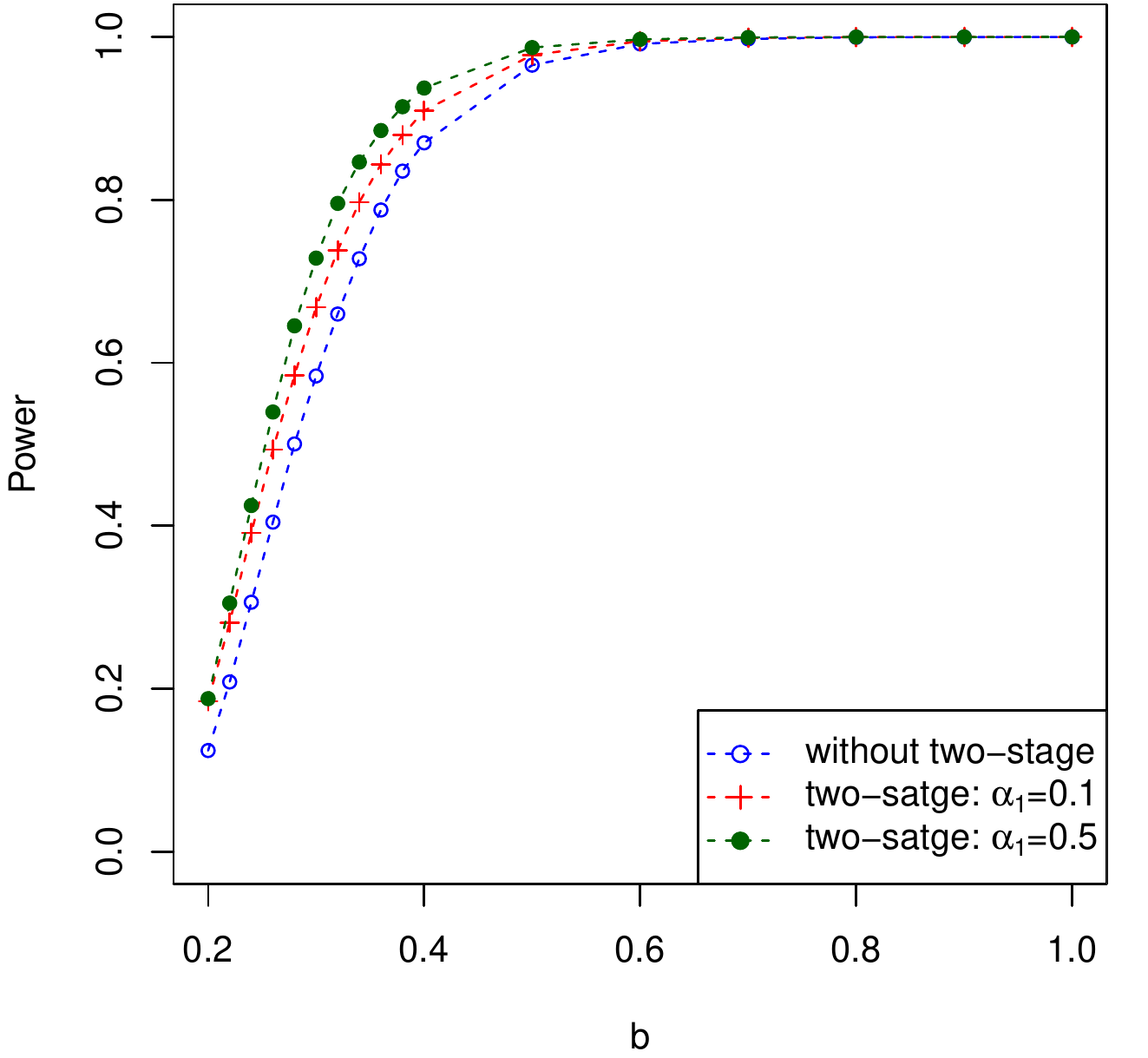}} \vskip -5pt 
\subfigure[p = 500]{
\label{Fig.sub.3}
\includegraphics[width=0.48\textwidth]{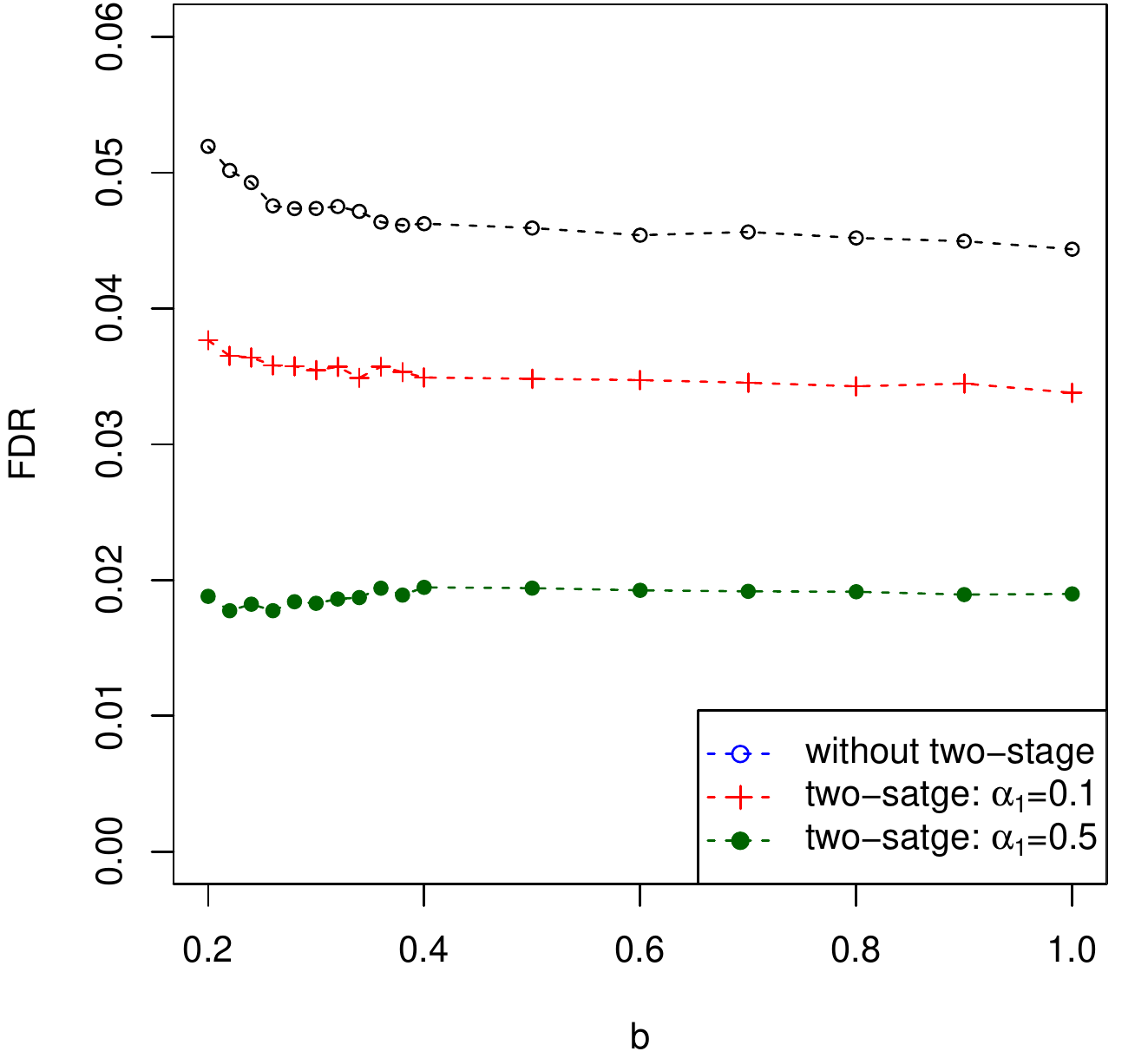}}
\subfigure[p = 500]{
\label{Fig.sub.2}
\includegraphics[width=0.48\textwidth]{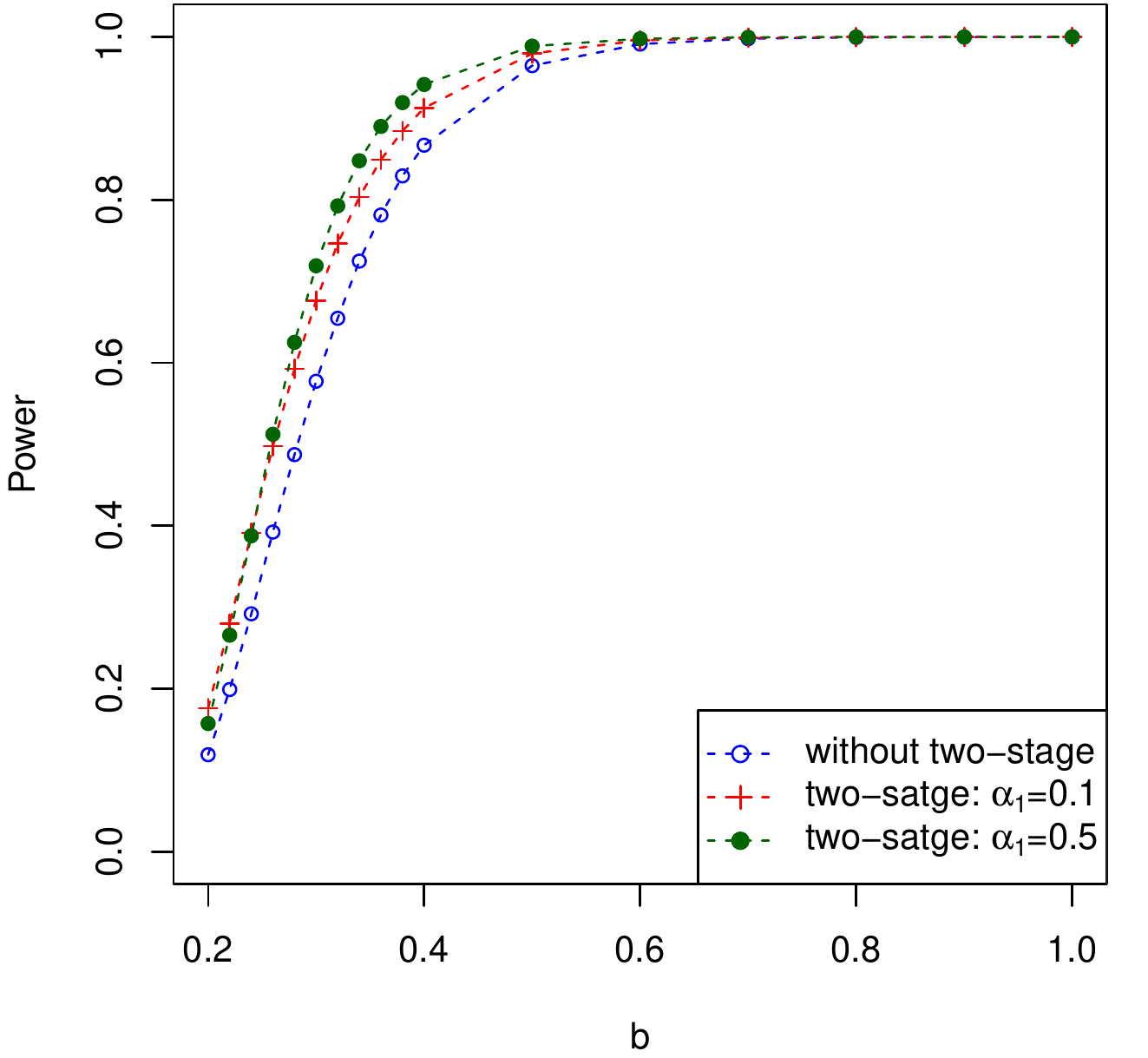}}
\caption{FDR and power curves for misspecified logistic models. }
\label{glm_miss}
\end{figure}

\begin{figure}
\centering
\subfigcapskip -5pt
\subfigure[n = 50, p = 100]{
\label{Fig.sub.1}
\includegraphics[width=0.48\textwidth]{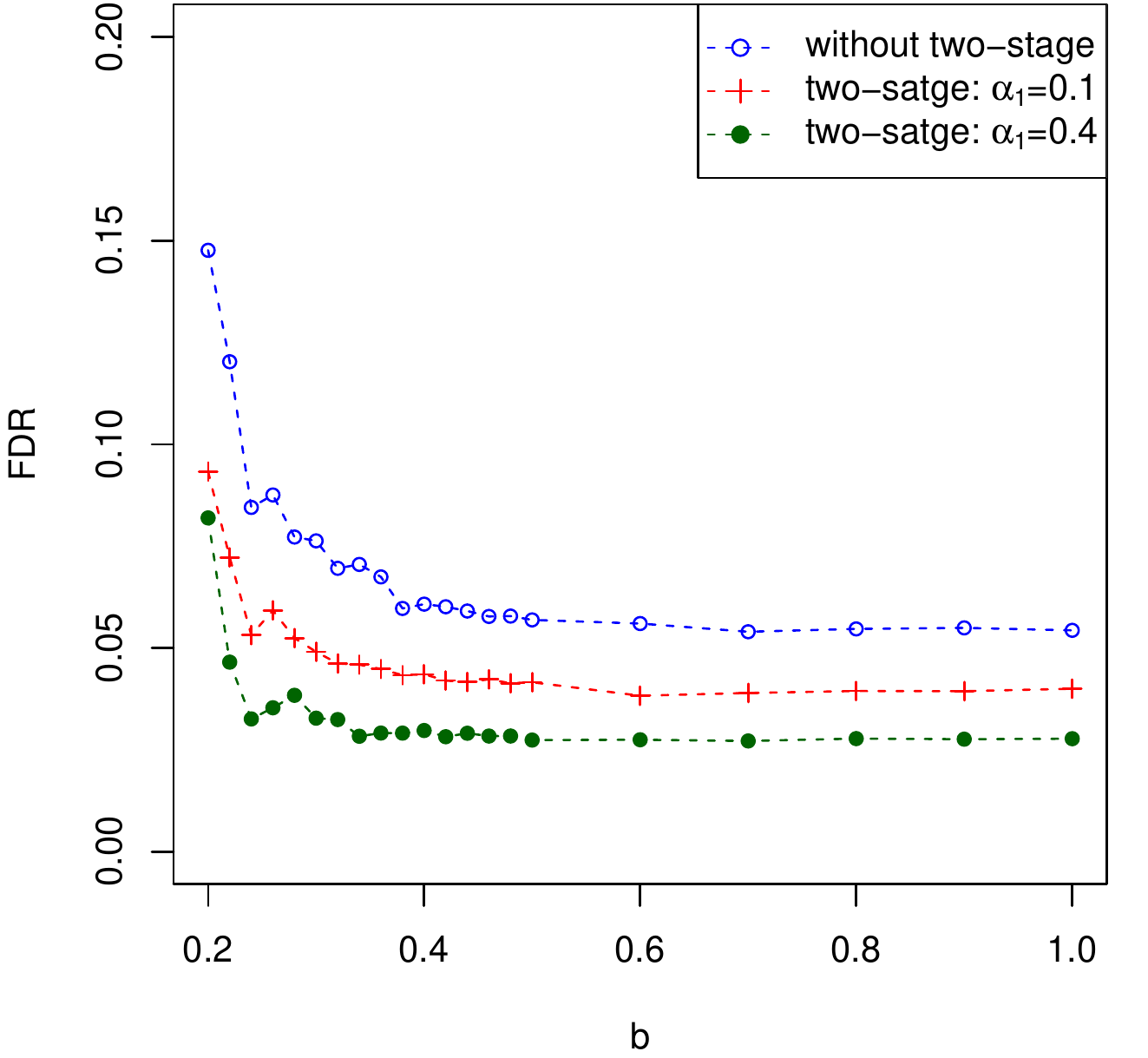}}
\subfigure[n = 50, p = 100]{
\label{Fig.sub.2}
\includegraphics[width=0.48\textwidth]{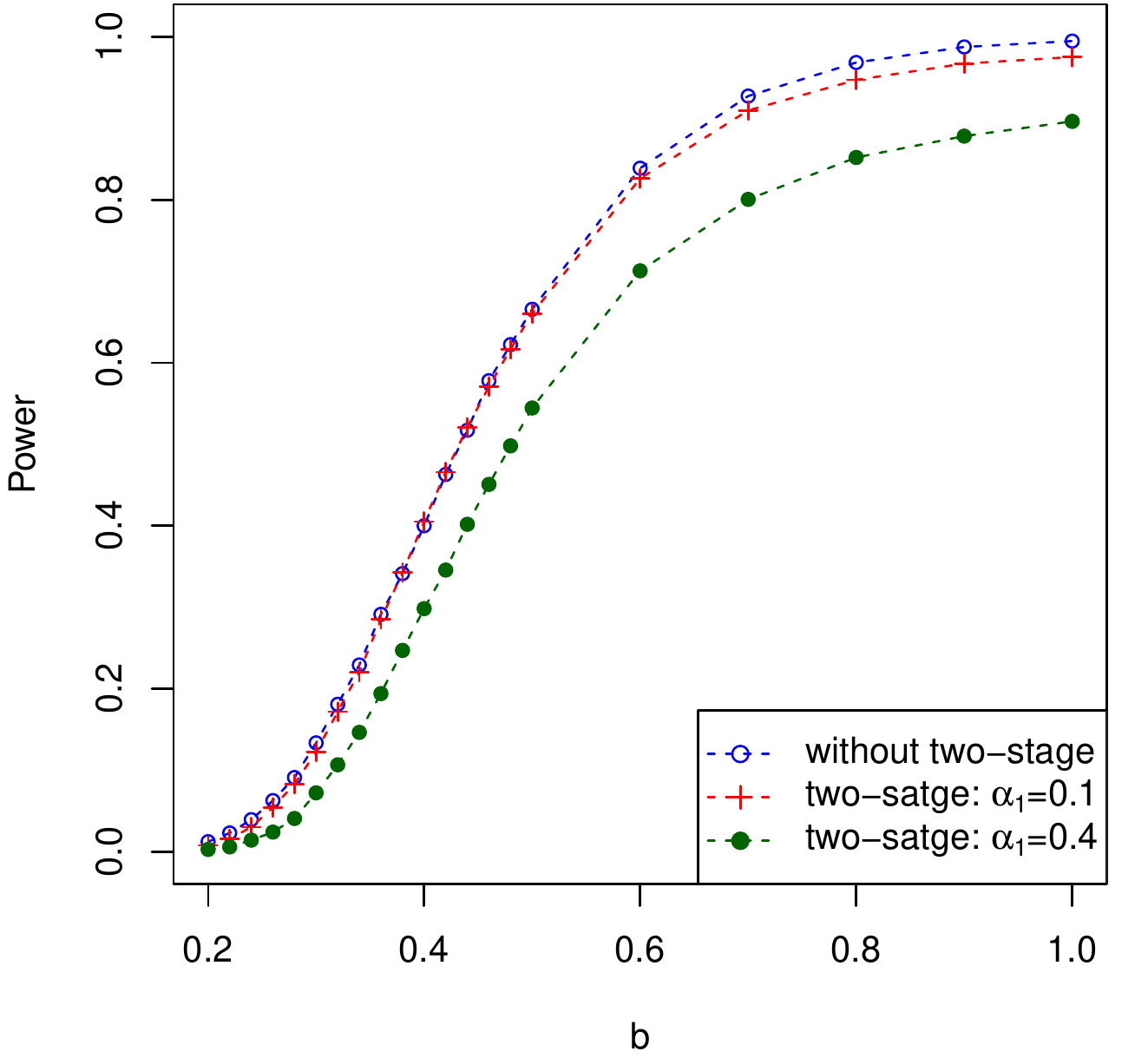}} \vskip -5pt 
\subfigure[n = 100, p = 500]{
\label{Fig.sub.3}
\includegraphics[width=0.48\textwidth]{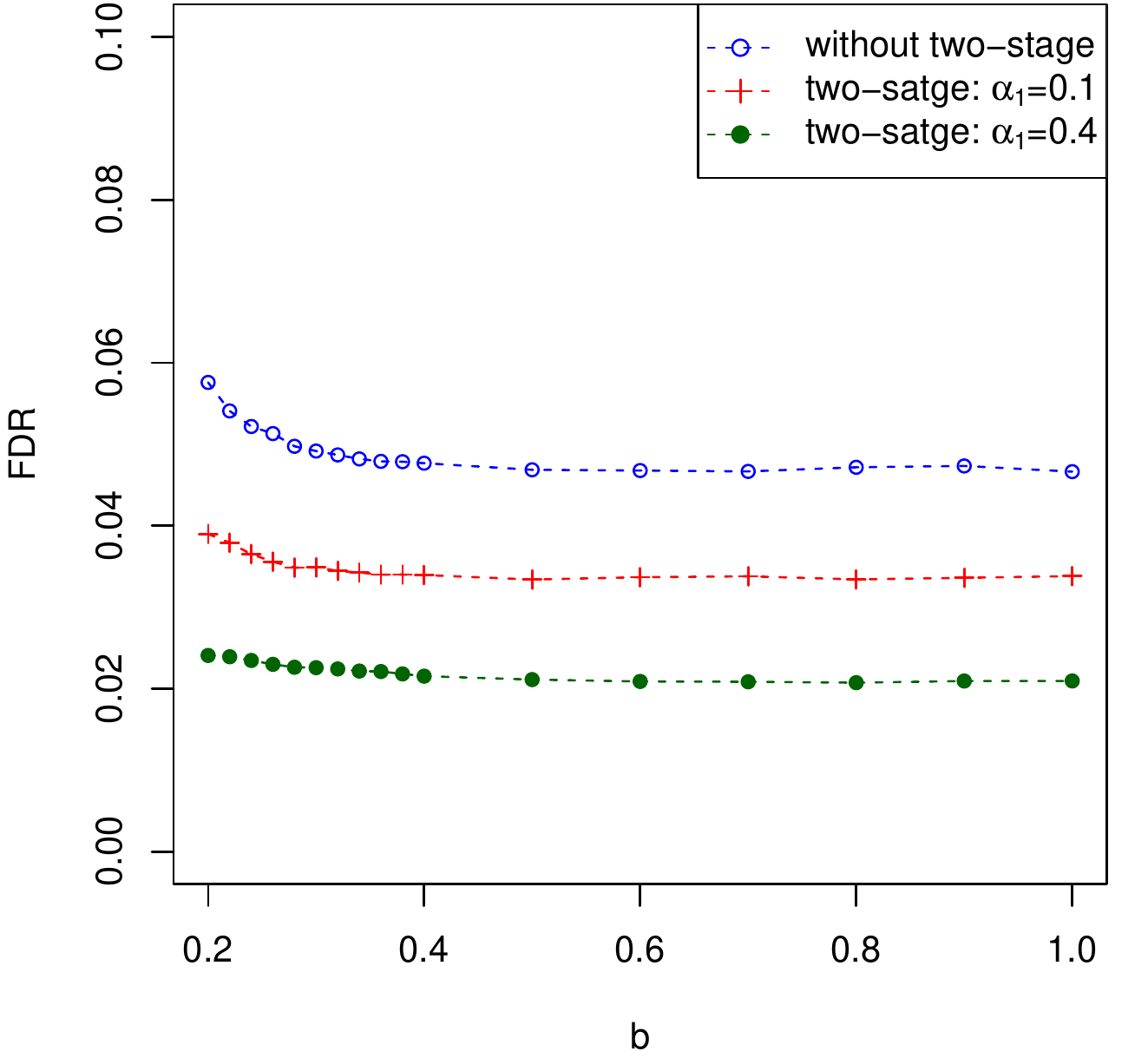}}
\subfigure[n = 100, p = 500]{
\label{Fig.sub.2}
\includegraphics[width=0.48\textwidth]{data_image/lm_p100_power.pdf}}
\caption{FDR and power curves for correctly-specified linear models.}
\label{lm_correct}
\end{figure}

\begin{figure}
\centering
\subfigcapskip -5pt
\subfigure[n = 50, p = 100]{
\label{Fig.sub.1}
\includegraphics[width=0.48\textwidth]{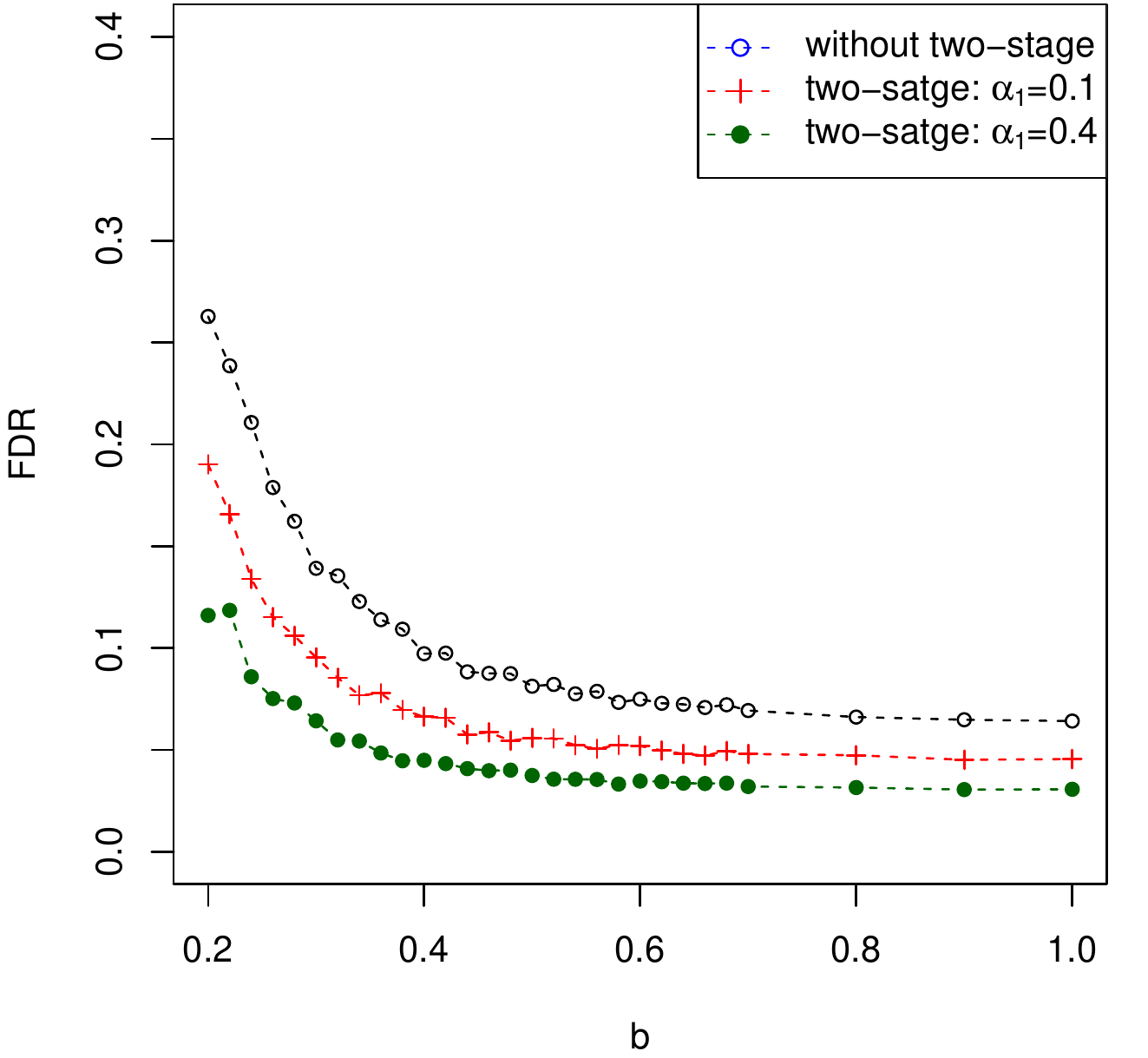}}
\subfigure[n = 50, p = 100]{
\label{Fig.sub.2}
\includegraphics[width=0.48\textwidth]{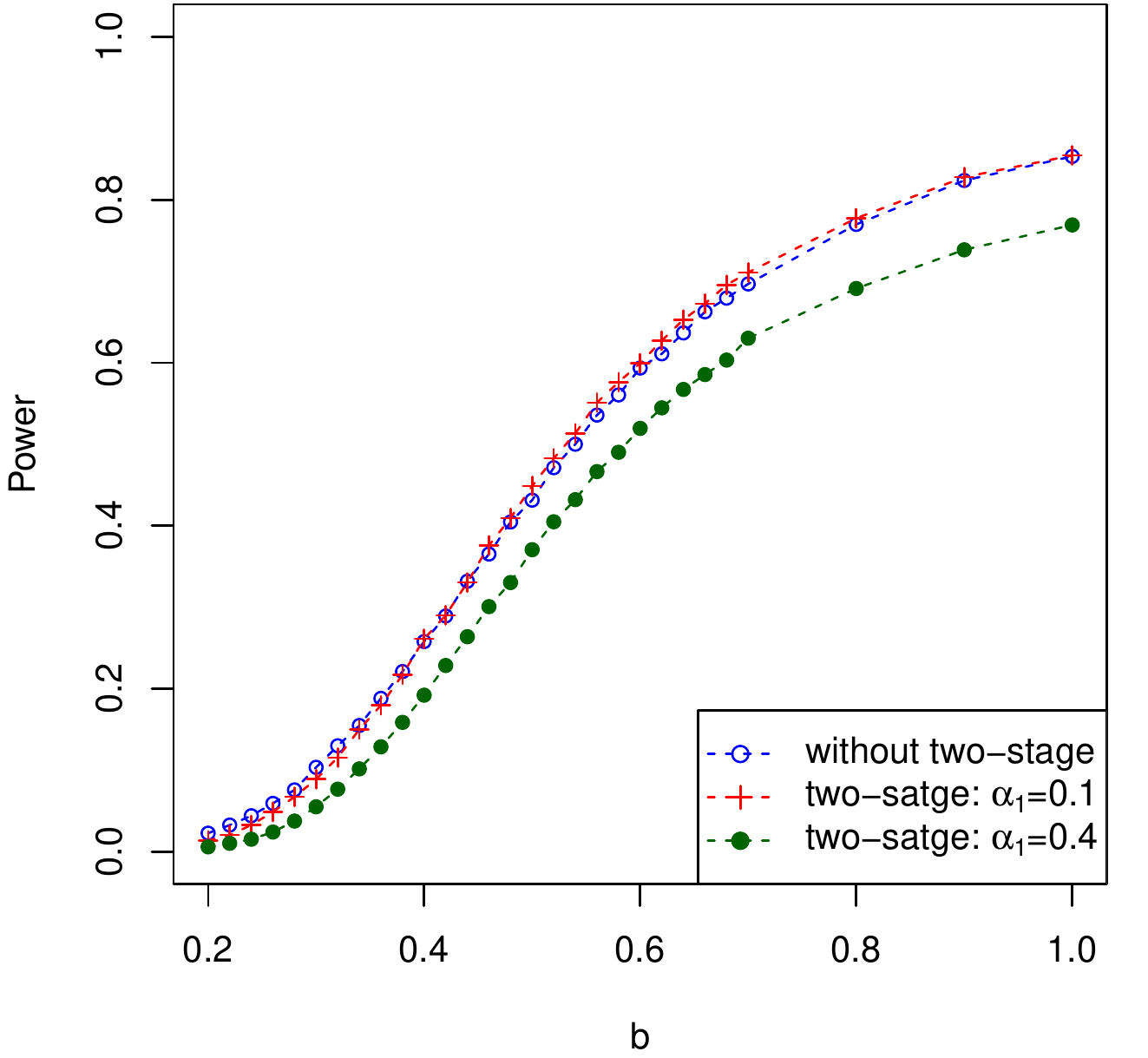}} \vskip -5pt 
\subfigure[n = 100, p = 500]{
\label{Fig.sub.3}
\includegraphics[width=0.48\textwidth]{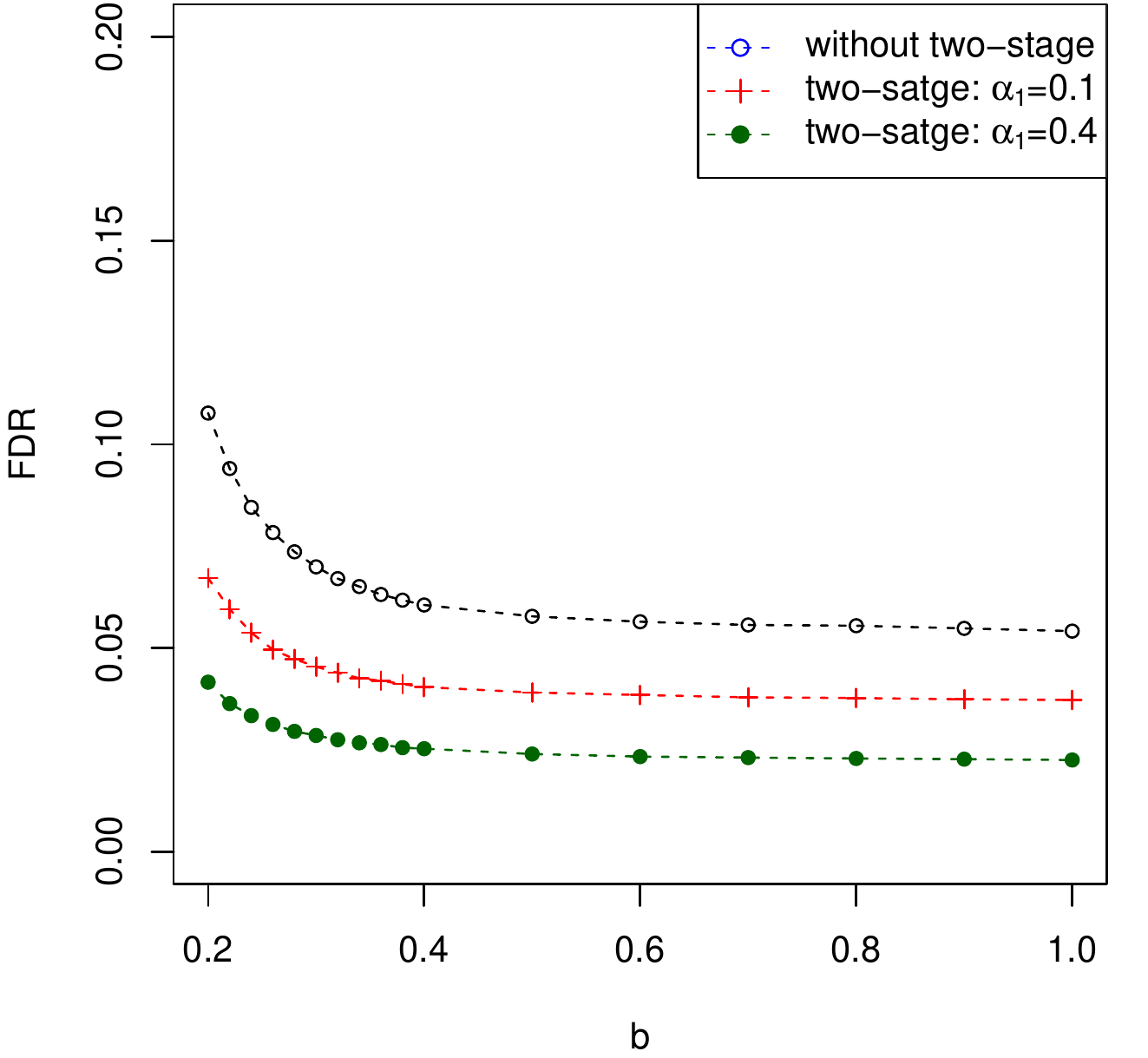}}
\subfigure[n = 100, p = 500]{
\label{Fig.sub.2}
\includegraphics[width=0.48\textwidth]{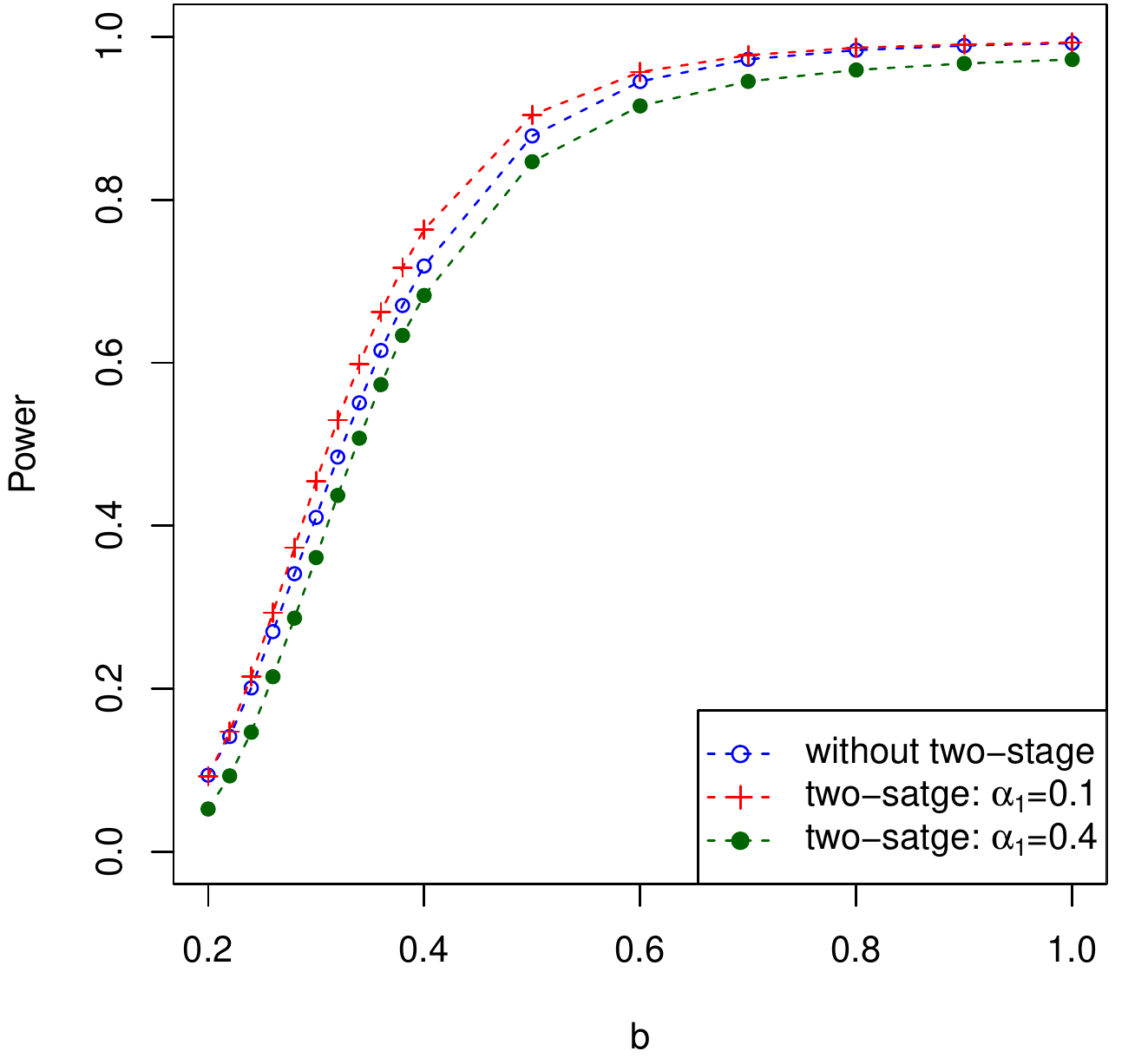}}
\caption{FDR and power curves for misspecified linear models. }
\label{lm_miss}
\end{figure}

\subsection{Simulation Results}
The empirical FDR and power curves for logistic model are shown in Figure~\ref{glm_correct} and Figure~\ref{glm_miss}.  In both correctly-specified and misspecified cases, FDR can be controlled below the desired level 0.05. Compared to the classical BH method, the FDR of the two-stage method is lower because variables with weak main effect are excluded from stage 2, and false discoveries are less likely to happen. As expected, the FDR reduces when we increase $\alpha_1$ from 0.1 to 0.5.

While one may expect the two-stage method with a relatively large $\alpha_1$ may screen out some informative variables in stage 1, leading to loss of power, interestingly, panels (b) and (d) in Figure~\ref{glm_correct} and Figure~\ref{glm_miss} show that the two-stage method with a proper $\alpha_1$ can be even more powerful than the classical BH method. This is in line with the discussion after Theorem \ref{thm_power}. 
Such power improvement is more evident when the signal size $b$ is small or moderate. When the signal size is large enough, such as when $b\geq 0.6$, the power of all the methods converges to 1.




For the linear models in Figure~\ref{lm_correct} and Figure~\ref{lm_miss}, the FDR from the BH procedure may sometimes far exceed the desired level $0.05$.  The reason is that, when the signal size is  small, there are very few discoveries based on the asymptotic p-values, leading to unstable FDR. The two-stage method, however, significantly outperforms the BH method and the resulting FDR is smaller or closer to the desired level $0.05$. In terms of the power, we see that the two-stage method with $\alpha_1 = 0.1$ is comparable to the BH method. As we increase the threshold to $\alpha_1 = 0.4$, the two-stage method becomes less powerful than the BH method, especially when $b$ is relatively large, meaning that we may miss some variables that have interaction effects when using a more stringent rejection rule in stage 1.

The comparison of the computation efficiency when using different $\alpha_1$ is summarized in Table~\hyperlink{Table1}{1}. It is seen that the number of tests conducted in the two-stage method is around $1/4\sim 1/2$ of the BH procedure. Thus, the two-stage method is much more computationally efficient than the BH procedure, especially when $p$ is large. 

In summary, compared to the standard BH procedure, the two-stage method often leads to a more reliable FDR control with improved or comparable power, and can be implemented with much less computation time. 

\begin{figure}
    \centering
    \includegraphics[  width=14cm,
  height=6.5cm,
  ]{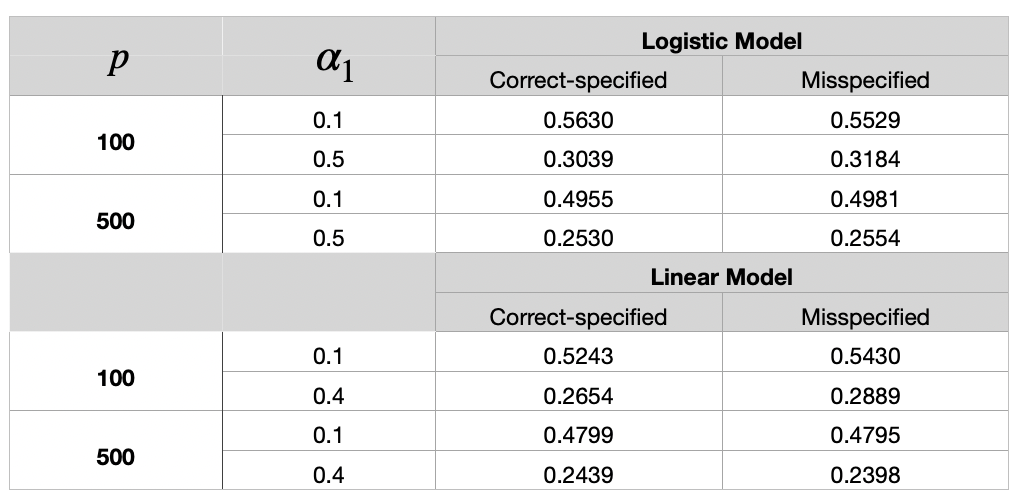}
    \caption*{Table 1: Computation efficiency $\omega$ of the two-stage method with different $\alpha_1$.}
    \hypertarget{Table1}{}
\end{figure}

\section{Real Data Application}\label{data}

Bladder cancer is one of the most common cancers. In 2022, an estimated 81,180 new cases in the United States were diagnosed with bladder cancer, with about 17,100 deaths from bladder cancer \citep{cancerstat}. A great deal of efforts have been devoted to identify genetic susceptibility loci for bladder cancer through GWAS studies \citep{kiemeney2008sequence, kiemeney2010sequence, rafnar2009sequence, wu2009genetic, rothman2010multi}. Despite these efforts, the molecular mechanism including epistasis for bladder cancer has not been well understood. In this section, we apply the proposed two-stage hypothesis testing procedure to a bladder cancer data set from the database of Genotypes and Phenotypes \citep{tryka2014ncbi}. In particular, we focus on the United States/Finland cohort (genotyped on a 610 K chip) from this dataset.

Before applying our two-stage hypothesis testing procedure, we conduct quality control (QC) filters through PLINK \citep{purcell2007plink}, including removing subjects with more than 5\% missing genotypes, and removing SNPs with a minor allele frequency less than 1\% and those with more than 5\% missing genotypes. This leads to a total 102,172 SNPs from 2,479 cases and 2,273 controls. We also use the PLINK software to prune the SNPs using a pairwise $r^2 > 0.2$ to reduce the influence of strong linkage disequilibrium (LD) on the assessment of interaction effects. The final dataset contains 95,094 SNPs to be analyzed. To control for the potential impacts of population stratification, we apply the principal component analysis (PCA) from the R package SNPRelate \citep{zheng2012high}. The potential effect of population stratification is adjusted in the second stage of our testing procedure by fitting the first five eigenvectors from the PCA of the SNP genotypes.

In this analysis, we adopt the dominant model. To implement the proposed two-stage hypothesis testing procedures, we set  $\alpha_1=0.8$ in (\ref{eq_alpha_1}) and in total 385 SNPs pass the first stage. 
We then test for the interactions between these SNPs in the second stage. As a result, 67 pairs of SNPs are identified by the proposed method with FDR level $0.1$.  Table~\hyperlink{Table2}{2}
presents 13 of them in which both SNPs occur within an identified gene in an ascending order of $p$-value, with the largest $p$-value being $1.09 \times 10^{-09}$.  



To evaluate potential biological relationships among the genes shown in Table~\hyperlink{Table2}{2}, we compare our findings with the existing results from GIANT \citep{greene2015understanding}, which provides genome-wide functional interaction networks obtained from a Bayesian approach that integrates thousands of diverse experiments. For ease of visualization, we only reproduce the gene network from GIANT for those identified in Table~\hyperlink{Table2}{2}. The result is shown in Figure~\ref{real.fig.1}. In this figure, two genes are connected if the posterior probability of the functional relationship is greater than 50\%, with bolder edges having posterior probability greater than 89\% \citep{greene2015understanding}. We find that all interactions identified by our proposed procedures in Table~\hyperlink{Table2}{2} are connected by a pathway having less than four genes within the network, with only one non-query gene (CCL21). In particular, PRKCQ and ZBTB16, PRKCQ and NDUFB9, PRKCQ and ZBTB20,  PRKCQ and BANP are identified by our procedure to have very strong interaction effect in Table~\hyperlink{Table2}{2}, which are consistent with the results from GIANT,  as all of them are connected by two bolder edges through only one additional gene. 

We notice that the gene PRKCQ appears 10 out of 13 pairs in Table~\hyperlink{Table2}{2}, which seems to suggest its importance in bladder cancer development. Such conjecture can be further verified in the bladder cancer literature. Notably, by examining the suitability of rodent models of bladder cancer in rats to model clinical bladder cancer specimens in humans, \cite{lu2011cross} found that the gene PRKCQ is differentially expressed  between tumor and normal groups, and consistently observed as a down-regulated gene in at least two datasets. Meanwhile, \cite{zaravinos2011identification} used microarrays to identify common differentially expressed  genes among clinically relevant subclasses of bladder cancer. Their results showed that the gene PRKCQ is differentially expressed and related to cell growth in bladder tissue. Finally, we also confirm the role of some other genes identified by our method in bladder cancer via GTEx \citep{gtex2015genotype}, a database of  tissue-specific gene expression and regulation. The detailed results are deferred to Appendix \ref{app_numerical}. All these results show that the genes identified by our method are expressed in bladder tissue, which supports our data analysis results.

\begin{figure}
    \centering
    \includegraphics[ width=15cm,
  height=10cm,
  ]{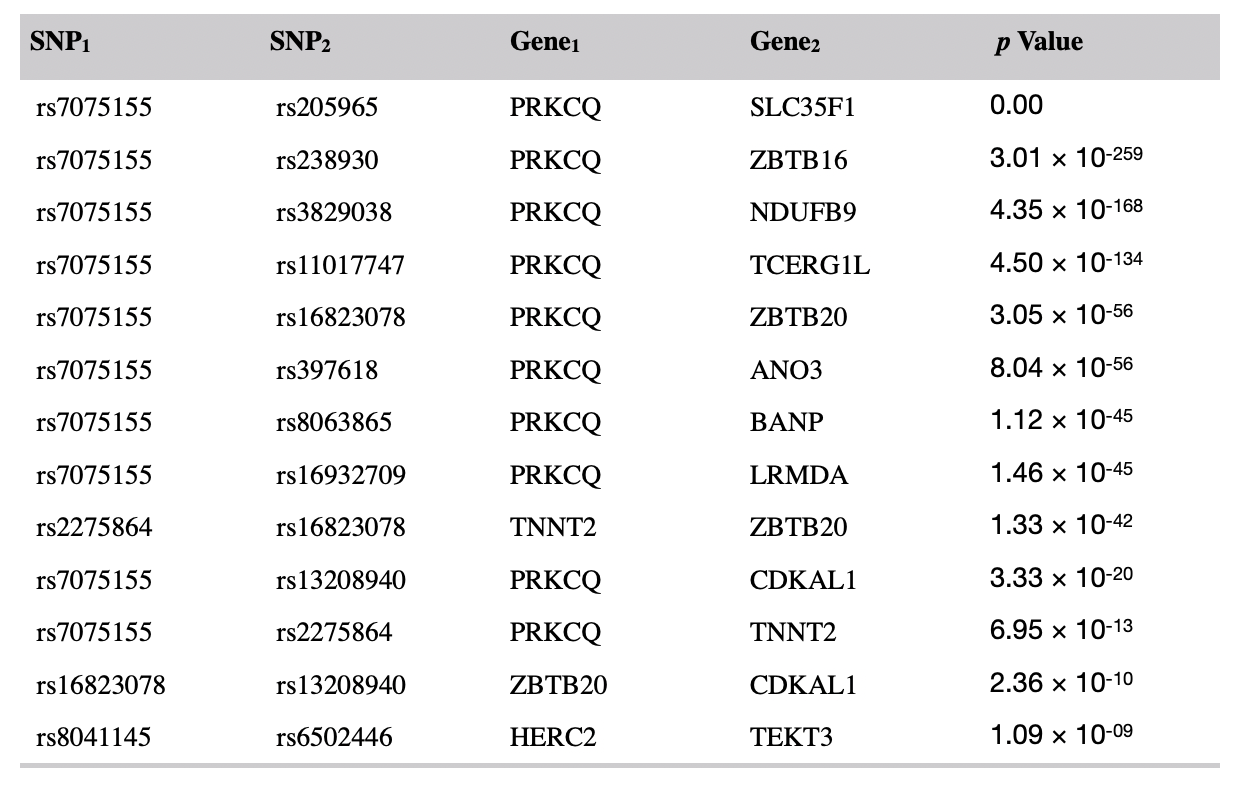}
    \caption*{Table 2: The SNPs and the corresponding genes in the bladder cancer data set that are identified by our two-stage hypothesis testing procedure with FDR level 0.1. The corresponding p-values in the second stage are also reported.}
\hypertarget{Table2}{}
\end{figure}

\begin{figure}
    \centering
    \includegraphics[ width=16cm,
  height=12cm,
  ]{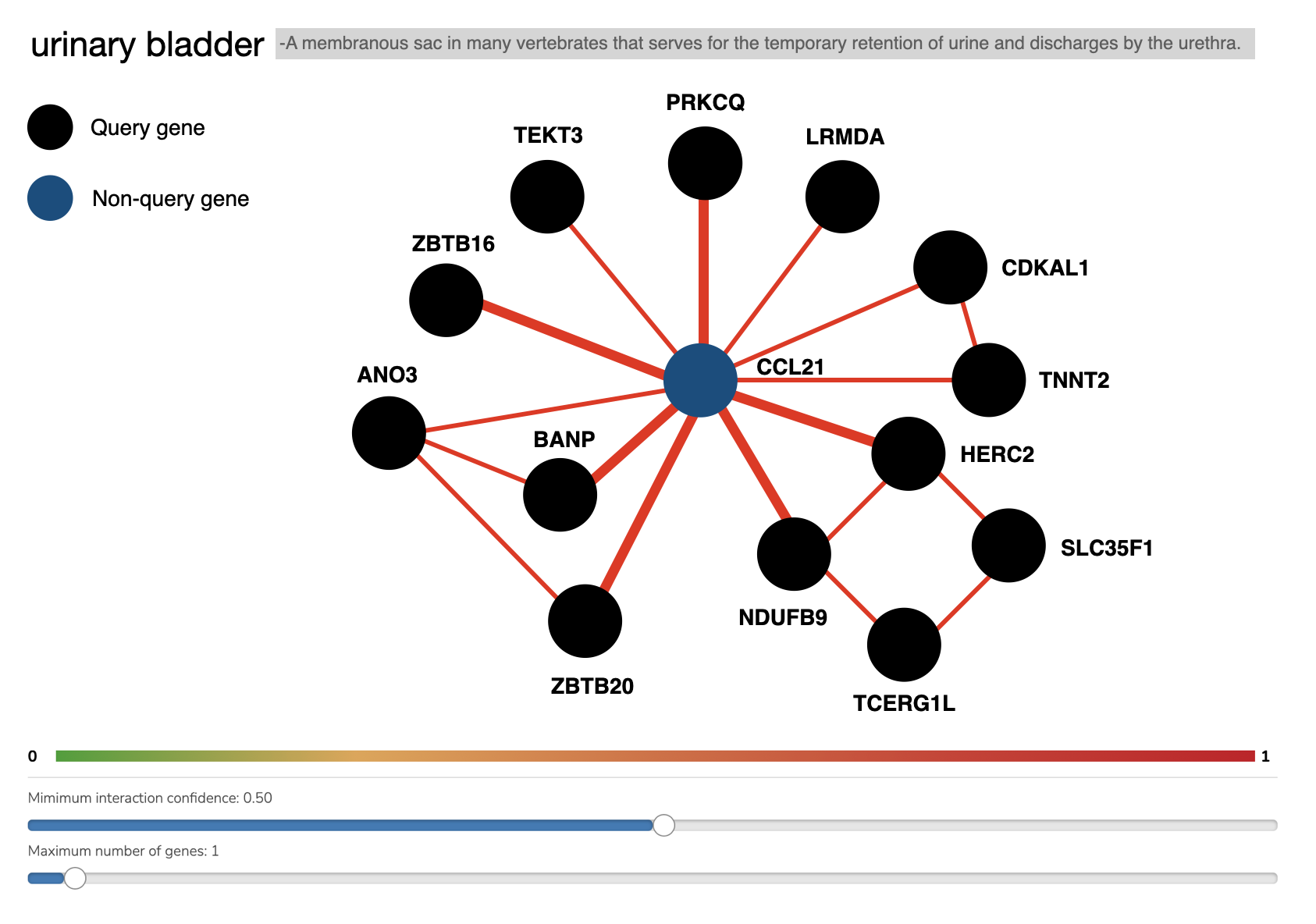}
    \caption{The GIANT urinary bladder network queried for the identified genes in Table 2. The network is filtered to include only edges with greater than a 50\% posterior probability of a functional relationship in this tissue, which represents a substantial increase over the prior probability, with only one additional gene CCL21.}
    \label{real.fig.1}
\end{figure}


\section*{Acknowledgement}

Yang Ning is supported by National Science Foundation (NSF) CAREER award DMS-1941945. Xi Chen is supported by the NSF [Grant IIS-1845444]. Yong Chen is supported in part by National Institutes of Health awards 1R01AG073435, 1R56AG074604, 1R01LM013519 and 1R56AG069880. 

\newpage

\appendix
\section{Technical Details for Assumption \ref{aspA6}}\label{app_assumption}

For any $1 \leq j \leq p$, if we set $\alpha = \sqrt{c\log p}$, then 
\begin{align}
     G_j(\alpha) &= 
      \mathrm{P}\left(|\mathcal{N}(0,1)+\sqrt{n}\mathrm{E}(U_{ij})|\geq \sqrt{c\log p}\right) \nonumber\\
     & =
     \mathrm{P}\left(\mathcal{N}(0,1)\geq
     \sqrt{c\log p} -  \sqrt{n}\mathrm{E}(U_{ij})\right)
     +
      \mathrm{P}\left(\mathcal{N}(0,1)\leq
     -\sqrt{c\log p} -  \sqrt{n}\mathrm{E}(U_{ij})\right)\nonumber\\
         &= 1 - \Phi\left( \sqrt{c\log p}-\sqrt{n}\mathrm{E}(U_{ij})\right)
     + 1 - \Phi\left(\sqrt{n}\mathrm{E}(U_{ij}) +\sqrt{c\log p}\right),\label{eq_app_1}
\end{align}
where $\Phi(t) = \mathrm{P}\left(\mathcal{N}(0,1) \leq t\right)$.
Note that 
\[
\mathrm{E}(U_{ij})=\frac{\left(\beta_0^j\right)_{(2)}}{\sqrt{\mathrm{\textbf{cov}}( u({\beta_0^{j}},\emph{\textbf{X}}_{j}^{\mathrm{s1}},Y))_{(2,2)}}},
\]
which depends on the least false value $\beta_0^j$. For this reason, we refer $|\mathrm{E}(U_{ij})|$ as the signal strength in the main paper. 

Let us define the set
$$
\tilde H_{01}(\tilde c)=\{(j,k)\in\tilde H_{01}: |\mathrm{E}(U_{ij})|\geq b, |\mathrm{E}(U_{ik})|\geq b\},
$$
where $b=\sqrt{\frac{\tilde c\log p}{n}}$ for some constant $\tilde c>0$. The set $\tilde H_{01}(\tilde c)$ is a subset of $\tilde H_{01}$ by excluding the pairs with the signal strength $|\mathrm{E}(U_{ik})|$ less than $b$. Under mild conditions on $\tilde H_{01}(\tilde c)$, we can verify (\ref{eq_G4}) in the following lemma. 

\begin{lemma}
Assume that $|\tilde H_{01}(\tilde c)|=\Omega(p^2)$. Then (\ref{eq_G4}) holds for any $\alpha = \sqrt{c\log p}$, where $c$ is a positive constant with 
\begin{equation}\label{eq_app_2}
c<\Big(\sqrt{\tilde c}+\sqrt{2-\max \left\{\frac{3}{2}+\frac{\delta}{1+\delta}, 2-\frac{\kappa}{2}+\frac{\delta}{1+\delta}\right\}}\Big)^2.    
\end{equation}
\end{lemma}

\begin{proof}
To show (\ref{eq_G4}), by (\ref{eq_app_1}) and the definition of $\tilde H_{01}(\tilde c)$, we have
\begin{align*}
\sum_{(j,k)\in \tilde{H}_{01}}G_j(\alpha)G_k(\alpha)&\geq \sum_{(j,k)\in \tilde{H}_{01}(\tilde c)}G_j(\alpha)G_k(\alpha)\\
&\geq \sum_{(j,k)\in \tilde{H}_{01}(\tilde c)} \Big\{1-\Phi((\sqrt{c}-\sqrt{\tilde c})\sqrt{\log p})\Big\}^2.
\end{align*}
When $c\leq \tilde c$, it holds that $1-\Phi((\sqrt{c}-\sqrt{\tilde c})\sqrt{\log p})\geq 1/2$, and therefore 
$$
\sum_{(j,k)\in \tilde{H}_{01}}G_j(\alpha)G_k(\alpha)\geq |\tilde H_{01}(\tilde c)|/4=\Omega(p^2).
$$ 
Thus, (\ref{eq_G4}) holds with $\xi=2$. When $c> \tilde c$, by using the Gaussian tail inequality $1-\Phi\left(t\right) \geq \sqrt{\frac{2}{{\pi}}}\frac{1}{t+(t^2+4)^{1/2}} \exp \left(-t^{2} / 2\right)$ for any $t>0$, we can show that 
$$
\sum_{(j,k)\in \tilde{H}_{01}(\tilde c)}\Big\{1-\Phi((\sqrt{c}-\sqrt{\tilde c})\sqrt{\log p})\Big\}^2\geq \frac{C}{\log p}p^{2-(\sqrt{c}-\sqrt{\tilde c})^2},
$$
for some constant $C>0$. If (\ref{eq_app_2}) holds, then there exists a constant  $\xi$ such that 
$$
\max \left\{\frac{3}{2}+\frac{\delta}{1+\delta}, 2-\frac{\kappa}{2}+\frac{\delta}{1+\delta}\right\} <\xi<2-(\sqrt{c}-\sqrt{\tilde c})^2.
$$
Thus, $\frac{1}{\log p}p^{2-(\sqrt{c}-\sqrt{\tilde c})^2}=\Omega(p^\xi)$, which implies (\ref{eq_G4}). 
\end{proof}

\section{Proof of Theorem \ref{them1}}\label{app_main}
For notational simplicity, we use $C$ to denote a generic constant, whose value may change from line to line. Since 
$$
\sum_{\{1\leq j<k\leq p: |\hat{T}_{j}| \geq \alpha,|\hat{T}_{k}| \geq \alpha\}}\mathbbm{1}\left\{|\hat{T}_{jk}| \geq t\right\}=\sum_{1 \leq j < k \leq p}\mathbbm{1}\left\{|\hat{T}_{jk}| \geq t,|\hat{T}_{j}| \geq \alpha,|\hat{T}_{k}| \geq \alpha\right\},
$$
we can equivalently define FDP as
$$
\mathrm{FDP}=\frac{\sum_{(j,k)\in \emph{H}_0}\mathbbm{1}\left\{|\widehat{T}_{jk}| \geq \hat{t},|\widehat{T}_j| \geq \alpha,|\widehat{T}_k| \geq \alpha\right\}}{\max \left(\sum_{1 \leq j < k \leq p}\mathbbm{1}\left\{|\widehat{T}_{jk}| \geq \hat{t},|\widehat{T}_{j}| \geq \alpha,|\widehat{T}_{k}| \geq \alpha\right\},1 \right)}.
$$
Noting that
\begin{equation*}
\begin{split}
    \mathrm{FDP}&=\frac{G(\hat{t})M}{\max \left(\sum_{1 \leq j < k \leq p}\mathbbm{1}\left\{|\widehat{T}_{jk}| \geq \hat{t},|\widehat{T}_{j}| \geq \alpha,|\widehat{T}_{k}| \geq \alpha\right\},1 \right)}\\
    &\cdot \frac{\sum_{(j,k)\in \emph{H}_0}\mathbbm{1}\left\{|\widehat{T}_{jk}| \geq \hat{t},|\widehat{T}_j| \geq \alpha,|\widehat{T}_k| \geq \alpha\right\}}{NG(\hat{t})}\cdot \frac{N}{M},
\end{split}
\end{equation*}
to prove Theorem~\ref{them1}, it suffices to show
 \begin{equation*}
\sup \limits_{0 \leq t \leq \sqrt{2\log p}} \left|\frac{\sum_{(j,k)\in \emph{H}_0}\mathbbm{1}\left\{|\widehat{T}_{jk}| \geq t,|\widehat{T}_j| \geq \alpha,|\widehat{T}_k| \geq \alpha\right\}}{NG(t)}-1\right| \rightarrow 0
\end{equation*} 
in probability.
Let $0=t_0<t_1<...<t_m= \sqrt{2\log p}$ satisfy $t_i-t_{i-1}=z_p$ for $1 \leq i \leq m-1 $ and $t_m-t_{m-1}\leq z_p$. Hence $m\sim \sqrt{\log p} /z_p$, which will be specified later. For any $t_{j-1} \leq t \leq t_j$, we have
\begin{equation*}
\frac{\sum_{(j,k)\in \emph{H}_0}\mathbbm{1}\left\{|\widehat{T}_{jk}| \geq t,|\widehat{T}_j| \geq \alpha,|\widehat{T}_k| \geq \alpha\right\}}{NG(t)} \leq \frac{\sum_{(j,k)\in \emph{H}_0}\mathbbm{1}\left\{|\widehat{T}_{jk}| \geq t_{j-1},|\widehat{T}_j| \geq \alpha,|\widehat{T}_k| \geq \alpha\right\}}{NG(t_{j-1})} \cdot \frac{G(t_{j-1})}{G(t_j)}
\end{equation*}
and
\begin{equation*}
\frac{\sum_{(j,k)\in \emph{H}_0}\mathbbm{1}\left\{|\widehat{T}_{jk}| \geq t,|\widehat{T}_j| \geq \alpha,|\widehat{T}_k| \geq \alpha\right\}}{NG(t)} \geq \frac{\sum_{(j,k)\in \emph{H}_0}\mathbbm{1}\left\{|\widehat{T}_{jk}| \geq t_{j},|\widehat{T}_j| \geq \alpha,|\widehat{T}_k| \geq \alpha\right\}}{NG(t_{j})} \cdot \frac{G(t_{j})}{G(t_{j-1})}.
\end{equation*}
Let $z_p = o(1/\sqrt{\log p})$, then $G(t_i)/G(t_{i-1})\rightarrow 1$, and we only need to prove 
\begin{equation}
\label{main1}
\max \limits_{0 \leq r \leq m} \left|\frac{\sum_{(j,k)\in \emph{H}_0}\mathbbm{1}\left\{|\widehat{T}_{jk}| \geq t_r,|\widehat{T}_{j}| \geq \alpha,|\widehat{T}_{k}| \geq \alpha\right\}}{NG(t_r)}-1\right| \rightarrow 0
\end{equation}
in probability.

To show (\ref{main1}), we first prove the following two lemmas that give the nonasymptotic $L_1$-error bound for the MLE estimator. The proofs of the two lemmas are deferred to Appendix \ref{app_secondary}.

\begin{lemma} \label{glm_lemma1}
Let $\epsilon\coloneqq (\epsilon_1,...,\epsilon_n)^T$, and  $\emph{\textbf{X}}\coloneqq(X_1,...,X_n)^T$ be the $n \times s$ design matrix. Assume $\max_{1\leq j \leq s} \max_{1 \leq i \leq n}|X_{ij}| \leq K$, $\mathrm{E}(\epsilon_i^2) \leq \sigma^2$ and $\mathrm{E}(\epsilon_i^4)\leq \kappa^4$. For any $t>0$, define
\[
\lambda_{\epsilon}(t) =  4\sqrt{\frac{8t+6\log(2s)}{n}}K\sigma,
\]
then for all positive $M \leq 1$ and $\beta, \beta_0 \in R^s$, $s \geq 2$ we have
\begin{equation*}
   \mathrm{P}\left(\sup\limits_{\|\beta-\beta_0\|_1\leq M}\left|\epsilon^T\emph{\textbf{X}}(\beta -\beta_0)\right|/n>\lambda_{\epsilon}(t)M\right) \leq 3\exp(-t)+\frac{3\kappa^4}{n\sigma^4}.
\end{equation*}
\end{lemma}

\begin{lemma}\label{glm_lemma2}
 Let $\epsilon\coloneqq (\epsilon_1,...,\epsilon_n)^T$, $\emph{\textbf{X}}\coloneqq(X_1,...,X_n)^T$ be the $n \times s$ design matrix, and  $\hat{\beta}$ be the MLE estimator for (\ref{glm0}). Suppose all conditions in Lemma~\ref{glm_lemma1} hold. Assume there exists constants $K_0$, $C_b$ and $\tau > 0$, such that $\max\limits_{1\leq i \leq n}|X_i^T\beta_0| \leq K_0$, $\tau = \lambda_{\mathrm{min}}(\emph{\textbf{X}}^T\emph{\textbf{X}})/n > 0$ and for all $|z|\leq K + K_0$, $1/C_b\leq b^{\prime\prime}(z)\leq C_b$. Then with probability at least $1-(3\exp(-t)+\frac{3\kappa^4}{n\sigma^4})$,  it holds that 
 \[
 \|\hat{\beta}-\beta_0\|_1 \leq \frac{4\lambda_{\epsilon}(t)sC_b}{\tau},
 \]
 where $\lambda_{\epsilon}(t)$ is defined in Lemma~\ref{glm_lemma1}.
\end{lemma}

By using the above two lemmas, we can show the following lemma that characterizes the difference between the test statistic $\hat T_j$ (and $\hat T_{jk}$) and its linear representation $U_j$ in (\ref{eq_Uj}) (and $U_{ij}$) in a truncated relative error. 

\begin{lemma} \label{UTlemma}
Suppose Assumptions~\ref{aspA1}-\ref{aspA4} hold, then we have for any constant $c>0$,
\[
 \max\limits_{1\leq j \leq p}\left|\frac{\widehat{T}_j-U_j}{|U_j|\vee c}\right|= 
 O_p(\frac{\log p}{\sqrt{n}}),~~~ \max\limits_{1\leq j< k \leq p}\left|\frac{\widehat{T}_{jk}-U_{jk}}{|U_{jk}|\vee c}\right|= 
 O_p(\frac{\log p}{\sqrt{n}}).
\]
\end{lemma}

To proceed, for any $1 \leq j<k \leq p$, we can write 
\begin{align*}
    \begin{split}
        \mathbbm{1}\left\{|\widehat{T}_{jk}| \geq t\right\}&=
         \mathbbm{1}\left\{|U_{jk}+\widehat{T}_{jk}-U_{jk}| \geq t\right\}\\
         &= \mathbbm{1}\left\{U_{jk} \geq t -(\widehat{T}_{jk}-U_{jk})\right\}
         +\mathbbm{1}\left\{U_{jk} \leq -t -(\widehat{T}_{jk}-U_{jk})\right\}.
    \end{split}
\end{align*}
If $|U_j| >c$, where $c$ is given by Lemma~\ref{UTlemma}, we have
\[
|\hat{T}_{jk} - U_{jk}| = |U_{jk}|O_p\left(\frac{\log p}{\sqrt{n}}\right),
\]
and it yields
\[
 \mathbbm{1}\left\{|\widehat{T}_{jk}| \geq t\right\}
 = \mathbbm{1}\left\{U_{jk} \geq t\left(1+O_p\left(\frac{\log p}{\sqrt{n}}\right)\right)\right\}
         +\mathbbm{1}\left\{U_{jk}\leq -t\left(1+O_p\left(\frac{\log p}{\sqrt{n}}\right)\right)\right\}.
\]
Hence for any $0 \leq t \leq \sqrt{2\log p}$, there exist $\tilde{t}_1,\tilde{t}_2 = t +o_p(\frac{1}{\sqrt{\log p}})$  such that
\begin{equation}
\label{two-side}
     \mathbbm{1}\left\{|\widehat{U}_{jk}| \geq \tilde{t}_1\right\}
 \leq
 \mathbbm{1}\left\{|\widehat{T}_{jk}| \geq t\right\}
 \leq
  \mathbbm{1}\left\{|\widehat{U}_{jk}| \geq \tilde{t}_2\right\}.
\end{equation}
Similarly, if $|U_j| \leq c$, we have
\[
|\hat{T}_{jk} - U_{jk}| = O_p\left(\frac{\log p}{\sqrt{n}}\right),
\]
and (\ref{two-side}) still holds. Similarly, (\ref{two-side}) also holds for $\widehat{T}_{j}$ and $\widehat{T}_{k}$. 

Recall that the goal is to show (\ref{main1}). From (\ref{two-side}), it suffices to show
\begin{equation}
\label{rednew}
    \max \limits_{0 \leq r \leq m} \left|\frac{\sum_{(j,k)\in \emph{H}_0}\mathbbm{1}\left\{|U_{jk}| \geq t_r,|U_{j}| \geq \alpha,|U_{k}| \geq \alpha \right\}}{\sum_{(j,k)\in\emph{H}_0} \mathbbm{1} \left\{|U_{j}| \geq \alpha,|U_{k}| \geq \alpha \right\}G(t_r)
}
-1\right| 
\rightarrow 0
\end{equation}
in probability. Define 
$$
J=\sum_{(j,k)\in\emph{H}_0} \mathbbm{1} \left\{|U_{j}| \geq \alpha,|U_{k}| \geq \alpha \right\}.
$$
Write
\begin{align*}
 & \frac{\sum_{(j,k)\in \emph{H}_0}\mathbbm{1}\left\{|U_{jk}| \geq t_r,|U_{j}| \geq \alpha,|U_{k}| \geq \alpha \right\}}{\sum_{(j,k)\in\emph{H}_0} \mathbbm{1} \left\{|U_{j}| \geq \alpha,|U_{k}| \geq \alpha \right\}G(t_r)
}-1\\
&= 
\frac{\sum_{(j,k)\in \emph{H}_0}\mathbbm{1}\left\{|U_{jk}| \geq t_r,|U_{j}| \geq \alpha,|U_{k}| \geq \alpha \right\}-G(t_r)J}{\mathrm{E}(J)G(t_r)
}\cdot \frac{\mathrm{E}(J)}{J}.
\end{align*}
For any $0 \leq r \leq m$, denote 
\[
\begin{split}
J_{r,1}&=\left|\frac{\sum_{(j,k)\in \emph{H}_0}\left(\mathbbm{1}\left\{|U_{jk}| \geq t_r,|U_{j}| \geq \alpha,|U_{k}| \geq \alpha\right\}-\mathrm{P}\left\{|U_{jk}| \geq t_r,|U_{j}| \geq \alpha,|U_{k}| \geq \alpha\right\}\right)}{\mathrm{E}(J)G(t_r)}\right|,\\
 J_{r,2}&=\left|\frac{\sum_{(j,k)\in \emph{H}_0}\mathrm{P}\left\{|U_{jk}| \geq t_r,|U_{j}| \geq \alpha,|U_{k}| \geq \alpha\right\}-G(t_r)J}{\mathrm{E}(J)G(t_r)}\right|,
	\end{split}
\]
Then we have
\begin{equation}\label{eq_mainproof_1}
   \left|\frac{\sum_{(j,k)\in \emph{H}_0}\mathbbm{1}\left\{|U_{jk}| \geq t_r,|U_{j}| \geq \alpha,|U_{k}| \geq \alpha \right\}}{\sum_{(j,k)\in\emph{H}_0} \mathbbm{1} \left\{|U_{j}| \geq \alpha,|U_{k}| \geq \alpha \right\}G(t_r)
}
-1\right|\leq (J_{r,1}+J_{r,2})\cdot \frac{\mathrm{E}(J)}{J}. 
\end{equation}
Using the following sequence of lemmas, we will show that 
\begin{equation}\label{eq_mainproof_2}
\max \limits_{0 \leq r \leq m}J_{r,1} \rightarrow 0,~~~\max \limits_{0 \leq r \leq m}J_{r,2} \rightarrow 0
\end{equation}
in probability separately. Also, we will show that for some constant $C$, 
\begin{equation}\label{eq_mainproof_3}
\frac{\mathrm{E}(J)}{J}=O_p(1).
\end{equation}
Together with (\ref{eq_mainproof_1}), we get (\ref{rednew}), and therefore we complete the proof.

To prove (\ref{eq_mainproof_2}) and (\ref{eq_mainproof_3}), we need the following Lemma \ref{lemma4}, which is the Cramér type moderate deviation bound in our setting. 
\begin{lemma}
\label{lemma4}
Suppose Assumption \ref{aspA2} holds. Recall that
\begin{equation*}
    \tilde{U}_{ijkml} = (U_{ijk}, S_{ij}, S_{ik}, U_{iml}, S_{im}, S_{il})\in\RR^6.
\end{equation*}
For any $1 \leq j < k \leq p$, $1 \leq m < l \leq p$, where $(j,k)\in\emph{H}_0$, $(m,l)\in\emph{H}_0$ and for some constant $\gamma >0$, assume  $\|\textbf{cov}(\tilde{U}_{ijkml})- \textbf{I}\|_{\infty} \leq C(\log p)^{-2-\gamma}$. Denote
\[
\mathrm{P}_{jk}=\mathrm{P}\left(|U_{jk}| \geq t,|U_{j}| \geq \alpha,|U_{k}| \geq \alpha\right),
\]
\[
\mathrm{P}_{jkml}=\mathrm{P}\left(|U_{jk}| \geq t,|U_{j}| \geq \alpha,|U_{k}| \geq \alpha, |U_{ml}| \geq t,|U_{m}| \geq \alpha,|U_{l}| \geq \alpha\right).
\]
Then we have
\begin{equation}\label{eq_lemma4_1}
    \sup\limits_{0 \leq t \leq  \sqrt{2\log p}}\left|\frac{\mathrm{P}_{jk}}{G(t)G_{j}(\alpha)G_k(\alpha)}-1\right| \leq C (\log p)^{-1-\gamma_1},
\end{equation}
and
\begin{equation}\label{eq_lemma4_2}
    \sup\limits_{0 \leq t \leq  \sqrt{2\log p}}\left|\frac{\mathrm{P}_{jkml}}{G^2(t)G_{j}(\alpha)G_k(\alpha)G_{m}(\alpha)G_l(\alpha)}-1\right| \leq C (\log p)^{-1-\gamma_1},
\end{equation}
where $\gamma_1 = \min\{\gamma, 1/2\}$ and $G_j(\alpha) \coloneqq \mathrm{P}\left(|\mathcal{N}(0,1)+\sqrt{n}\mathrm{E}(U_{ij})|\geq \alpha\right)$.
\end{lemma}

The following Lemma \ref{EJ} is also an intermediate step to prove (\ref{eq_mainproof_2}) and (\ref{eq_mainproof_3}). 

\begin{lemma}
\label{EJ}
Under Assumptions \ref{aspA5} and \ref{aspA6}, we have 
\[
\mathrm{E}(J) = \sum_{(j,k)\in\emph{H}_0}\mathrm{P}\left\{|U_{j}| \geq \alpha,|U_{k}| \geq \alpha\right\} = \Omega(p^{\xi}).
\]
\end{lemma}

Now, we are ready to prove (\ref{eq_mainproof_2}) and (\ref{eq_mainproof_3}), which are shown in Lemmas \ref{lemma6} and \ref{lemma7}.

\begin{lemma}\label{lemma6}
Under Assumptions \ref{aspA1}-\ref{aspA6}, we have 
\begin{equation*}
    \max \limits_{0 \leq r \leq m}J_{r,1} \rightarrow 0
\end{equation*}
in probability. 
\end{lemma}

\begin{lemma}\label{lemma7}
Under Assumptions \ref{aspA1}-\ref{aspA6}, we have 
\begin{equation*}
    \max \limits_{0 \leq r \leq m}J_{r,2} =  \max \limits_{0 \leq r \leq m}\left|\frac{\sum_{(j,k)\in \emph{H}_0}\mathrm{P}\left\{|U_{jk}| \geq t_r,|U_{j}| \geq \alpha,|U_{k}| \geq \alpha\right\}-G(t_r)J}{\mathrm{E}(J)G(t_r)}\right|\rightarrow 0
\end{equation*}
in probability.
\end{lemma}

Finally, the following last lemma shows (\ref{eq_mainproof_3}), which completes the proof.
\begin{lemma}\label{lem_J}
Recall that
$J=\sum_{(j,k)\in\emph{H}_0} I \left\{|U_{j}| \geq \alpha,|U_{k}| \geq \alpha \right\}$,
Assume \ref{aspA1}-\ref{aspA6} hold, then for some constant $C$, we have
\begin{equation*}
\mathrm{P}\left(\frac{\mathrm{E}(J)}{J} \leq C \right) \rightarrow 1.
\end{equation*}
\end{lemma}

\section{Proof of Theorem \ref{thm_power}}\label{app_power}

Define $t^* = c^*\sqrt{\log p}$ and the event $E=\{\inf_{(j,k)\in H_1}|\hat{T}_{jk}| \geq t^*\}$. Form the Gaussian tail bound, it can be shown that $c^*>0$ and is upper bounded by a constant. The following lemma shows that the event $E$ holds with probability tending to 1. 
\begin{lemma}\label{lem_power}
Under the same conditions in Theorem \ref{thm_power}, the event $E$ holds with probability tending to 1.
\end{lemma}
So, under the event $E$, 
\begin{equation}
    \label{pow1}
    \sum_{(j,k)\in H_1}\mathbbm{1}\left\{|\hat{T}_{jk}| \geq t^*,|\hat{T}_{j}| \geq \alpha,|\hat{T}_{k}| \geq \alpha\right\}=\sum_{(j,k)\in H_1}\mathbbm{1}\left\{|\hat{T}_{j}| \geq \alpha,|\hat{T}_{k}| \geq \alpha\right\}=M-N.
\end{equation}
Recall that our FDR threshold is given by
\[
\hat{t}=\inf \left\{
0 \leq \emph{t} \leq \sqrt{ 2\log p}: G(t)\leq \eta \frac{\max \left(\sum_{1 \leq j < k \leq p}\mathbbm{1}\left\{|\hat{T}_{jk}| \geq t,|\hat{T}_{j}| \geq \alpha,|\hat{T}_{k}| \geq \alpha\right\},1 \right)}{\sum_{1 \leq j < k \leq p}\mathbbm{1}\left\{|\hat{T}_{j}| \geq \alpha,|\hat{T}_{k}| \geq \alpha\right\}}
\right\}.
\]
If we consider $t=t^*$, then
\[
\eta\frac{\sum_{1 \leq j < k \leq p}\mathbbm{1}\left\{|\hat{T}_{jk}| \geq t^*,|\hat{T}_{j}| \geq \alpha,|\hat{T}_{k}| \geq \alpha\right\}}{\sum_{1 \leq j < k \leq p}\mathbbm{1}\left\{|\hat{T}_{j}| \geq \alpha,|\hat{T}_{k}| \geq \alpha\right\}} \geq  \eta\frac{M-N}{M} = G(t^*),
\]
where the last equality is from the definition of $t^*$. This implies $\hat t\leq t^*$. Since the event $E$ holds with probability tending to one, we have $\mathrm{P}(\hat E) \rightarrow 1$, where $\hat E=\{\inf_{(j,k)\in H_1}|\hat{T}_{jk}| \geq \hat{t}\}$. Thus, (\ref{pow1}) also works if $t^*$ is replaced with $\hat{t}$. Thus, with probability tending to one
$$
\frac{\sum_{(j,k)\in \emph{H}_1}\mathbbm{1}\left\{|\widehat{T}_{jk}| \geq \hat{t},|\widehat{T}_j| \geq \alpha,|\widehat{T}_k| \geq \alpha\right\}}{|\emph{H}_1|}=\frac{M-N}{|\emph{H}_1|}.
$$
Since the power is given by the expectation of the term in the left hand side of the above equation, we have 
\begin{align*}
\mathrm{power}&\geq \mathrm{E}\Big\{\frac{\sum_{(j,k)\in \emph{H}_1}\mathbbm{1}\left\{|\widehat{T}_{jk}| \geq \hat{t},|\widehat{T}_j| \geq \alpha,|\widehat{T}_k| \geq \alpha\right\}}{|\emph{H}_1|}I(\hat E)\Big\}\\
&=\mathrm{E}\Big\{\frac{M-N}{|\emph{H}_1|}I(\hat E)\Big\}= \mathrm{E}\Big\{\frac{M-N}{|\emph{H}_1|}\Big\}-\mathrm{E}\Big\{\frac{M-N}{|\emph{H}_1|}I(\hat E^c)\Big\}.
\end{align*}
Since $M-N\leq |H_1|$, we have $\mathrm{E}\{\frac{M-N}{|\emph{H}_1|}I(\hat E^c)\}\leq \mathrm{P}(\hat E^c)=o(1)$. Thus,
$$
\mathrm{power}\geq \mathrm{E}\Big\{\frac{M-N}{|\emph{H}_1|}\Big\}-o(1).
$$
Similarly, we can show that $\mathrm{power}\leq \mathrm{E}\{\frac{M-N}{|\emph{H}_1|}\}$. This completes the proof.

\section{Proof of Additional Lemmas}\label{app_secondary}

\subsection{Proof of Lemma \ref{glm_lemma1}}
Let $\gamma_1,...,\gamma_n$ be a Rademacher sequence independent of $\epsilon$ and $\emph{\textbf{X}}$. Let
\[
\textbf{Z}\coloneqq\sup\limits_{\|\beta-\beta_0\|_1\leq M}\frac{1}{n}\left|\epsilon^T\emph{\textbf{X}}(\beta -\beta_0)\right| = \sup\limits_{\|\beta-\beta_0\|_1\leq M}\frac{1}{n}\left|\sum\limits_{i=1}^n \epsilon_iX_i^T(\beta-\beta_0)\right|,
\]
\[
\textbf{Z}_{\gamma}\coloneqq\sup\limits_{\|\beta-\beta_0\|_1\leq M}\frac{1}{n}\left|\sum\limits_{i=1}^n\gamma_i\epsilon_iX_i^T(\beta- \beta_0)\right| \coloneqq \sup\limits_{\|\beta-\beta_0\|_1\leq M}\frac{1}{n}\left|\sum\limits_{i=1}^nz_i\right|.
\]
Denote $\mathrm{P}_{\epsilon,x}$ and $\mathrm{E}_{\epsilon,x}$ as the conditional probability and expectation given $\epsilon$ and $\emph{\textbf{X}}$. Note that $|z_i| \leq MK|\epsilon_i|$. By H\"{o}lder inequality and Nemirovski moment inequality in \cite{buhlmann2011statistics}, we have
\[
 \mathrm{E}_{\epsilon,x}(\textbf{Z}_{\gamma}) \leq M\mathrm{E}_{\epsilon,x}\max _{1 \leq j \leq s}\left|\sum_{i=1}^{n} \gamma_iX_{ij}\epsilon_{i}  / n\right| \leq
 MK\left(\frac{8 \log (2 s)}{n}\right)^{1/2}\left(\frac{1}{n}\sum\limits_{i=1}^n\epsilon_i^2\right)^{1/2}.
\]
Let $v=\mathrm{E}_{\epsilon,x}\sum\limits_{i=1}^n(2z_i)^2 \leq 4M^2K^2\sum\limits_{i=1}^n\epsilon_i^2$. Then Massart's inequality implies that for any positive $t$, 
\begin{align*}
    \begin{split}
        \exp \left(-t\right) &\geq
\mathrm{P}_{\epsilon,x}\left(\textbf{Z}_{\gamma} \geq \mathrm{E}_{\epsilon,x}(\textbf{Z}_{\gamma}) +2\sqrt{vt}/n+\frac{5MK}{2
n}\max\limits_{1\leq i \leq n}|\epsilon_i|\right)\\
&\geq 
\mathrm{P}_{\epsilon,x}\left(\textbf{Z}_{\gamma} \geq  MK\left(\frac{8 \log (2 s)}{n}\right)^{1/2}\left(\frac{1}{n}\sum\limits_{i=1}^n\epsilon_i^2\right)^{1/2}
+
4MK\sqrt{\frac{t}{n}}\left(\frac{1}{n}\sum\limits_{i=1}^n\epsilon_i^2\right)^{1/2}
+
\frac{5MK}{2
n}\max\limits_{1\leq i \leq n}|\epsilon_i|\right) \\
&=  
\mathrm{P}_{\epsilon,x}\left(\textbf{Z}_{\gamma} \geq  \left(\sqrt{\frac{8 \log (2 s)}{n}}+4\sqrt{\frac{t}{n}}\right)MK\left(\frac{1}{n}\sum\limits_{i=1}^n\epsilon_i^2\right)^{1/2}
+
\frac{5MK}{2
n}\max\limits_{1\leq i \leq n}|\epsilon_i|\right)\\
& \geq 
\mathrm{P}_{\epsilon,x}\left(\textbf{Z}_{\gamma} \geq  4\sqrt{\frac{\log(2s)+2t}{n}}MK\left(\frac{1}{n}\sum\limits_{i=1}^n\epsilon_i^2\right)^{1/2}
+
\frac{5MK}{2
n}\max\limits_{1\leq i \leq n}|\epsilon_i|\right)\\
& \geq
\mathrm{P}_{\epsilon,x}\left(\textbf{Z}_{\gamma} \geq  4\sqrt{\frac{\log(2s)+2t}{n}}MK\left(\frac{1}{n}\sum\limits_{i=1}^n\epsilon_i^2\right)^{1/2}
+
\frac{5MK}{2
}\sqrt{\frac{1}{n}}\left(\frac{1}{n}\sum\limits_{i=1}^n\epsilon_i^2\right)^{1/2}\right)\\
& \geq
\mathrm{P}_{\epsilon,x}\left(\textbf{Z}_{\gamma} \geq  4\sqrt{\frac{3\log(2s)+4t}{n}}MK\left(\frac{1}{n}\sum\limits_{i=1}^n\epsilon_i^2\right)^{1/2}
\right).
    \end{split}
\end{align*}
Then we integrate out $\epsilon$ and $\emph{\textbf{X}}$ and have
\begin{align*}
    \begin{split}
        \mathrm{P}\left(\textbf{Z}_{\gamma} \geq 4\sqrt{\frac{8t+6\log(2s)}{n}}MK\sigma\right) 
        &\leq \exp(-t) +
\mathrm{P}\left(\frac{1}{n}\sum\limits_{i=1}^n \epsilon_i^2>2\sigma^2\right)\\
&\leq \exp(-t) + \frac{\kappa^4}{n\sigma^4},
    \end{split}
\end{align*}
where the last step holds by Markov inequality. Note that 
\[
\mathrm{P}\left(\textbf{Z} \geq u\right) \leq 
\frac{2\mathrm{P}\left(\textbf{Z}_{\gamma} \geq u/4\right) }{1-4\sigma^2M^2K^2/nu^2}.
\]
With $u = 16\sqrt{\frac{8t+6\log(2s)}{n}}MK\sigma$ and $\log(2s)\geq 1 $, we can derive the result.

\subsection{Proof of Lemma \ref{glm_lemma2}}
Take $M= \frac{4\lambda_{\epsilon}(t)sC_b}{\tau}$ and assume we are on the set
\[
\gamma \coloneqq \left\{\sup\limits_{\|\beta-\beta_0\|_1\leq M}\left|\epsilon^T\emph{\textbf{X}}(\beta -\beta_0)\right|/n\leq \lambda_{\epsilon}(t)M\right\}.
\]
Let 
\[
t \coloneqq \frac{M}{M+\|\hat{\beta}-\beta_0\|_1},
\]
and 
\[
\tilde{\beta}_t = t\hat{\beta}+(1-t)\beta_0.
\]
Then 
\[
\|\tilde{\beta}_t - \beta_0\|_1=\frac{M\|\hat{\beta}-\beta_0\|_1}{M + \|\hat{\beta}-\beta_0\|_1 }.
\]
Note that if we can show $\|\tilde{\beta}_t - \beta_0\|_1 \leq M/2$, then from the above display we get $\|\hat{\beta}-\beta_0\|_1 \leq M$. So, in the following, we focus on the proof of $\|\tilde{\beta}_t - \beta_0\|_1 \leq M/2$. By the convexity of the negative log-likelihood function, we have
\begin{align*}
    \begin{split}
        \frac{1}{n}\sum\limits_{i=1}^n\left(Y_iX_i^T\tilde{\beta}_t-b(X_i^T\tilde{\beta}_t )\right)
        &\geq t\left\{\frac{1}{n}\sum\limits_{i=1}^n\left(Y_iX_i^T\hat{\beta}-b(X_i^T\hat{\beta} )\right)\right\}+(1-t)\left\{\frac{1}{n}\sum\limits_{i=1}^nY_iX_i^T\beta_0-b(X_i^T\beta_0 )\right\}\\
        &\geq \frac{1}{n}\sum\limits_{i=1}^n\left(Y_iX_i^T\beta_0-b(X_i^T\beta_0 )\right).
    \end{split}
\end{align*}
Note that $\epsilon_i= Y_i - b^{\prime}(X_i^T\beta_0)$, hence we can write this as
\[
\frac{1}{n}\sum\limits_{i=1}^n\left( b(X_i^T\tilde{\beta}_t )-b(X_i^T\beta_0 )-b^{\prime}(X_i^T\beta_0 )X_i^T(\tilde{\beta}_t-\beta_0)\right)
\leq \frac{1}{n}\epsilon^T\emph{\textbf{X}}(\tilde{\beta}_t-\beta_0).
\]
By Taylor expansion and assumption $1/C_b\leq b^{\prime\prime}(z)$, we have 
\[
\frac{1}{n}\sum\limits_{i=1}^n\left( b(X_i^T\tilde{\beta}_t )-b(X_i^T\beta_0 )-b^{\prime}(X_i^T\beta_0 )X_i^T(\tilde{\beta}_t-\beta_0)\right)\geq \frac{1}{C_b}\frac{\|\emph{\textbf{X}}(\tilde{\beta}_t-\beta_0)\|_2^2}{n}.
\]
Note that by the definition we have $\|\tilde{\beta}_t - \beta_0\|_1 < M$. Thus, on the set $\gamma$, we have
\[
\frac{\|\emph{\textbf{X}}(\tilde{\beta}_t-\beta_0)\|_2^2}{nC_b}  \leq \frac{1}{n}\epsilon^T\emph{\textbf{X}}(\tilde{\beta}_t-\beta_0)\leq 
\lambda_{\epsilon}(t)M.
\]
We now use the proof by contradiction. If $\|\tilde{\beta}_t - \beta_0\|_1 > M/2$, then 
\[
\frac{\tau\|\tilde{\beta}_t - \beta_0\|_1^2}{C_bs}
\leq
\frac{\|\emph{\textbf{X}}(\tilde{\beta}_t-\beta_0)\|_2^2}{nC_b} \leq 
\lambda_{\epsilon}(t)M < 2\lambda_{\epsilon}(t)\|\tilde{\beta}_t - \beta_0\|_1,
\]
hence
\[
\|\tilde{\beta}_t - \beta_0\|_1 < \frac{2\lambda_{\epsilon}(t)C_bs}{\tau}=\frac{M}{2},
\]
which leads to a contradiction. Therefore we must have $\|\tilde{\beta}_t - \beta_0\|_1 \leq M/2$ which further concludes that $\|\hat{\beta}-\beta_0\|_1 \leq M$.

\subsection{Proof of Lemma \ref{UTlemma}}
Consider $\hat{T}_j$ and $U_j$ first. 
By Lemma~\ref{glm_lemma2} we have for some $t>0$, 
\[
\alpha(t)=3\exp(-t)+\frac{3\kappa^4}{n\sigma^4},~~~
\lambda_{\epsilon}(t) =  4\sqrt{\frac{8t+6\log4}{n}}K\sigma,
\]
with probability at least $1-\alpha(t)$, it holds that
\[
 \|\hat{\beta}_{}^j-\beta_0^j\|_1 \leq \frac{8\lambda_{\epsilon}(t)C_b}{\tau}.
\]
For some $\tilde{\beta}^j$ on the line segment between $\hat{\beta}_{}^j$ and $\beta_0^j$, by Taylor expansion we have
\begin{align}\label{taylor1}
\begin{split}
         0&= \Psi_n({\beta_0^{j}},\emph{\textbf{X}}_{j}^{\mathrm{s1}},Y)+\left(\frac{\partial \Psi_n({\tilde{\beta}^j},\emph{\textbf{X}}_{j}^{\mathrm{s1}},Y)}{\partial \beta}\right)^T(\hat{\beta}_{}^j-\beta_0^j)\\
     &=  \Psi_n({\beta_0^{j}},\emph{\textbf{X}}_{j}^{\mathrm{s1}},Y) -
     \frac{1}{n}\sum\limits_{i=1}^nb^{\prime\prime}(\emph{\textbf{X}}_{ij}^{\mathrm{s1}}\tilde{\beta}^j)(\emph{\textbf{X}}_{ij}^{\mathrm{s1}})^T\emph{\textbf{X}}_{ij}^{\mathrm{s1}}
     (\hat{\beta}_{}^j-\beta_0^j).
\end{split}
\end{align}
For some $c>0$, let $t = (2+c)\log p$, then $\alpha(t) = O(1/p^{2+c})$, $\lambda_{\epsilon}(t)= O(\sqrt{\frac{\log p}{n}})$. Applying union bound for all $1\leq j \leq p$, we have
\[
\max\limits_{1\leq j \leq p}\|\tilde{\beta}^j-\beta_0^j\|_1 \leq \max\limits_{1\leq j \leq p}\|\hat{\beta}_{}^j-\beta_0^j\|_1 = O_p\left(\sqrt{\frac{\log p}{n}}\right).
\]
For some $\beta^j_*$ on the line segment between $\tilde{\beta}^j$ and $\beta_0^j$, we have
\[
\max\limits_{1\leq j \leq p}\left|b^{\prime\prime}(\emph{\textbf{X}}_{ij}^{\mathrm{s1}}\tilde{\beta}^j)
-
b^{\prime\prime}(\emph{\textbf{X}}_{ij}^{\mathrm{s1}}\beta_0^j)\right|=
\max\limits_{1\leq j \leq p}\left|b^{\prime\prime\prime}(\emph{\textbf{X}}_{ij}^{\mathrm{s1}}\beta_*^j)
\emph{\textbf{X}}_{ij}^{\mathrm{s1}}(\tilde{\beta}^j-\beta_0^j)\right| \leq C_{\tilde{b}}K
\max\limits_{1\leq j \leq p}\|\tilde{\beta}^j-\beta_0^j\|_1 ,
\]
therefore it follows that
\[
\max\limits_{1\leq j \leq p}\left\|\frac{1}{n}\sum\limits_{i=1}^nb^{\prime\prime}(\emph{\textbf{X}}_{ij}^{\mathrm{s1}}\tilde{\beta}^j)(\emph{\textbf{X}}_{ij}^{\mathrm{s1}})^T\emph{\textbf{X}}_{ij}^{\mathrm{s1}}-
\frac{1}{n}\sum\limits_{i=1}^nb^{\prime\prime}(\emph{\textbf{X}}_{ij}^{\mathrm{s1}}\beta_0^j)(\emph{\textbf{X}}_{ij}^{\mathrm{s1}})^T\emph{\textbf{X}}_{ij}^{\mathrm{s1}}\right\|_{\infty} =O_p\left(\sqrt{\frac{\log p}{n}}\right).
\]
The Nemirovski moment inequality implies
\[
\max\limits_{1\leq j \leq p}\left\|
\frac{1}{n}\sum\limits_{i=1}^nb^{\prime\prime}(\emph{\textbf{X}}_{ij}^{\mathrm{s1}}\beta_0^j)(\emph{\textbf{X}}_{ij}^{\mathrm{s1}})^T\emph{\textbf{X}}_{ij}^{\mathrm{s1}}
-\mathrm{E}\left(b^{\prime\prime}(\emph{\textbf{X}}_{ij}^{\mathrm{s1}}\beta_0^j)(\emph{\textbf{X}}_{ij}^{\mathrm{s1}})^T\emph{\textbf{X}}_{ij}^{\mathrm{s1}}\right)\right\|_{\infty}= O_p\left(\sqrt{\frac{\log p}{n}}\right),
\]
and hence
\begin{equation}\label{bound1}
    \max\limits_{1\leq j \leq p}\left\|
\frac{1}{n}\sum\limits_{i=1}^nb^{\prime\prime}(\emph{\textbf{X}}_{ij}^{\mathrm{s1}}\tilde{\beta}^j)(\emph{\textbf{X}}_{ij}^{\mathrm{s1}})^T\emph{\textbf{X}}_{ij}^{\mathrm{s1}}-
\mathrm{E}\left(b^{\prime\prime}(\emph{\textbf{X}}_{ij}^{\mathrm{s1}}\beta_0^j)(\emph{\textbf{X}}_{ij}^{\mathrm{s1}})^T\emph{\textbf{X}}_{ij}^{\mathrm{s1}}\right)\right\|_{\infty}= O_p\left(\sqrt{\frac{\log p}{n}}\right).
\end{equation}
Following the same proof, we can also obtain
\begin{align}\label{bound3}
    \begin{split}
        \max\limits_{1\leq j \leq p}\left\|
        \mathrm{E}\left(b^{\prime\prime}(\emph{\textbf{X}}_{ij}^{\mathrm{s1}}\beta_0^j)(\emph{\textbf{X}}_{ij}^{\mathrm{s1}})^T\emph{\textbf{X}}_{ij}^{\mathrm{s1}}\right)
        -\frac{1}{n}\sum\limits_{i=1}^nb^{\prime\prime}\left(\emph{\textbf{X}}_{ij}^{\mathrm{s1}}\hat{\beta}_{}^{j}\right)(\emph{\textbf{X}}_{ij}^{\mathrm{s1}})^T\emph{\textbf{X}}_{ij}^{\mathrm{s1}}
        \right\|_{\infty}
        = O_p\left(\sqrt{\frac{\log p}{n}}\right).
    \end{split}
\end{align}
Recall that
\[
\Psi_{n}(\beta_0^j,\emph{\textbf{X}}_{j}^{\mathrm{s1}},Y)=\frac{1}{n} \sum_{i=1}^{n} \psi_{\beta}\left((\emph{\textbf{X}}_{ij}^{\mathrm{s1}})^T, Y_{i}\right)=\frac{1}{n} \sum_{i=1}^{n}\left\{Y_{i}-b^{\prime}\left(\emph{\textbf{X}}_{ij}^{\mathrm{s1}} \beta_0^j\right)\right\} \cdot (\emph{\textbf{X}}_{ij}^{\mathrm{s1}})^T=\frac{1}{n} \sum_{i=1}^{n}\epsilon_{ij}(\emph{\textbf{X}}_{ij}^{\mathrm{s1}})^T.
\]
Since $\max_{1\leq j \leq p}\mathrm{E}\epsilon_{ij}^2$ is finite and $\max_{1\leq j \leq p} \max_{1 \leq i \leq n}|X_{ij}| \leq K$, we can again apply  Nemirovski moment inequality to get
\begin{equation}\label{bound2}
    \max\limits_{1\leq j \leq p}\left\|\Psi_{n}(\beta_0^j,\emph{\textbf{X}}_{j}^{\mathrm{s1}},Y)
\right\|_1 = O_p\left(\sqrt{\frac{\log p }{n}}\right).
\end{equation}
From (\ref{taylor1}) together with (\ref{bound1}) and (\ref{bound2}), we have
\begin{align}\label{bound4}
        &\max\limits_{1\leq j \leq p}\|\sqrt{n}(\hat{\beta}_{}^{j}-\beta_0^j)-
u({\beta_0^{j}},\emph{\textbf{X}}_{j}^{\mathrm{s1}},Y)
\|_1 \nonumber\\
&= 
\max\limits_{1\leq j \leq p}\left\|\left\{\left[\frac{1}{n}\sum\limits_{i=1}^nb^{\prime\prime}(\emph{\textbf{X}}_{ij}^{\mathrm{s1}}\tilde{\beta}^j)(\emph{\textbf{X}}_{ij}^{\mathrm{s1}})^T\emph{\textbf{X}}_{ij}^{\mathrm{s1}}\right]^{-1} - \left[\mathrm{E}(b^{\prime\prime}(\emph{\textbf{X}}_{ij}^{\mathrm{s1}}\beta_0^j)(\emph{\textbf{X}}_{ij}^{\mathrm{s1}})^T\emph{\textbf{X}}_{ij}^{\mathrm{s1}})\right]^{-1}\right\}n^{1/2}\Psi_n({\beta_0^{j}},\emph{\textbf{X}}_{j}^{\mathrm{s1}},Y)\right\|_1\nonumber\\
&= O_p\left(\frac{\log p}{\sqrt{n}}\right).
\end{align}
Similar to the proof of (\ref{bound3}), we can also show that
\begin{equation}\label{bound6}
\max\limits_{1\leq j \leq p}\left\|\mathrm{\textbf{cov}}( u({\beta_0^{j}},\emph{\textbf{X}}_{j}^{\mathrm{s1}},Y))
- 
\hat{\mathrm{\textbf{cov}}}(\hat{\beta}_{}^{j})
\right\|_{\infty}
= O_p\left(\sqrt{\frac{\log p }{n}}\right).
\end{equation}
Then 
\begin{align*}
        & \max\limits_{1\leq j \leq p}\left|\frac{\widehat{T}_j-U_j}{U_j\vee c}\right|
         =  \max\limits_{1\leq j \leq p}
         \left|
        \left( \frac{\widehat{T}_j}{U_j} - 1\right)\frac{U_j}{U_j\vee c}  \right|\\
        &=\max\limits_{1\leq j \leq p}\left|  \left( \frac{\sqrt{n}\left(\hat{\beta}_{}^{j}\right)_{(2)}}{u({\beta_0^{j}},\emph{\textbf{X}}_{j}^{\mathrm{s1}},Y)_{(2)}+\sqrt{n}(\beta_0^{j})_{(2)}}
             \frac{\sqrt{\mathrm{\textbf{cov}}( u({\beta_0^{j}},\emph{\textbf{X}}_{j}^{\mathrm{s1}},Y))_{(2,2)}}}{\sqrt{\hat{\mathrm{\textbf{cov}}}(\hat{\beta}_{}^{j})_{(2,2)}}}
             -1\right)\frac{U_j}{U_j\vee c}\right|\\
             &\leq
        \max\limits_{1\leq j \leq p}
        \left|\frac{\left(\sqrt{n}\left(\hat{\beta}_{}^{j}\right)_{(2)}
               - \left(u({\beta_0^{j}},\emph{\textbf{X}}_{j}^{\mathrm{s1}},Y)_{(2)}+\sqrt{n}(\beta_0^{j})_{(2)}\right)
             \right)
             \sqrt{\mathrm{\textbf{cov}}( u({\beta_0^{j}},\emph{\textbf{X}}_{j}^{\mathrm{s1}},Y))_{(2,2)}}
             }
             {\left(u({\beta_0^{j}},\emph{\textbf{X}}_{j}^{\mathrm{s1}},Y)_{(2)}+\sqrt{n}(\beta_0^{j})_{(2)}\right)\sqrt{\hat{\mathrm{\textbf{cov}}}(\hat{\beta}_{}^{j})_{(2,2)}}}\frac{U_j}{U_j\vee c}\right|\\
             &~~~~~+\max\limits_{1\leq j \leq p}
             \left|
             \frac{\sqrt{\hat{\mathrm{\textbf{cov}}}(\hat{\beta}_{}^{j})_{(2,2)}}- \sqrt{\mathrm{\textbf{cov}}( u({\beta_0^{j}},\emph{\textbf{X}}_{j}^{\mathrm{s1}},Y))_{(2,2)}}}
             {\sqrt{\hat{\mathrm{\textbf{cov}}}(\hat{\beta}_{}^{j})_{(2,2)}}}\frac{U_j}{U_j\vee c}\right|. 
\end{align*}
From (\ref{bound6}) and $|U_j|\leq |U_j|\vee c$, we know
$$
\max\limits_{1\leq j \leq p}
             \left|
             \frac{\sqrt{\hat{\mathrm{\textbf{cov}}}(\hat{\beta}_{}^{j})_{(2,2)}}- \sqrt{\mathrm{\textbf{cov}}( u({\beta_0^{j}},\emph{\textbf{X}}_{j}^{\mathrm{s1}},Y))_{(2,2)}}}
             {\sqrt{\hat{\mathrm{\textbf{cov}}}(\hat{\beta}_{}^{j})_{(2,2)}}}\frac{U_j}{U_j\vee c}\right|=O_p\Big(\sqrt{\frac{\log p}{n}}\Big).
$$
From (\ref{bound4}), the definition of $U_j$ and $|U_j|\vee c\geq c>0$, we can show that
\begin{align*}
 &\max\limits_{1\leq j \leq p}
\left|\frac{\left(\sqrt{n}\left(\hat{\beta}_{}^{j}\right)_{(2)}
  - \left(u({\beta_0^{j}},\emph{\textbf{X}}_{j}^{\mathrm{s1}},Y)_{(2)}+\sqrt{n}(\beta_0^{j})_{(2)}\right)
 \right)
 \sqrt{\mathrm{\textbf{cov}}( u({\beta_0^{j}},\emph{\textbf{X}}_{j}^{\mathrm{s1}},Y))_{(2,2)}}
 }
 {\left(u({\beta_0^{j}},\emph{\textbf{X}}_{j}^{\mathrm{s1}},Y)_{(2)}+\sqrt{n}(\beta_0^{j})_{(2)}\right)\sqrt{\hat{\mathrm{\textbf{cov}}}(\hat{\beta}_{}^{j})_{(2,2)}}}\frac{U_j}{U_j\vee c}\right| = O_p\left(\frac{\log p}{\sqrt{n}}\right).
\end{align*}
Combining the above bounds, we obtain the desire bound for $\max_{1\leq j \leq p}|\frac{\widehat{T}_j-U_j}{U_j\vee c}|$. Following the similar proof, it's easy to show that
\[
 \max\limits_{1\leq j \leq p}\left|\frac{\widehat{T}_{jk}-U_{jk}}{|U_{jk}|\vee c}\right|= 
 O_p\left(\frac{\log p}{\sqrt{n}}\right).
\].

\subsection{Proof of Lemma \ref{lemma4}}
For $1 \leq i \leq n$, denote
\begin{align*} 
	\begin{split}
	S_{ij} &= U_{ij} - \mathrm{E}(U_{ij}),\\
	\hat{S}_{ij} &= S_{ij}\mathbbm{1}\left\{|S_{ij}| \leq \sqrt{n}/(\log p)^4 \right\}-\mathrm{E}\left(S_{ij}\mathbbm{1}\left\{|S_{ij}| \leq \sqrt{n}/(\log p)^4 \right\}\right),\\
	\tilde{S}_{ij} &= S_{ij}- \hat{S}_{ij},\\
	\hat{U}_{ijk}&=U_{ijk}\mathbbm{1}\left\{|U_{ijk}| \leq \sqrt{n}/(\log p)^4 \right\}-\mathrm{E}\left(U_{ijk}\mathbbm{1}\left\{|U_{ijk}| \leq \sqrt{n}/(\log p)^4 \right\}\right),\\
	\widetilde{U}_{ijk}&=U_{ijk}-\hat{U}_{ijk}.	
	\end{split}					
\end{align*}
 We have
\begin{align*}
    \begin{split}
        &\mathrm{P}_{jk} = \mathrm{P}\left(|U_{jk}| \geq t,|U_{j}| \geq \alpha,|U_{k}| \geq \alpha\right)\\
        & = 	\mathrm{P}\left(\left|\sum\limits_{i=1}^nU_{ijk}\right| \geq t\sqrt{n},\left|\sum\limits_{i=1}^n \left(S_{ij}+ \mathrm{E}(U_{ij})\right)\right| \geq \alpha\sqrt{n}, \left|\sum\limits_{i=1}^n (S_{ik}+ \mathrm{E}(U_{ik}))\right| \geq \alpha\sqrt{n}\right)\\
        &\leq
        \mathrm{P}\left(\left|\sum\limits_{i=1}^nU_{ijk}\right| \geq t\sqrt{n},\left|\sum\limits_{i=1}^n (\hat{S}_{ij}+ \mathrm{E}(U_{ij}))\right| \geq \alpha\sqrt{n} -\sqrt{n}/(\log p)^2, \left|\sum\limits_{i=1}^n (\hat{S}_{ik}+ \mathrm{E}(U_{ik}))\right| \geq \alpha\sqrt{n} -\sqrt{n}/(\log p)^2\right)\\
        &+  
        \mathrm{P}\left(\left|\sum\limits_{i=1}^nU_{ijk}\right| \geq t\sqrt{n},\left|\sum\limits_{i=1}^n \tilde{S}_{ij}\right| \geq \sqrt{n}/(\log p)^2\right)
        +
        \mathrm{P}\left(\left|\sum\limits_{i=1}^nU_{ijk}\right| \geq t\sqrt{n},\left|\sum\limits_{i=1}^n \tilde{S}_{ik}\right| \geq \sqrt{n}/(\log p)^2\right)
        \\
        &\leq \mathrm{P}\left(\left|\sum\limits_{i=1}^n\hat{U}_{ijk}\right| \geq t\sqrt{n} -\frac{\sqrt{n}}{(\log p)^2} ,\left|\sum\limits_{i=1}^n (\hat{S}_{ij}+ \mathrm{E}(U_{ij}))\right| \geq \alpha\sqrt{n} -\frac{\sqrt{n}}{(\log p)^2}, \left|\sum\limits_{i=1}^n (\hat{S}_{ik}+ \mathrm{E}(U_{ik}))\right| \geq \alpha\sqrt{n} -\frac{\sqrt{n}}{(\log p)^2}\right)\\
        &+  \mathrm{P}\left(\left|\sum\limits_{i=1}^n \tilde{S}_{ij}\right| \geq \sqrt{n}/(\log p)^2\right)
        +
        \mathrm{P}\left(\left|\sum\limits_{i=1}^n \tilde{S}_{ik}\right| \geq \sqrt{n}/(\log p)^2\right)+ \mathrm{P}\left(\left|\sum\limits_{i=1}^n\tilde{U}_{ijk}\right| \geq \sqrt{n}/(\log p)^2\right).
    \end{split}
\end{align*}
Since $\mathrm{E}(U_{ilk})=0$ under $H_0$, we have
\begin{align} 
n|\mathrm{E}(U_{ilk}\mathbbm{1}\left\{|U_{ijk}| \leq \sqrt{n}/(\log p)^4 \right\})|&=n|\mathrm{E}(U_{ilk}\mathbbm{1}\left\{|U_{ijk}| > \sqrt{n}/(\log p)^4 \right\})|\nonumber\\
&\leq n\{\mathrm{E}(U^2_{ilk})\}^{1/2}\{\mathrm{P}(|U_{ijk}| > \sqrt{n}/(\log p)^4\}^{1/2}\nonumber\\
&=o(\sqrt{n}/(\log p)^2),\label{eq_proof_lemma4_1}
\end{align}
where the last step holds by the Markov inequality and the condition $\mathrm{E}|U_{ijk}|^{2r+2+\epsilon}$ is bounded. Thus, we can show that
\begin{align*} 
&\mathrm{P}\left(\left|\sum\limits_{i=1}^n\widetilde{U}_{ijk}\right| \geq \sqrt{n}/(\log p)^2\right)\\
&\leq \mathrm{P}\left(\left|\sum\limits_{i=1}^n U_{ijk}\mathbbm{1}\left\{|U_{ijk}| > \sqrt{n}/(\log p)^4 \right\}\right| \geq \sqrt{n}/(\log p)^2-n\Big|\mathrm{E}(U_{ilk}\mathbbm{1}\left\{|U_{ijk}| \leq \sqrt{n}/(\log p)^4 \right\})\Big|\right)\\
&\leq \mathrm{P}\left(\left|\sum\limits_{i=1}^n U_{ijk}\mathbbm{1}\left\{|U_{ijk}| > \sqrt{n}/(\log p)^4 \right\}\right| \geq c\sqrt{n}/(\log p)^2\right),
\end{align*}
for some small constant $c>0$, where the last step is from (\ref{eq_proof_lemma4_1}). Note that the event in the above probability implies there exists at least some $i$ from $1,...,n$, such that $|U_{ijk}| \geq \sqrt{n}/(\log p)^4$. From the union bound, we obtain that 
$$
\mathrm{P}\left(\left|\sum\limits_{i=1}^n\widetilde{U}_{ijk}\right| \geq \sqrt{n}/(\log p)^2\right) \leq nP(|U_{ijk}| \geq \sqrt{n}/(\log p)^4) \leq C(\log p)^{-3/2}G(t),
$$
uniformly over $0\leq t\leq \sqrt{2\log p}$, where again the last step holds by the Markov inequality and the condition $\mathrm{E}|U_{ijk}|^{2r+2+\epsilon}$ is bounded. Similarly, we get
\begin{align*} 
     \mathrm{P}\left(\left|\sum\limits_{i=1}^n\tilde{S}_{ij}\right| \geq \sqrt{n}/(\log p)^2\right)&\leq C(\log p)^{-3/2}G(t).				
\end{align*}
Hence it follows that
\begin{equation*}
    \begin{split}
\mathrm{P}_{jk}\leq &\mathrm{P}\Big(\left|\sum\limits_{i=1}^n\hat{U}_{ijk}\right| \geq t\sqrt{n} -\frac{\sqrt{n}}{(\log p)^2} ,\left|\sum\limits_{i=1}^n (\hat{S}_{ij}+ \mathrm{E}(U_{ij}))\right| \geq \alpha\sqrt{n} -\frac{\sqrt{n}}{(\log p)^2},\\
&~~~~~\left|\sum\limits_{i=1}^n (\hat{S}_{ik}+ \mathrm{E}(U_{ik}))\right| \geq \alpha\sqrt{n} -\frac{\sqrt{n}}{(\log p)^2}\Big)+ C(\log p)^{-3/2}G(t).
    \end{split}
\end{equation*}
Similarly,
\begin{equation*}
    \begin{split}
\mathrm{P}_{jk}  \geq &\mathrm{P}\Big(\left|\sum\limits_{i=1}^n\hat{U}_{ijk}\right| \geq t\sqrt{n} +\frac{\sqrt{n}}{(\log p)^2} ,\left|\sum\limits_{i=1}^n (\hat{S}_{ij}+ \mathrm{E}(U_{ij}))\right| \geq \alpha\sqrt{n} +\frac{\sqrt{n}}{(\log p)^2},\\
&~~~~~\left|\sum\limits_{i=1}^n (\hat{S}_{ik}+ \mathrm{E}(U_{ik}))\right| \geq \alpha\sqrt{n} +\frac{\sqrt{n}}{(\log p)^2}\Big)- C(\log p)^{-3/2}G(t).
    \end{split}
\end{equation*}
By Theorem1 in \cite{zaitsev1987gaussian}, we have
\begin{equation*}
\begin{split}
	&\mathrm{P}\left(\left|\sum\limits_{i=1}^n\hat{U}_{ijk}\right| \geq t\sqrt{n} -\frac{\sqrt{n}}{(\log p)^2} ,\left|\sum\limits_{i=1}^n (\hat{S}_{ij}+ \mathrm{E}(U_{ij}))\right| \geq \alpha\sqrt{n} -\frac{\sqrt{n}}{(\log p)^2}, \left|\sum\limits_{i=1}^n (\hat{S}_{ik}+ \mathrm{E}(U_{ik}))\right| \geq \alpha\sqrt{n} -\frac{\sqrt{n}}{(\log p)^2}\right)\\
	 &\leq 
	\mathrm{P}\left(|\tilde{x}_1|\geq t-2/(\log p)^2,
	|\tilde{x}_2+\sqrt{n}\mathrm{E}(U_{1j})|\geq \alpha-2/(\log p)^2, |\tilde{x}_3+\sqrt{n}\mathrm{E}(U_{1k})|\geq \alpha-2/(\log p)^2 \right)\\
	&~~~~+C_1\exp(-C_2(\log p)^2),\\	
\end{split}				
\end{equation*}
where $\tilde{\textbf{x}}=(\tilde{x}_1,\tilde{x}_2,\tilde{x}_3)^T$ is a multivariate normal vector with mean zero and covariance matrix \ \ $\tilde{\Sigma}=\mathrm{\textbf{cov}}\left(\sum_{i=1}^n\widehat{U}_{ijk}/\sqrt{n}, \sum_{i=1}^n\hat{S}_{ij}/\sqrt{n},\sum_{i=1}^n\hat{S}_{ik}/\sqrt{n}\right)$. It's easy to show the similar result for the other direction.
Therefore,
\begin{align}
\mathrm{P}_{jk} &- \mathrm{P}\left(|\tilde{x}_1|\geq t-2/(\log p)^2,
	|\tilde{x}_2 + \sqrt{n}\mathrm{E}(U_{1j})|\geq \alpha-2/(\log p)^2, |\tilde{x}_3+\sqrt{n} \mathrm{E}(U_{1k})|\geq \alpha-2/(\log p)^2 \right)\nonumber\\
	&\leq C(\log p)^{-3/2}G(t),\label{dif1}
\end{align}
and
\begin{align*}
        \mathrm{P}_{jk} &- \mathrm{P}\left(|\tilde{x}_1|\geq t+2/(\log p)^2,
	|\tilde{x}_2+\sqrt{n}\mathrm{E}(U_{1j})|\geq \alpha+2/(\log p)^2, |\tilde{x}_3+\sqrt{n}\mathrm{E}(U_{1k})|\geq \alpha+2/(\log p)^2 \right)\nonumber\\
	&\geq -C(\log p)^{-3/2}G(t).
\end{align*}
Define
\begin{equation*}
    \Sigma = \mathrm{\textbf{cov}} \left(\sum\limits_{i=1}^nU_{ijk}/\sqrt{n},\sum\limits_{i=1}^nU_{ij}/\sqrt{n},\sum\limits_{i=1}^nU_{ik}/\sqrt{n}\right).
\end{equation*}
By condition $\mathrm{E}|S_{ij}|^{2r+2+\epsilon}<C$ and $\mathrm{E}|U_{ijk}|^{2r+2+\epsilon}<C$, we have $\left\|\Sigma- \tilde{\Sigma}\right\|_{\infty} \leq C(\log p)^{4r+4+2\epsilon}/n^{r+1+\epsilon/2}$. Therefore $\left\|\tilde{\Sigma}- \mathrm{\textbf{I}}\right\|_{\infty} \leq C(\log p)^{-2-\gamma}$. Since  $\tilde{\textbf{x}}$ is multivariat normal, it's easy to show that
\begin{align}
        &\mathrm{P}\left(|\tilde{x}_1|\geq t-2/(\log p)^2,
	|\tilde{x}_2 + \sqrt{n}\mathrm{E}(U_{1j})|\geq \alpha-2/(\log p)^2, |\tilde{x}_3+ \sqrt{n}\mathrm{E}(U_{1k})|\geq \alpha-2/(\log p)^2 \right)\nonumber\\
	&\leq (1+C(\log p)^{-1-\gamma})G(t)G_{j}(\alpha)G_{k}(\alpha),\label{dif2}    
\end{align}
uniformly over $0\leq t\leq \sqrt{2\log p}$. Combine (\ref{dif1}) and (\ref{dif2}) we have
\[
    \mathrm{P}_{jk} \leq (1+C(\log p)^{-1-\gamma_1})G(t)G_{j}(\alpha)G_{k}(\alpha).
\]
Similarly, for the other direction we can show that
\[
  \mathrm{P}_{jk} \geq (1-C(\log p)^{-1-\gamma_1})G(t)G_{j}(\alpha)G_{k}(\alpha).
\]
Hence, we obtain the first result (\ref{eq_lemma4_1}),
\[
 \sup\limits_{0 \leq t \leq  \sqrt{2\log p}}\left|\frac{\mathrm{P}_{jk}}{G(t)G_{j}(\alpha)G_k(\alpha)}-1\right| \leq C (\log p)^{-1-\gamma_1}.
\]
To show the second result  (\ref{eq_lemma4_2}), we can follow the similar proof for $P_{jk}$ to get  
\begin{equation*}
     \begin{split}
         &\mathrm{P}_{jkml} \leq (1+C(\log p)^{-1-\gamma_1})G^2(t)G_{j}(\alpha)G_{k}(\alpha)G_{m}(\alpha)G_{l}(\alpha),\\
         &\mathrm{P}_{jkml} \geq (1-C(\log p)^{-1-\gamma_1})G^2(t)G_{j}(\alpha)G_{k}(\alpha)G_{m}(\alpha)G_{l}(\alpha),
     \end{split}
\end{equation*}
which yields the desired bound.

\subsection{Proof of Lemma \ref{EJ}}
For any $1 \leq j < k \leq p$, by \ref{aspA6} we have
\[
\tilde{\emph{H}}_{01} = 
     \{(j,k)\in \emph{H}_0: 
        |\mathrm{\textbf{cov}}(U_{ij},U_{ik})| \leq C(\log p)^{-2-\gamma} \},
\]
and denote $\tilde{\emph{H}}_{02}= \emph{H}_{0} \setminus  \tilde{\emph{H}}_{01}$. 
Let  $\gamma_1 = \min\{\gamma, 1/2\}$, by lemma 6.1 in \cite{liu2013gaussian} and the proof of Lemma~\ref{lemma4}, we have 
 \begin{equation}
     \max\limits_{(j,k)\in \tilde{\emph{H}}_{01}}\left|\frac{\mathrm{P}(|U_{j}| \geq \alpha, |U_{k}| \geq \alpha)}{G_j(\alpha)G_k(\alpha)}-1\right| \leq C(\log p)^{-1-\gamma_1}.
 \end{equation}
 Then we have
 \begin{align*}
    \begin{split}
          \mathrm{E}(J) &= 
  \sum\limits_{(j,k)\in\tilde{\emph{H}}_{02}}\mathrm{P}\left\{|U_{j}| \geq \alpha,|U_{k}| \geq \alpha\right\} 
  +\sum\limits_{(j,k)\in\tilde{\emph{H}}_{01}}\mathrm{P}\left\{|U_{j}| \geq \alpha,|U_{k}| \geq \alpha\right\} \\
  & \geq 
  (1-C(\log p)^{-1-\gamma_1}) \sum\limits_{(j,k)\in \tilde{\emph{H}}_{01}}G_j(\alpha)G_k(\alpha)=\Omega(p^\xi),
    \end{split} 
\end{align*}
where the last step is from Assumption \ref{aspA6}.

\subsection{Proof of Lemma \ref{lemma6}}
Let 
\begin{equation*}
I(t)= \left|\frac{\sum_{(j,k)\in \emph{H}_0}\left(\mathbbm{1}\left\{|U_{jk}| \geq t,|U_{j}| \geq \alpha,|U_{k}| \geq \alpha\right\}-\mathrm{P}\left\{|U_{jk}| \geq t,|U_{j}| \geq \alpha,|U_{k}| \geq \alpha\right\}\right)}{\mathrm{E}(J)G(t)}\right|.
\end{equation*}
Define $\emph{H}^*_{01} = \big\{\{(j,k),(m,l)\}:(j,k),(m,l) \in \emph{H}_0,\  j = m \ or \  k = l \big\}$, $\emph{H}^*_{02} = \big\{\{(j,k),(m,l)\}: (j,k),(m,l) \in \emph{H}_0,\   j \neq m \ and \  k \neq l \big\} $, then we have $|\emph{H}^*_{01}| \asymp p^3$, $|\emph{H}^*_{02}|\asymp p^4$. Denote 
\begin{equation*}
\begin{split}
     \mathrm{E}_i &=  \sum_{\{(j,k),(m,l)\}\in \emph{H}^*_{0i}} \Big(\frac{\mathrm{P}\left(|U_{jk}| \geq t,|U_{j}| \geq \alpha,|U_{k}| \geq \alpha, |U_{ml}| \geq t,|U_{m}| \geq \alpha,|U_{l}| \geq \alpha\right)}{\mathrm{E}^2(J)G^2(t)}  \\
     &\quad\quad\quad - \frac{\mathrm{P}\left(|U_{jk}| \geq t,|U_{j}| \geq \alpha,|U_{k}| \geq \alpha\right)\mathrm{P}\left(|U_{ml}| \geq t,|U_{m}| \geq \alpha,|U_{l}| \geq \alpha\right)}{\mathrm{E}^2(J)G^2(t)}\Big) \\
     & = \sum_{\{(j,k),(m,l)\}\in \emph{H}^*_{0i}} \frac{\mathrm{P}_{jkml}-\mathrm{P}_{jk}\mathrm{P}_{ml}}{\mathrm{E}^2(J)G^2(t)},
\end{split}
\end{equation*}
where $\mathrm{P}_{jkml}$ and $\mathrm{P}_{jk}$ are defined in Lemma \ref{lemma4}. Then we can write
\[
	\begin{split}
	\mathrm{E}I^2(t) = \mathrm{E}_1 + \mathrm{E}_2.
	\end{split}					
\]
 Note that $|\mathrm{Cov(U_{ijk}, U_{iml})}| \leq \delta$,
$U_{jk}= n^{-1/2}\sum\limits_{i=1}^n U_{ijk}$, and $\big\{U_{ijk}$: $0 \leq i \leq n \big\}$ are i.i.d. random variables with mean zero. By lemma 6.2 in \cite{liu2013gaussian}, for some constant $C_1$ and  any $\{(j,k),(m,l)\}\in \emph{H}^*_{01}$, we have
\begin{equation}
\label{e11}
\begin{split}
     \mathrm{P}_{jkml} 
     \leq \mathrm{P}\left( |U_{jk}| \geq t, |U_{ml}| \geq t\right) \leq \frac{C_1}{(t+1)^2\exp(\frac{t^2}{1+\delta})}.
\end{split}
\end{equation}
By Lemma~\ref{EJ}, we have $\mathrm{E}(J)= \Omega(p^{\xi})$, hence for any $0 \leq t \leq \sqrt{2\log p}$,
\begin{equation}
\label{E1}
E_1\leq |H^*_{01}|\frac{\mathrm{P}_{jkml}}{\mathrm{E}^2(J)G^2(t)} = O(1/p^{2\xi-3-2\frac{\delta}{1+\delta}}).
\end{equation}
For $E_2$, we first split $\emph{H}^*_{02}$ into two subsets. Define 
\[
\emph{H}^*_{021} = \big\{\{(j,k),(m,l)\}: (j,k),(m,l)\in \emph{H}^*_{02}, \left\|\mathrm{\textbf{cov}}(\tilde{U}_{ijkml})- \mathrm{\textbf{I}}\right\|_{\infty} \leq C(\log p)^{-2-\gamma}\big\},
\]
and $\emph{H}^*_{022} = \emph{H}^*_{02} \setminus \emph{H}^*_{021}$.

Consider $\emph{H}^*_{021}$ first.
By Lemma~\ref{lemma4} we have for any $\{(j,k)(m,l)\}\in \emph{H}^*_{021}$,
\[
\left|\mathrm{P}_{jkml}-\mathrm{P}_{jk}\mathrm{P}_{ml}\right| \leq 
C(\log p)^{-1-\gamma_1}G_j(\alpha)G_k(\alpha)G_m(\alpha)G_l(\alpha)G^2(t).
\]
Note that for some constant $C_1$, we can show that
\begin{align}\label{E^2(J)}
    \begin{split}
        \mathrm{E}^2(J) &= \Big(\sum_{(j,k)\in H_0}\mathrm{P}\left\{|U_{j}| \geq \alpha,|U_{k}|\geq \alpha\right\}\Big)^2 \\
        &\geq 
\sum\limits_{\{(j,k),(m,l)\}\in \emph{H}^*_{021}}\mathrm{P}\left\{|U_{j}| \geq \alpha,|U_{k}| \geq \alpha\right\}\mathrm{P}\left\{|U_{m}| \geq \alpha,|U_{l}| \geq \alpha\right\}\\
&\geq(1-C_1(\log p)^{-1-\gamma_1})\sum\limits_{\{(j,k),(m,l)\}\in \emph{H}^*_{021}}G_j(\alpha)G_k(\alpha)G_m(\alpha)G_l(\alpha).
    \end{split}.
\end{align}
Therefore,
\[
    \left|\sum_{\{(j,k),(m,l)\}\in \emph{H}^*_{021}} \frac{\mathrm{P}_{jkml}-\mathrm{P}_{jk}\mathrm{P}_{ml}}{\mathrm{E}^2(J)G^2(t)}\right| \leq C_2(\log p)^{-1-\gamma_1}.
\]
For $H^*_{022}$, by (\ref{asp2}) we have $|H^*_{022}| = O(p^{4-\kappa})$. By the same proof for $H^*_{01}$, for any $0\leq t \leq \sqrt{2\log p}$, we have
\[
  \left|\sum_{\{(j,k),(m,l)\}\in \emph{H}^*_{022}} \frac{\mathrm{P}_{jkml}-\mathrm{P}_{jk}\mathrm{P}_{ml}}{\mathrm{E}^2(J)G^2(t)}\right| \leq 
  O(1/p^{2\xi-4+\kappa-2\frac{\delta}{1+\delta}}).
\]
Hence
\begin{equation}
\label{E2}
    E_2 \leq C(\log p)^{-1-\gamma_1}.
\end{equation}
Combining (\ref{E1}) and (\ref{E2}) we get for any $0\leq t \leq \sqrt{2\log p}$,
\begin{equation*}
    	\mathrm{E}I^2(t) \leq C(\log p)^{-1-\gamma_1}.
\end{equation*}
Note that $m\sim  \sqrt{\log p}/z_p$, and $z_p = (\log p)^{-\frac{1+\gamma_1}{2}}= o(1/\sqrt{\log p})$, then
for any $\epsilon >0$,
\begin{equation*}
    \mathrm{P}(\max\limits_{0 \leq r \leq m} J_{r,1} \geq \epsilon) \leq \sum\limits_{r=0}^m \mathrm{P}(J_{r,1} \geq \epsilon) = \sum\limits_{r=0}^m\mathrm{P}(I(t_r) \geq \epsilon) \leq \sum\limits_{r=0}^m\frac{\mathrm{E}[I^2(t_r)]}{\epsilon^2} = O\left((\log p)^{-\gamma_1/2}\right).
\end{equation*}
This finishes the proof.

\subsection{Proof of Lemma \ref{lemma7}}
Similar as the proof of Lemma~\ref{lemma6}, for any $0 \leq t \leq \sqrt{2\log p}$, denote 
\begin{equation*}
     I(t)=\left|\frac{\sum_{(j,k)\in \emph{H}_0}(\mathrm{P}\left\{|U_{jk}| \geq t,|U_{j}| \geq \alpha,|U_{k}| \geq \alpha\right\}-G(t) \mathbbm{1} \left\{|U_j| \geq \alpha,|U_k| \geq \alpha \right\})}{\mathrm{E}(J)G(t)}\right|.
\end{equation*}
Let $\mathrm{P}_{jk} = \mathrm{P}\left(|U_{jk}| \geq t,|U_{j}| \geq \alpha,|U_{k}| \geq \alpha\right)$, $\mathrm{Q}_{jk} = \mathrm{P}\left(|U_{j}| \geq \alpha,|U_{k}| \geq \alpha\right) $ and 
$$
\mathrm{Q}_{jkml}=\mathrm{P}\left(|U_{j}| \geq \alpha,|U_{k}| \geq \alpha,|U_{m}| \geq \alpha,|U_{l}| \geq \alpha\right).$$ 
Denote $\emph{H}^*_{01} = \big\{\{(j,k),(m,l)\}:(j,k),(m,l) \in \emph{H}_0,\  j = m \ or \  k = l \big\}$, $\emph{H}^*_{02} = \big\{\{(j,k),(m,l)\}: (j,k),(m,l) \in \emph{H}_0,\   j \neq m \ and \  k \neq l \big\} $, then we have $|\emph{H}^*_{01}| \asymp p^3$, $|\emph{H}^*_{02}|\asymp p^4$. We have
\begin{equation*}
\begin{split}
    \mathrm{E}I^2(t) &= \frac{\sum_{(j,k)(m,l)\in \emph{H}_0}\Big\{\mathrm{P}_{jk}\mathrm{P}_{ml}- G(t)\mathrm{P}_{ml}\mathrm{Q}_{jk}
    - G(t)\mathrm{P}_{jk}\mathrm{Q}_{ml}
    + G^2(t)\mathrm{Q}_{jkml}\Big\}}{\mathrm{E}^2(J)G^2(t)}\\
    &\coloneqq \mathrm{E}_1 + \mathrm{E}_2,
    \end{split}
\end{equation*}
where 
\begin{equation}\label{E^I}
    \mathrm{E}_i=\frac{\sum_{(j,k)(m,l)\in \emph{H}^*_{0i}}\Big\{\mathrm{P}_{jk}\mathrm{P}_{ml}- G(t)\mathrm{P}_{ml}\mathrm{Q}_{jk}
    - G(t)\mathrm{P}_{jk}\mathrm{Q}_{ml}
    + G^2(t)\mathrm{Q}_{jkml}\Big\}}{\mathrm{E}^2(J)G^2(t)}.
\end{equation}
First we consider $\mathrm{E}_1$.  By lemma 6.1 in \cite{liu2013gaussian} we have
 \begin{equation} \label{ujk}
   \sup\limits_{0 \leq t \leq  \sqrt{2\log p}}\left|\frac{\mathrm{P}(|U_{jk}| \geq t)}{G(t)}-1\right| \leq C(\log p)^{-1-\gamma_1}.
 \end{equation}
 Therefore, 
\begin{equation}
\label{Pjk}
     \mathrm{P}_{jk} \leq 
 \mathrm{P}\left(|U_{jk}| \geq t\right)
 \leq G(t)(1+C(\log p)^{-1-\gamma})
\end{equation}
 and
 \[
 \mathrm{P}_{jk}\mathrm{P}_{ml} \leq 
  \mathrm{P}\left(|U_{jk}| \geq t\right)\mathrm{P}\left(|U_{ml}| \geq t\right)
  \leq G^2(t)(1+C(\log p)^{-1-\gamma_1}).
 \]
Note that Lemma~\ref{EJ} gives $\mathrm{E}(J)= \Omega(p^{\xi})$, then  given the two inequalities above we have 
\begin{equation}
\label{E_1}
     \mathrm{E}_1 = O(1/p^{2\xi-3}).
\end{equation}
Next we consider $\mathrm{E}_2$,  we further split $\emph{H}^*_{02}$ into
\[
\emph{H}^*_{021} = \big\{\{(j,k),(m,l)\}: (j,k),(m,l)\in \emph{H}^*_{02}, \left\|\mathrm{\textbf{cov}}(\tilde{U}_{ijkml})- \mathrm{\textbf{I}}\right\|_{\infty} \leq C(\log p)^{-2-\gamma}\big\},
\]
and $\emph{H}^*_{022} = \emph{H}^*_{02} \setminus \emph{H}^*_{021}$. Write $\mathrm{E}_2 \coloneqq \mathrm{E}_{21}+\mathrm{E}_{22}$, where $\mathrm{E}_{21}$ is for
$\{(j,k),(m,l)\}\in H_{021}^*$ and $\mathrm{E}_{22}$ is for $\{(j,k),(m,l)\}\in H_{022}^*$.
 For $H_{021}^*$, by Lemma~\ref{lemma4} we have
\begin{equation*}
     \sup\limits_{0 \leq t \leq  \sqrt{2\log p}}\left|\frac{\mathrm{P}_{jk}}{G(t)G_j(\alpha)G_k(\alpha)}-1\right| \leq C (\log p)^{-1-\gamma_1},
\end{equation*}
and we can similarly show that
\begin{equation*}
    \sup\limits_{0 \leq \alpha \leq  \sqrt{2\log p}}\left|\frac{\mathrm{Q}_{jk}}{G_j(\alpha)G_k(\alpha)}-1\right| \leq C (\log p)^{-1-\gamma_1},
\end{equation*}
\begin{equation*}
    \sup\limits_{0 \leq \alpha \leq  \sqrt{2\log p}}\left|\frac{\mathrm{Q}_{jkml}}{G_j(\alpha)G_k(\alpha)G_m(\alpha)G_l(\alpha)}-1\right| \leq C (\log p)^{-1-\gamma_1},
\end{equation*}
where $\gamma_1 = \min \{1/2, \gamma\}$. From (\ref{E^2(J)}) we have 
\[
\mathrm{E}^2(J) \geq (1-C_1(\log p)^{-1-\gamma_1})\sum\limits_{\{(j,k),(m,l)\}\in \emph{H}^*_{021}}G_j(\alpha)G_k(\alpha)G_m(\alpha)G_l(\alpha),
\]
then combining the inequalities above gives
\begin{equation}\label{E_21}
         \mathrm{E}_{21}\leq 
    C(\log p)^{-1-\gamma_1}.
\end{equation}
Under the condition of (\ref{asp2}) we have $\left|\emph{H}^*_{022}\right| = O(p^{4-\kappa})$.
(\ref{Pjk}) leads to 
\begin{equation}\label{E_22}
    \mathrm{E}_{22}\leq 2|\emph{H}^*_{02}|(1+o(1))/\mathrm{E}^2(J)=O(1/p^{\kappa-4+2\xi}).
\end{equation}
Combine (\ref{E_1}), (\ref{E_21}) and (\ref{E_22}) and by Markov's inequality, we finish the proof.

\subsection{Proof of Lemma \ref{lem_J}}
Denote $ U_{ijkml} = ( S_{ij}, S_{ik}, S_{im}, S_{il})\in\RR^4$, $\emph{H}_{01} = \big\{\{(j,k),(m,l)\}:(j,k),(m,l) \in \emph{H}_0,\  j = m \ or \  k = l \big\}$, $\emph{H}_{02} = \big\{\{(j,k),(m,l)\}: (j,k),(m,l) \in \emph{H}_0,\   j \neq m \ and \  k \neq l \big\} $, 
\[
\emph{H}_{021} = \big\{\{(j,k),(m,l)\}: (j,k),(m,l)\in \emph{H}_{02}, \left\|\mathrm{\textbf{cov}}({U}_{ijkml})- \mathrm{\textbf{I}}\right\|_{\infty} \leq C(\log p)^{-2-\gamma}\big\},
\]
and $\emph{H}_{022} = \emph{H}_{02} \setminus \emph{H}_{021}$. Then we have $|\emph{H}_{01}| \asymp p^3$ and by (\ref{asp2}) we also have $|\emph{H}_{021}|=\Omega(p^{4})$, $|\emph{H}_{022}|=O(p^{4-\kappa})$.
Write
\[
\frac{J}{\mathrm{E}(J)}=\frac{J-\mathrm{E}(J)}{\mathrm{E}(J)}+1.
\]
Note that 
\[
\frac{J-\mathrm{E}(J)}{\mathrm{E}(J)} = \frac{\sum_{(j,k), (m,l)\in \emph{H}_{0}}
 \Big\{I \left\{|U_{j}| \geq \alpha,|U_{k}| \geq \alpha \right\} - \mathrm{P}\left(|U_{j}| \geq \alpha,|U_{k}| \geq \alpha\right)\Big\}
}{\mathrm{E}(J)}.
\]
Denote $\mathrm{Q}_{jkml}=\mathrm{P}\left(|U_{j}| \geq \alpha,|U_{k}| \geq \alpha,|U_{m}| \geq \alpha,|U_{l}| \geq \alpha\right)$ and $\mathrm{Q}_{jk}=\mathrm{P}\left(|U_{j}| \geq \alpha,|U_{k}| \geq \alpha\right)$, then we have
\[
\mathrm{E}\left(\frac{J-\mathrm{E}(J)}{\mathrm{E}(J)}\right)^2
=\frac{\sum_{(j,k), (m,l)\in \emph{H}_{0}}(\mathrm{Q}_{jkml}-\mathrm{Q}_{jk}\mathrm{Q}_{ml})}{\mathrm{E}^2(J)}.
\]
Follow the same proof as Lemma~\ref{lemma6}, we have
\[
\sum\limits_{\{(j,k),(m,l)\}\in \emph{H}_{01}}(\mathrm{Q}_{jkml}-\mathrm{Q}_{jk}\mathrm{Q}_{ml})
= O(p^3),
\]
\[
\sum_{\{(j,k),(m,l)\}\in \emph{H}_{021}}(\mathrm{Q}_{jkml}-\mathrm{Q}_{jk}\mathrm{Q}_{ml})
\leq C(\log p)^{-1-\gamma_1}
\sum_{\{(j,k),(m,l)\}\in \emph{H}_{021}}
G_j(\alpha)G_k(\alpha)G_m(\alpha)G_l(\alpha),
\]
and 
\[
\sum\limits_{\{(j,k),(m,l)\}\in \emph{H}_{022}}(\mathrm{Q}_{jkml}-\mathrm{Q}_{jk}\mathrm{Q}_{ml})
= O(p^{4-\kappa}).
\]
Note that
\begin{align*}
    \begin{split}
        \mathrm{E}^2(J) &= \Big(\sum\limits_{(j,k)\in \emph{H}_0}\mathrm{P}\left\{|U_{j}| \geq \alpha,|U_{k}|\geq \alpha\right\}\Big)^2 \\
        &\geq 
\sum\limits_{\{(j,k),(m,l)\}\in \emph{H}_{021}}\mathrm{P}\left\{|U_{j}| \geq \alpha,|U_{k}| \geq \alpha\right\}\mathrm{P}\left\{|U_{m}| \geq \alpha,|U_{l}| \geq \alpha\right\}\\
&\geq (1-C_1(\log p)^{-1-\gamma_1})\sum\limits_{\{(j,k),(m,l)\}\in \emph{H}_{021}}G_j(\alpha)G_k(\alpha)G_m(\alpha)G_l(\alpha).
    \end{split}
\end{align*}
Then by Lemma~\ref{EJ}  we have
\begin{align*}
    \begin{split}
        \mathrm{E}\left(\frac{J-\mathrm{E}(J)}{\mathrm{E}(J)}\right)^2
&=O\left(\frac{p^3}{p^{2\xi}}+
(\log p)^{-1-\gamma_1}+ \frac{\sum_{\{(j,k),(m,l)\}\in \emph{H}_{022}}(\mathrm{Q}_{jkml}-\mathrm{Q}_{jk}\mathrm{Q}_{ml})}{\mathrm{E}^2(J)}\right)\\
&=O\left(\frac{1}{p^{2\xi-3}}+
(\log p)^{-1-\gamma_1}+ \frac{p^{4-\kappa}}{p^{2\xi}}\right)= o(1).
    \end{split}
\end{align*}
Therefore by Markov inequality we have
\[
\left|\frac{J}{\mathrm{E}(J)}-1\right|=o_p(1)
\]
and we finish the proof.

\subsection{Proof of Lemma \ref{lem_power}}
For simplicity, denote $(\Sigma_{jk}^*(\emph{\textbf{X}}_{ijk}^{\mathrm{s2}})^T)_{(4)} \coloneqq x_{ijk}$. From assumption \ref{asp_power}, we know that $x_{ijk}$ is bounded by $\tilde K$. Then $x_{ijk}\epsilon_{ijk}$ is Sub-Exponential with parameter $\lambda \tilde K$, and the following inequality holds
\begin{equation}\label{eq_lem_power_pf_1}
    \mathrm{P}\left\{\left|\frac{1}{n}\sum_{i=1}^{n} x_{ijk}\epsilon_{ijk}\right| \geq t\right\} \leq 2 \exp  \left[-\frac{n}{2}\left(\frac{t^{2}}{\lambda^{2}\tilde K^2} \wedge \frac{t}{\lambda \tilde K}\right)\right].
\end{equation}
Recall that
 \[
  U_{ijk} = \frac{\left(-\left[\mathrm{E}_{\beta_0^{jk}}(b^{\prime\prime}(\emph{\textbf{X}}_{ijk}^{\mathrm{s2}}\beta_0^{jk})(\emph{\textbf{X}}_{ijk}^{\mathrm{s2}})^T\emph{\textbf{X}}_{ijk}^{\mathrm{s2}})\right]^{-1}\left\{Y_{i}-b^{\prime}\left(\emph{\textbf{X}}_{ijk}^{\mathrm{s2}} \beta_0^{jk}\right)\right\} \cdot (\emph{\textbf{X}}_{ijk}^{\mathrm{s2}})^T+\beta_0^{jk}\right)_{(4)}}{\sqrt{\mathrm{\textbf{cov}}( u({\beta_0^{jk}},\emph{\textbf{X}}_{jk}^{\mathrm{s2}},Y))_{(4,4)}}} ,
 \]  
and
\[
\sqrt{n}\mathrm{E}(U_{ijk}) = \mathrm{E}(U_{jk})=\frac{\sqrt{n}(\beta_{0}^{jk})_{(4)}}{\sqrt{\mathrm{\textbf{cov}}( u({\beta_0^{jk}},\emph{\textbf{X}}_{jk}^{\mathrm{s2}},Y))_{(4,4)}}}.
\]
For any $c>0$ which is upper bounded by a constant and an arbitrary small $\epsilon>0$, using (\ref{two-side}) we obtain
\begin{equation}\label{eq_lem_power_pf_2}
\mathrm{P}(\sup\limits_{(j,k)\in H_1}|\hat{T}_{jk}| < c\sqrt{\log p})\leq \mathrm{P}(\sup\limits_{(j,k)\in H_1}|U_{jk}| < c\sqrt{\log p}+\frac{1}{\sqrt{\log p}})+\epsilon.
\end{equation}
By the triangle inequality and the standard union bound, we further have 
\begin{align*}
&\mathrm{P}(\sup\limits_{(j,k)\in H_1}|U_{jk}| < c\sqrt{\log p}+\frac{1}{\sqrt{\log p}})\\
&\leq \mathrm{P}(\sup\limits_{(j,k)\in H_1}|\sqrt{n}\mathrm{E}(U_{ijk})|-|U_{jk}-\sqrt{n}\mathrm{E}(U_{ijk})| < c\sqrt{\log p}+\frac{1}{\sqrt{\log p}}) \\
&\leq \sum\limits_{(j,k)\in H_1}\mathrm{P}(|U_{jk}-\sqrt{n}\mathrm{E}(U_{ijk})|>|\sqrt{n}\mathrm{E}(U_{ijk})|- c\sqrt{\log p}-\frac{1}{\sqrt{\log p}})\\
&\leq \sum\limits_{(j,k)\in H_1}\mathrm{P}\left(\left|\frac{1}{n}\sum\limits_{i=1}^nx_{ijk}\epsilon_{ijk}\right|> \delta\sqrt{ \frac{\log p}{n}} -\sqrt{\mathrm{\textbf{cov}}( u({\beta_0^{jk}},\emph{\textbf{X}}_{jk}^{\mathrm{s2}},Y))_{(4,4)}}(c\sqrt{\frac{\log p}{n}}+\sqrt{\frac{1}{n\log p}})\right)\\
&= O\Big( p^{2-\frac{\left(\delta-\sqrt{\mathrm{\textbf{cov}}( u({\beta_0^{jk}},\emph{\textbf{X}}_{jk}^{\mathrm{s2}},Y))_{(4,4)}}c\right)^2}{2\lambda^2\tilde K^2}}\Big),
\end{align*}
where we use (\ref{eq_signal}) and (\ref{eq_lem_power_pf_2}) in the last two lines. Finally, if $\delta-\sqrt{\mathrm{\textbf{cov}}( u({\beta_0^{jk}},\emph{\textbf{X}}_{jk}^{\mathrm{s2}},Y))_{(4,4)}}c \geq 2\lambda \tilde K+\zeta$ for some constant $\zeta>0$, then 
$$
\mathrm{P}(\sup\limits_{(j,k)\in H_1}|U_{jk}| < c\sqrt{\log p}+\frac{1}{\sqrt{\log p}})=o(1),
$$
and together with (\ref{eq_lem_power_pf_2}), we obtain the desired result.

\section{Additional Numerical Results}\label{app_numerical}

Furthermore, for other most connected genes, we queried GTEx (\cite{gtex2015genotype}), a database of  tissue-specific gene expression and regulation. The results of gene ZBTB16, NDUFB9 and BANP are shown in Figure~\ref{zbtb16}, \ref{ndufb9} and \ref{banp} respectively. All these results show that the genes identified by our method are expressed in bladder tissue, which supports our data analysis result.

\begin{figure}
    \hspace{-4em}
    \includegraphics[ scale=0.45
  ]{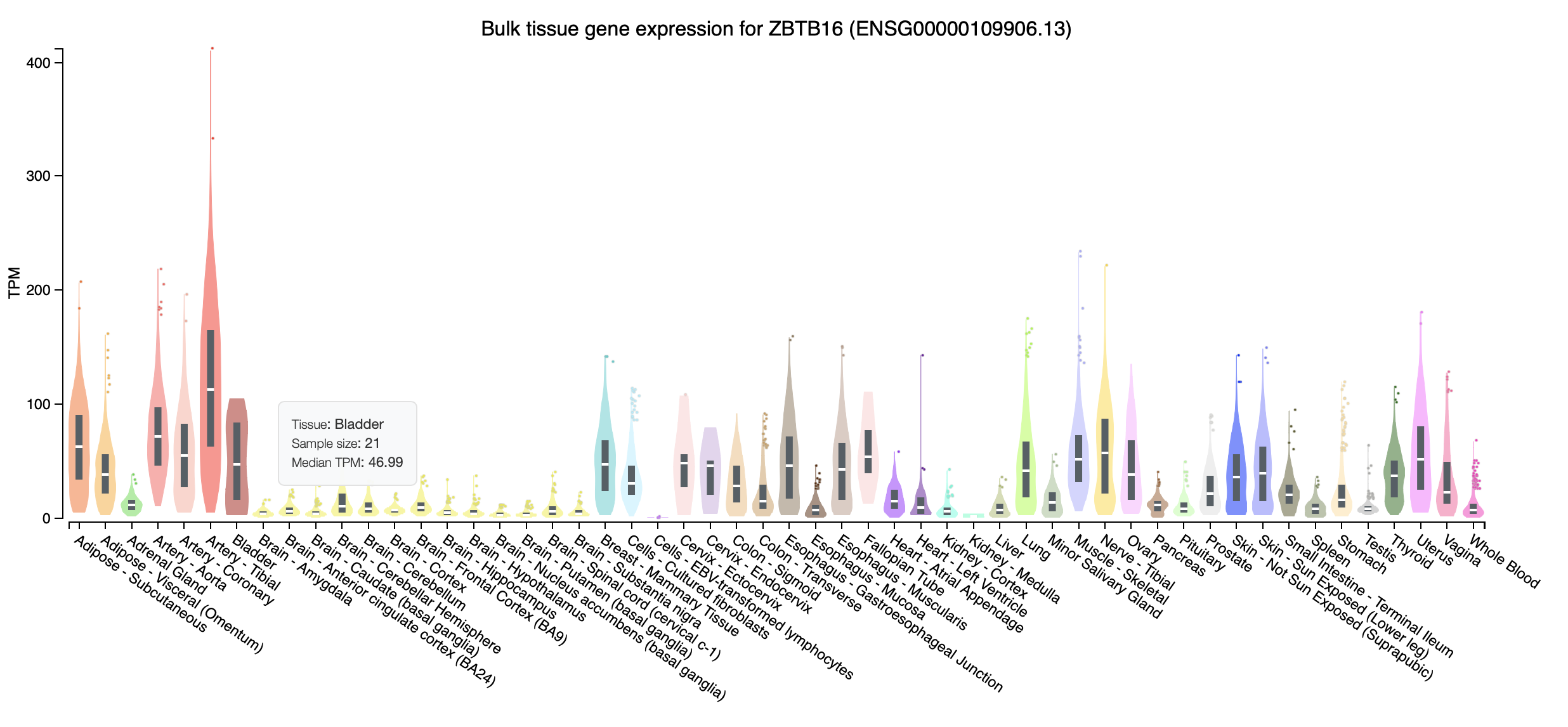}
    \caption{GTEx Portal: Gene expression for ZBTB16 (rs238930).}
    \label{zbtb16}
\end{figure}

\begin{figure}
    \hspace{-4em}
    \includegraphics[ scale=0.45
  ]{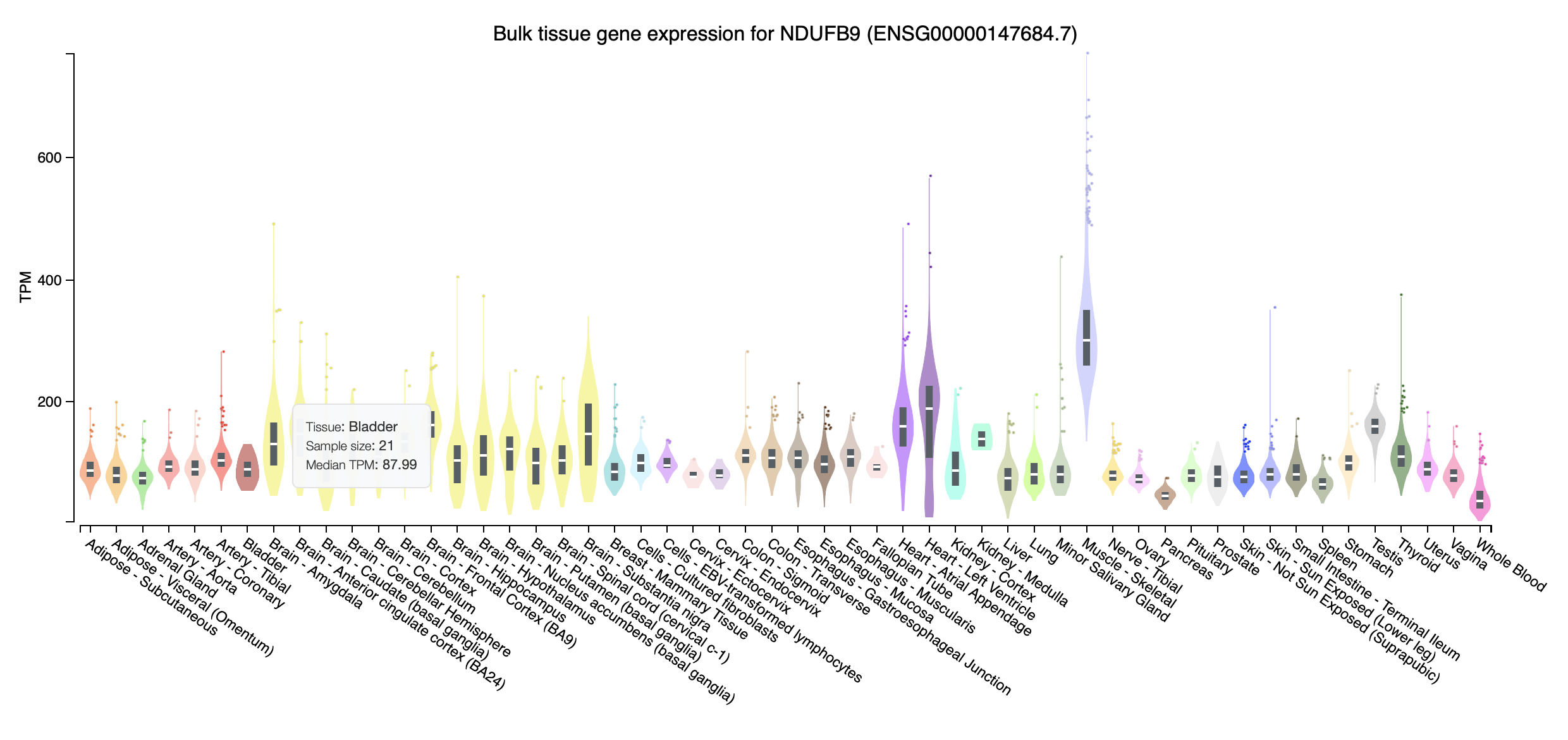}
    \caption{GTEx Portal: Gene expression for NDUFB9 (rs3829038).}
    \label{ndufb9}
\end{figure}

\begin{figure}
    \hspace{-4em}
    \includegraphics[ scale=0.45
  ]{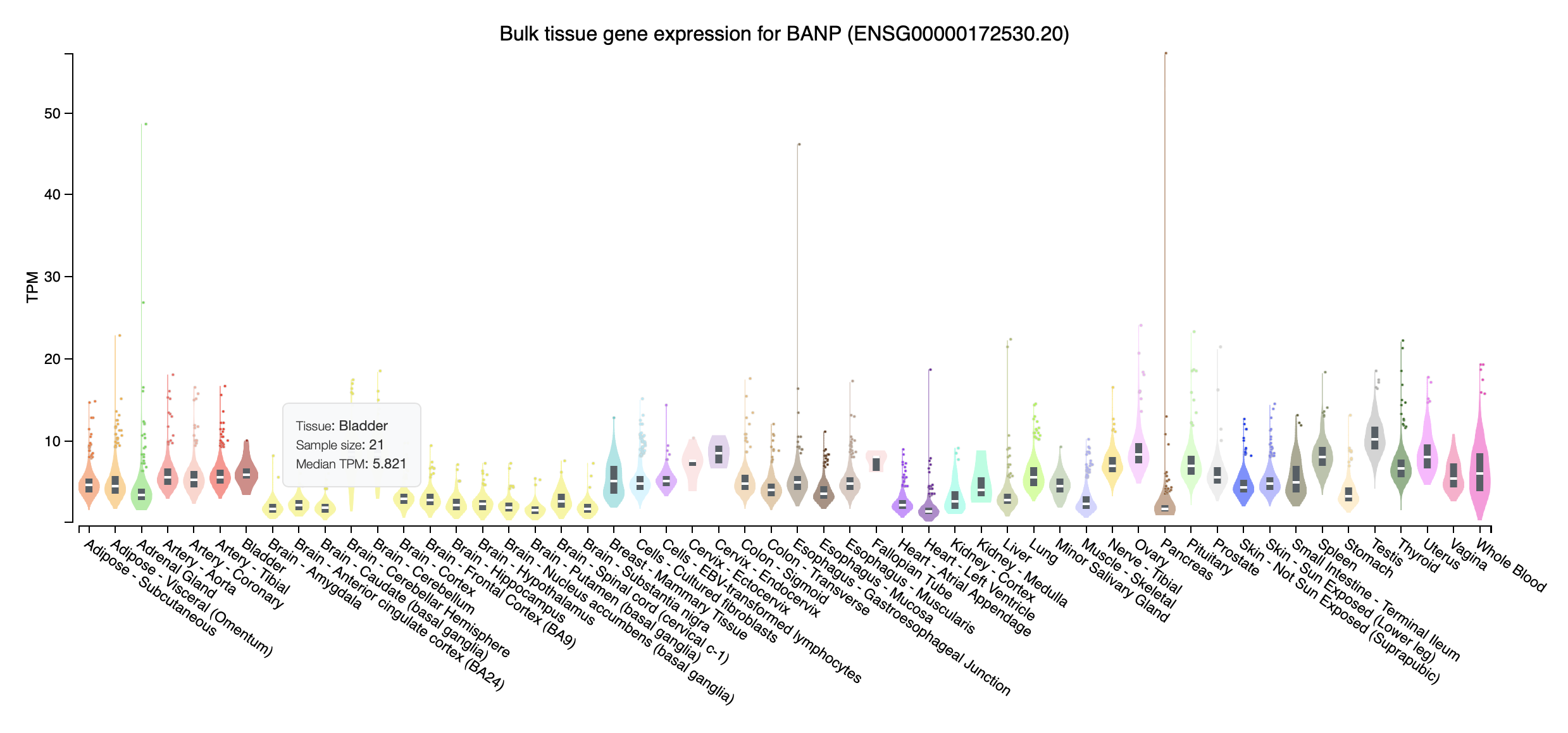}
    \caption{GTEx Portal: Gene expression for BANP (rs8063865).}
    \label{banp}
\end{figure}

\setlength{\bibsep}{0.85pt}
{
\bibliographystyle{ims}
\bibliography{ref}

\begin{thebibliography}{41}
\expandafter\ifx\csname natexlab\endcsname\relax\def\natexlab#1{#1}\fi
\expandafter\ifx\csname url\endcsname\relax
  \def\url#1{\texttt{#1}}\fi
\expandafter\ifx\csname urlprefix\endcsname\relax\def\urlprefix{URL }\fi

\bibitem[{{American Cancer Society}(2022)}]{cancerstat}
\textsc{{American Cancer Society}} (2022).
\newblock Key statistics for bladder cancer. retrieved from
  \url{https://www.cancer.org/cancer/bladder-cancer/about/key-statistics.html}.

\bibitem[{Benjamini and Hochberg(1995)}]{benjamini1995controlling}
\textsc{Benjamini, Y.} and \textsc{Hochberg, Y.} (1995).
\newblock Controlling the false discovery rate: a practical and powerful
  approach to multiple testing.
\newblock \textit{Journal of the Royal statistical society: series B
  (Methodological)} \textbf{57} 289--300.

\bibitem[{Benjamini and Yekutieli(2001)}]{benjamini2001control}
\textsc{Benjamini, Y.} and \textsc{Yekutieli, D.} (2001).
\newblock The control of the false discovery rate in multiple testing under
  dependency.
\newblock \textit{Annals of statistics}  1165--1188.

\bibitem[{Bien et~al.(2013)Bien, Taylor and Tibshirani}]{bien2013lasso}
\textsc{Bien, J.}, \textsc{Taylor, J.} and \textsc{Tibshirani, R.} (2013).
\newblock A lasso for hierarchical interactions.
\newblock \textit{Annals of statistics} \textbf{41} 1111.

\bibitem[{B{\"u}hlmann and Van De~Geer(2011)}]{buhlmann2011statistics}
\textsc{B{\"u}hlmann, P.} and \textsc{Van De~Geer, S.} (2011).
\newblock \textit{Statistics for high-dimensional data: methods, theory and
  applications}.
\newblock Springer Science \& Business Media.

\bibitem[{Consortium et~al.(2015)Consortium, Ardlie, Deluca, Segr{\`e},
  Sullivan, Young, Gelfand, Trowbridge, Maller, Tukiainen
  et~al.}]{gtex2015genotype}
\textsc{Consortium, G.}, \textsc{Ardlie, K.~G.}, \textsc{Deluca, D.~S.},
  \textsc{Segr{\`e}, A.~V.}, \textsc{Sullivan, T.~J.}, \textsc{Young, T.~R.},
  \textsc{Gelfand, E.~T.}, \textsc{Trowbridge, C.~A.}, \textsc{Maller, J.~B.},
  \textsc{Tukiainen, T.} \textsc{et~al.} (2015).
\newblock The genotype-tissue expression (gtex) pilot analysis: multitissue
  gene regulation in humans.
\newblock \textit{Science} \textbf{348} 648--660.

\bibitem[{Dai et~al.(2012)Dai, Kooperberg, Leblanc and Prentice}]{dai2012two}
\textsc{Dai, J.~Y.}, \textsc{Kooperberg, C.}, \textsc{Leblanc, M.} and
  \textsc{Prentice, R.~L.} (2012).
\newblock Two-stage testing procedures with independent filtering for
  genome-wide gene-environment interaction.
\newblock \textit{Biometrika} \textbf{99} 929--944.

\bibitem[{Fan et~al.(2019)Fan, Zhu and Ma}]{fan2019nonlinear}
\textsc{Fan, G.}, \textsc{Zhu, L.} and \textsc{Ma, S.} (2019).
\newblock Nonlinear interaction detection through model-based sufficient
  dimension reduction.
\newblock \textit{Statistica Sinica} \textbf{29} 917--937.

\bibitem[{Fan et~al.(2015)Fan, Kong, Li and Zheng}]{fan2015innovated}
\textsc{Fan, Y.}, \textsc{Kong, Y.}, \textsc{Li, D.} and \textsc{Zheng, Z.}
  (2015).
\newblock Innovated interaction screening for high-dimensional nonlinear
  classification.
\newblock \textit{The Annals of Statistics} \textbf{43} 1243--1272.

\bibitem[{Gauderman et~al.(2010)Gauderman, Thomas, Murcray, Conti, Li and
  Lewinger}]{gauderman2010efficient}
\textsc{Gauderman, W.~J.}, \textsc{Thomas, D.~C.}, \textsc{Murcray, C.~E.},
  \textsc{Conti, D.}, \textsc{Li, D.} and \textsc{Lewinger, J.~P.} (2010).
\newblock Efficient genome-wide association testing of gene-environment
  interaction in case-parent trios.
\newblock \textit{American journal of epidemiology} \textbf{172} 116--122.

\bibitem[{Greene et~al.(2015)Greene, Krishnan, Wong, Ricciotti, Zelaya,
  Himmelstein, Zhang, Hartmann, Zaslavsky, Sealfon
  et~al.}]{greene2015understanding}
\textsc{Greene, C.~S.}, \textsc{Krishnan, A.}, \textsc{Wong, A.~K.},
  \textsc{Ricciotti, E.}, \textsc{Zelaya, R.~A.}, \textsc{Himmelstein, D.~S.},
  \textsc{Zhang, R.}, \textsc{Hartmann, B.~M.}, \textsc{Zaslavsky, E.},
  \textsc{Sealfon, S.~C.} \textsc{et~al.} (2015).
\newblock Understanding multicellular function and disease with human
  tissue-specific networks.
\newblock \textit{Nature genetics} \textbf{47} 569--576.

\bibitem[{Hao and Zhang(2014)}]{hao2014interaction}
\textsc{Hao, N.} and \textsc{Zhang, H.~H.} (2014).
\newblock Interaction screening for ultrahigh-dimensional data.
\newblock \textit{Journal of the American Statistical Association} \textbf{109}
  1285--1301.

\bibitem[{Kiemeney et~al.(2010)Kiemeney, Sulem, Besenbacher, Vermeulen,
  Sigurdsson, Thorleifsson, Gudbjartsson, Stacey, Gudmundsson, Zanon
  et~al.}]{kiemeney2010sequence}
\textsc{Kiemeney, L.~A.}, \textsc{Sulem, P.}, \textsc{Besenbacher, S.},
  \textsc{Vermeulen, S.~H.}, \textsc{Sigurdsson, A.}, \textsc{Thorleifsson,
  G.}, \textsc{Gudbjartsson, D.~F.}, \textsc{Stacey, S.~N.},
  \textsc{Gudmundsson, J.}, \textsc{Zanon, C.} \textsc{et~al.} (2010).
\newblock A sequence variant at 4p16. 3 confers susceptibility to urinary
  bladder cancer.
\newblock \textit{Nature genetics} \textbf{42} 415--419.

\bibitem[{Kiemeney et~al.(2008)Kiemeney, Thorlacius, Sulem, Geller, Aben,
  Stacey, Gudmundsson, Jakobsdottir, Bergthorsson, Sigurdsson
  et~al.}]{kiemeney2008sequence}
\textsc{Kiemeney, L.~A.}, \textsc{Thorlacius, S.}, \textsc{Sulem, P.},
  \textsc{Geller, F.}, \textsc{Aben, K.~K.}, \textsc{Stacey, S.~N.},
  \textsc{Gudmundsson, J.}, \textsc{Jakobsdottir, M.}, \textsc{Bergthorsson,
  J.~T.}, \textsc{Sigurdsson, A.} \textsc{et~al.} (2008).
\newblock Sequence variant on 8q24 confers susceptibility to urinary bladder
  cancer.
\newblock \textit{Nature genetics} \textbf{40} 1307--1312.

\bibitem[{Kooperberg and LeBlanc(2008)}]{kooperberg2008increasing}
\textsc{Kooperberg, C.} and \textsc{LeBlanc, M.} (2008).
\newblock Increasing the power of identifying gene$\times$ gene interactions in
  genome-wide association studies.
\newblock \textit{Genetic Epidemiology: The Official Publication of the
  International Genetic Epidemiology Society} \textbf{32} 255--263.

\bibitem[{Li et~al.(2021)Li, Kong, Fan and Lv}]{li2021high}
\textsc{Li, D.}, \textsc{Kong, Y.}, \textsc{Fan, Y.} and \textsc{Lv, J.}
  (2021).
\newblock High-dimensional interaction detection with false sign rate control.
\newblock \textit{Journal of Business \& Economic Statistics}  1--12.

\bibitem[{Li et~al.(2014)Li, Zhong, Li and Wu}]{li2014fast}
\textsc{Li, J.}, \textsc{Zhong, W.}, \textsc{Li, R.} and \textsc{Wu, R.}
  (2014).
\newblock A fast algorithm for detecting gene--gene interactions in genome-wide
  association studies.
\newblock \textit{The annals of applied statistics} \textbf{8} 2292.

\bibitem[{Liu(2013)}]{liu2013gaussian}
\textsc{Liu, W.} (2013).
\newblock Gaussian graphical model estimation with false discovery rate
  control.
\newblock \textit{The Annals of Statistics}  2948--2978.

\bibitem[{Liu et~al.(2016)Liu, Cui and Li}]{liu2016partial}
\textsc{Liu, X.}, \textsc{Cui, Y.} and \textsc{Li, R.} (2016).
\newblock Partial linear varying multi-index coefficient model for integrative
  gene-environment interactions.
\newblock \textit{Statistica Sinica} \textbf{26} 1037.

\bibitem[{Lu et~al.(2011)Lu, Liu, Wen, Grubbs, Townsend, Malone, Lubet and
  You}]{lu2011cross}
\textsc{Lu, Y.}, \textsc{Liu, P.}, \textsc{Wen, W.}, \textsc{Grubbs, C.~J.},
  \textsc{Townsend, R.~R.}, \textsc{Malone, J.~P.}, \textsc{Lubet, R.~A.} and
  \textsc{You, M.} (2011).
\newblock Cross-species comparison of orthologous gene expression in human
  bladder cancer and carcinogen-induced rodent models.
\newblock \textit{American journal of translational research} \textbf{3} 8.

\bibitem[{Ma et~al.(2015)Ma, Carroll, Liang and Xu}]{shujie2015estimation}
\textsc{Ma, S.}, \textsc{Carroll, R.~J.}, \textsc{Liang, H.} and \textsc{Xu,
  S.} (2015).
\newblock Estimation and inference in generalized additive coefficient models
  for nonlinear interactions with high-dimensional covariates.
\newblock \textit{Annals of statistics} \textbf{43} 2102.

\bibitem[{Ma and Xu(2015)}]{ma2015semiparametric}
\textsc{Ma, S.} and \textsc{Xu, S.} (2015).
\newblock Semiparametric nonlinear regression for detecting gene and
  environment interactions.
\newblock \textit{Journal of Statistical Planning and Inference} \textbf{156}
  31--47.

\bibitem[{Manolio et~al.(2009)Manolio, Collins, Cox, Goldstein, Hindorff,
  Hunter, McCarthy, Ramos, Cardon, Chakravarti et~al.}]{manolio2009finding}
\textsc{Manolio, T.~A.}, \textsc{Collins, F.~S.}, \textsc{Cox, N.~J.},
  \textsc{Goldstein, D.~B.}, \textsc{Hindorff, L.~A.}, \textsc{Hunter, D.~J.},
  \textsc{McCarthy, M.~I.}, \textsc{Ramos, E.~M.}, \textsc{Cardon, L.~R.},
  \textsc{Chakravarti, A.} \textsc{et~al.} (2009).
\newblock Finding the missing heritability of complex diseases.
\newblock \textit{Nature} \textbf{461} 747--753.

\bibitem[{Murcray et~al.(2009)Murcray, Lewinger and
  Gauderman}]{murcray2009gene}
\textsc{Murcray, C.~E.}, \textsc{Lewinger, J.~P.} and \textsc{Gauderman, W.~J.}
  (2009).
\newblock Gene-environment interaction in genome-wide association studies.
\newblock \textit{American journal of epidemiology} \textbf{169} 219--226.

\bibitem[{Purcell et~al.(2007)Purcell, Neale, Todd-Brown, Thomas, Ferreira,
  Bender, Maller, Sklar, De~Bakker, Daly et~al.}]{purcell2007plink}
\textsc{Purcell, S.}, \textsc{Neale, B.}, \textsc{Todd-Brown, K.},
  \textsc{Thomas, L.}, \textsc{Ferreira, M.~A.}, \textsc{Bender, D.},
  \textsc{Maller, J.}, \textsc{Sklar, P.}, \textsc{De~Bakker, P.~I.},
  \textsc{Daly, M.~J.} \textsc{et~al.} (2007).
\newblock Plink: a tool set for whole-genome association and population-based
  linkage analyses.
\newblock \textit{The American journal of human genetics} \textbf{81} 559--575.

\bibitem[{Rafnar et~al.(2009)Rafnar, Sulem, Stacey, Geller, Gudmundsson,
  Sigurdsson, Jakobsdottir, Helgadottir, Thorlacius, Aben
  et~al.}]{rafnar2009sequence}
\textsc{Rafnar, T.}, \textsc{Sulem, P.}, \textsc{Stacey, S.~N.},
  \textsc{Geller, F.}, \textsc{Gudmundsson, J.}, \textsc{Sigurdsson, A.},
  \textsc{Jakobsdottir, M.}, \textsc{Helgadottir, H.}, \textsc{Thorlacius, S.},
  \textsc{Aben, K.~K.} \textsc{et~al.} (2009).
\newblock Sequence variants at the tert-clptm1l locus associate with many
  cancer types.
\newblock \textit{Nature genetics} \textbf{41} 221--227.

\bibitem[{Rothman et~al.(2010)Rothman, Garcia-Closas, Chatterjee, Malats, Wu,
  Figueroa, Real, Van Den~Berg, Matullo, Baris et~al.}]{rothman2010multi}
\textsc{Rothman, N.}, \textsc{Garcia-Closas, M.}, \textsc{Chatterjee, N.},
  \textsc{Malats, N.}, \textsc{Wu, X.}, \textsc{Figueroa, J.~D.}, \textsc{Real,
  F.~X.}, \textsc{Van Den~Berg, D.}, \textsc{Matullo, G.}, \textsc{Baris, D.}
  \textsc{et~al.} (2010).
\newblock A multi-stage genome-wide association study of bladder cancer
  identifies multiple susceptibility loci.
\newblock \textit{Nature genetics} \textbf{42} 978--984.

\bibitem[{Sing et~al.(2004)Sing, Steng{\aa}rd and Kardia}]{sing2004dynamic}
\textsc{Sing, C.~F.}, \textsc{Steng{\aa}rd, J.~H.} and \textsc{Kardia, S.~L.}
  (2004).
\newblock Dynamic relationships between the genome and exposures to
  environments as causes of common human diseases.
\newblock \textit{Nutrigenetics and Nutrigenomics} \textbf{93} 77--91.

\bibitem[{Tang et~al.(2020)Tang, Fang and Dong}]{tang2020high}
\textsc{Tang, C.~Y.}, \textsc{Fang, E.~X.} and \textsc{Dong, Y.} (2020).
\newblock High-dimensional interactions detection with sparse principal hessian
  matrix.
\newblock \textit{J. Mach. Learn. Res.} \textbf{21} 19--1.

\bibitem[{Tian and Feng(2021)}]{tian2021rase}
\textsc{Tian, Y.} and \textsc{Feng, Y.} (2021).
\newblock Rase: A variable screening framework via random subspace ensembles.
\newblock \textit{Journal of the American Statistical Association}  1--12.

\bibitem[{Tryka et~al.(2014)Tryka, Hao, Sturcke, Jin, Wang, Ziyabari, Lee,
  Popova, Sharopova, Kimura et~al.}]{tryka2014ncbi}
\textsc{Tryka, K.~A.}, \textsc{Hao, L.}, \textsc{Sturcke, A.}, \textsc{Jin,
  Y.}, \textsc{Wang, Z.~Y.}, \textsc{Ziyabari, L.}, \textsc{Lee, M.},
  \textsc{Popova, N.}, \textsc{Sharopova, N.}, \textsc{Kimura, M.}
  \textsc{et~al.} (2014).
\newblock Ncbi’s database of genotypes and phenotypes: dbgap.
\newblock \textit{Nucleic acids research} \textbf{42} D975--D979.

\bibitem[{van~de Geer and M{\"u}ller(2012)}]{van2012quasi}
\textsc{van~de Geer, S.} and \textsc{M{\"u}ller, P.} (2012).
\newblock Quasi-likelihood and/or robust estimation in high dimensions.
\newblock \textit{Statistical Science}  469--480.

\bibitem[{Wu et~al.(2009)Wu, Ye, Kiemeney, Sulem, Rafnar, Matullo, Seminara,
  Yoshida, Saeki, Andrew et~al.}]{wu2009genetic}
\textsc{Wu, X.}, \textsc{Ye, Y.}, \textsc{Kiemeney, L.~A.}, \textsc{Sulem, P.},
  \textsc{Rafnar, T.}, \textsc{Matullo, G.}, \textsc{Seminara, D.},
  \textsc{Yoshida, T.}, \textsc{Saeki, N.}, \textsc{Andrew, A.~S.}
  \textsc{et~al.} (2009).
\newblock Genetic variation in the prostate stem cell antigen gene psca confers
  susceptibility to urinary bladder cancer.
\newblock \textit{Nature genetics} \textbf{41} 991--995.

\bibitem[{Xia and Li(2019)}]{xia2019matrix}
\textsc{Xia, Y.} and \textsc{Li, L.} (2019).
\newblock Matrix graph hypothesis testing and application in brain connectivity
  alternation detection.
\newblock \textit{Statistica Sinica} \textbf{29} 303--328.

\bibitem[{Yan and Bien(2017)}]{yan2017hierarchical}
\textsc{Yan, X.} and \textsc{Bien, J.} (2017).
\newblock Hierarchical sparse modeling: A choice of two group lasso
  formulations.
\newblock \textit{Statistical Science} \textbf{32} 531--560.

\bibitem[{Ye et~al.(2021)Ye, Xia and Li}]{ye2021paired}
\textsc{Ye, Y.}, \textsc{Xia, Y.} and \textsc{Li, L.} (2021).
\newblock Paired test of matrix graphs and brain connectivity analysis.
\newblock \textit{Biostatistics} \textbf{22} 402--420.

\bibitem[{Zaitsev(1987)}]{zaitsev1987gaussian}
\textsc{Zaitsev, A.~Y.} (1987).
\newblock On the gaussian approximation of convolutions under multidimensional
  analogues of sn bernstein's inequality conditions.
\newblock \textit{Probability theory and related fields} \textbf{74} 535--566.

\bibitem[{Zaravinos et~al.(2011)Zaravinos, Lambrou, Boulalas, Delakas and
  Spandidos}]{zaravinos2011identification}
\textsc{Zaravinos, A.}, \textsc{Lambrou, G.~I.}, \textsc{Boulalas, I.},
  \textsc{Delakas, D.} and \textsc{Spandidos, D.~A.} (2011).
\newblock Identification of common differentially expressed genes in urinary
  bladder cancer.
\newblock \textit{PloS one} \textbf{6} e18135.

\bibitem[{Zhao and Leng(2016)}]{zhao2016analysis}
\textsc{Zhao, J.} and \textsc{Leng, C.} (2016).
\newblock An analysis of penalized interaction models.
\newblock \textit{Bernoulli} \textbf{22} 1937--1961.

\bibitem[{Zheng et~al.(2012)Zheng, Levine, Shen, Gogarten, Laurie and
  Weir}]{zheng2012high}
\textsc{Zheng, X.}, \textsc{Levine, D.}, \textsc{Shen, J.}, \textsc{Gogarten,
  S.~M.}, \textsc{Laurie, C.} and \textsc{Weir, B.~S.} (2012).
\newblock A high-performance computing toolset for relatedness and principal
  component analysis of snp data.
\newblock \textit{Bioinformatics} \textbf{28} 3326--3328.

\bibitem[{Zhou et~al.(2019)Zhou, Li, Lin and Song}]{zhou2019evaluating}
\textsc{Zhou, L.}, \textsc{Li, H.}, \textsc{Lin, H.} and \textsc{Song, P.
  X.-K.} (2019).
\newblock Evaluating functional covariate-environment interactions in the cox
  regression model.
\newblock \textit{Canadian Journal of Statistics} \textbf{47} 204--221.

\end{thebibliography}
}
\end{document}